%% file: main.tex
\documentclass[a4paper,UKenglish,cleveref,autoref, thm-restate]{lipics-v2021}

\nolinenumbers

\usepackage{xspace}
\usepackage{graphicx,tikz}
\usetikzlibrary{arrows.meta,decorations.pathreplacing,calc}
\usepackage{amsmath,amssymb}
\usepackage[capitalise]{cleveref}
\usepackage{amsthm}
\usepackage{hyperref}
\usepackage[dvipsnames]{xcolor}
\usepackage{algorithm}
\usepackage[noend]{algpseudocode}
\usepackage{float}
\usepackage{subcaption}
\usepackage{algorithm}
\usepackage{thmtools}
\usepackage[para]{footmisc}
\usepackage{bm}
\usepackage{amsbsy}
\usepackage{mathtools}
\usepackage{framed}

\newcommand\algorithmicupon{\textbf{upon}}
\algblockdefx{Upon}{EndUpon}[1]{\algorithmicupon\ #1\ \textbf{do}}{}
\algnotext{EndUpon}

\crefname{ALG@line}{line}{lines}
\crefalias{ALG@line}{line}

\makeatletter
\AddToHook{env/algorithmic/before}{\def\@currentcounter{ALG@line}}
\makeatother

\algblockdefx[EVENT]{Upon}{EndUpon}[1]
  {\textbf{upon }#1\textbf{ do}}
  {\textbf{end upon}}
\algtext*{EndUpon}

\crefname{algorithm}{algorithm}{algorithms}
\Crefname{algorithm}{Algorithm}{Algorithms}

\crefname{theorem}{theorem}{theorems}
\Crefname{theorem}{Theorem}{Theorems}

\crefname{lemma}{lemma}{lemmas}
\Crefname{lemma}{Lemma}{Lemmas}

\crefname{proposition}{proposition}{propositions}
\Crefname{proposition}{Proposition}{Propositions}

\crefname{corollary}{corollary}{corollaries}
\Crefname{corollary}{Corollary}{Corollaries}

\crefname{definition}{definition}{definitions}
\Crefname{definition}{Definition}{Definitions}

\crefname{remark}{remark}{remarks}
\Crefname{remark}{Remark}{Remarks}

\title{Seamless Reconfiguration for DAG BFT} 

\author{Michael Yiqing Hu}{National University of Singapore, Singapore}{}{}{}

\author{Leander Jehl}{University of Stavanger, Norway}{}{}{}

\author{Paul Franke-Bergmann}{Friedrich-Alexander-Universität Erlangen-Nürnberg, Germany}{}{}{}

\author{Jialin Li}{National University of Singapore, Singapore}{}{}{}

\author{Christian Berger}{Friedrich-Alexander-Universität Erlangen-Nürnberg, Germany}{}{}{}

\authorrunning{M.\,Y. Hu, L. Jehl, P. Franke-Bergmann, J. Li, and C. Berger}

\Copyright{Michael Yiqing Hu, Leander Jehl, Paul Franke-Bergmann, Jialin Li, and Christian Berger}



\category{} 

\relatedversion{} 

\begin{document}

\hypersetup{pageanchor=false}

\maketitle
\begin{abstract}

Byzantine Atomic Broadcast in Asynchronous networks has been studied extensively for decades. The FLP impossibility result rules out deterministic consensus in a fully asynchronous setting, motivating randomized protocols that combine reliable broadcast with private coins to achieve termination with probability one.

More recently, DAG-Rider popularized a new abstraction in which processes continuously reliably broadcast blocks to construct a directed acyclic graph (DAG) and then locally derive a total order using a randomized perfect coin. This separation of dissemination from ordering has renewed interest in asynchronous DAG-based Byzantine fault tolerance and inspired numerous protocols that improve expected latency or throughput.

Reconfiguration, however, remains largely unexplored. Existing DAG-based BFT protocols generally assume a static set of processes, whereas deployed systems must occasionally add or remove nodes in response to operational demands. Existing solutions are not seamless; they either pause the system or delay membership changes until a predetermined epoch boundary. This may be impractical and can introduce downtime. 
However, a seamless solution is non-trivial; agreement on when to start utilizing a new configuration may yield a logical round number that was already surpassed by the concurrent asynchronous dissemination.

In this work, we present a seamless reconfiguration protocol for asynchronous DAG-based Byzantine atomic broadcast. Our protocol allows the DAG to continue growing while the membership and corresponding quorum thresholds change, without halting the system.
To the best of our knowledge, ours is the first protocol to provide seamless reconfiguration for DAG-based Byzantine atomic broadcast in a fully asynchronous network.
%
%

\vfill

\end{abstract}

\clearpage
\setcounter{page}{1}
\hypersetup{pageanchor=true}

\input{new-sections/Introduction}
\input{new-sections/Background}
\input{new-sections/Strawman_solutions}
\input{new-sections/Problem_def}

\input{new-sections/Our_sol}

\input{new-sections/Proofs}
\input{new-sections/Conclusion}





\bibliography{biblo}

\appendix
\include{Appendix/Full-algos}

\include{Appendix/Full-proofs}

\end{document}

%% file: new-sections/Introduction.tex
\section{Introduction}
\label{sec:introduction}

The search for scalable Asynchronous Byzantine atomic broadcast (BAB) protocols has culminated in DAG-based protocols~\cite{shoalplusplus,narwhal,lemonshark,dagrider,shoal,bullshark}, which decouple transaction dissemination from consensus:
a dissemination layer continuously reliably broadcasts blocks and assembles them, in a sequence of logical rounds, with each message referencing a quorum threshold number of blocks from the previous round.
A separate consensus layer then selects a monotonic canonical list of \emph{leader} blocks from the DAG.
Ordering each block's causal history in the order of that list, yields the total order.
The consensus layer utilizes what has already been disseminated and requires no additional communication.
To overcome the FLP impossibility result~\cite{flp}, leader blocks are selected utilizing a randomized perfect coin~\cite{Boneh2001ShortSF,cachincoin,Libert2014BornAR,Shoup2000PracticalTS}, yielding expected liveness with probability one.


Existing work on DAG-based protocols largely assumes static membership, i.e. the set of processes is fixed and known to all.
Deployed systems, however, do not run on a fixed set of processes:
operators must add and remove nodes in response to operational demands and failures.
This requires more than informing processes of a change:
quorum thresholds must be adjusted, and joining and leaving processes need a procedure for taking up and relinquishing their protocol tasks.
Deployments today, therefore, either pause the protocol while the change is installed or defer it to a predetermined epoch boundary~\cite{aptos_epoch,clement2026,sui_platform}.

However, both options surrender the property that makes these protocols attractive in the first place.
Dissemination achieves throughput by never waiting for agreement, so halting it to install a membership change wastes that advantage;
a reconfiguration is usually prompted by the very failure it is meant to repair.
We would rather have reconfiguration be \emph{seamless}:
dissemination should never need to stop on its account, so correct processes keep producing and delivering blocks and advancing rounds while the membership and its quorum thresholds change, and none waits for a predetermined boundary.


The classical way to reconfigure a long-lived replicated service routes the change through the consensus layer~\cite{clement2026,Lamport2001PMS,Schneider1990}.
The new membership is ordered as an ordinary command and installed at the log position where it was decided, or a fixed number of commands later.
Such a single switching point would be useful for a DAG, since quorum thresholds change with the membership and every process must know which threshold a round demands.
In an asynchronous execution, processes may learn the decision and apply
the new configuration at different times.
Dissemination, which runs ahead of ordering by design, may meanwhile
produce blocks beyond the switching point under the old configuration.
This direct use of a single switching point therefore does not make
reconfiguration seamless: dissemination would have to wait for the
decision, or discard those blocks.

Reconfiguration without consensus avoids the need for such a point, and has been studied for distributed storage~\cite{dynastore,smartmerge,spiegelmanDisc17}, lattice agreement~\cite{reconfLatticeA}, and, most relevant here, Byzantine reliable broadcast~\cite{dynamicbrb}.
These approaches share the idea of multi-configuration operation:
an operation begun in one configuration may complete in another, and different processes may complete it in different ones.
However, reconfiguring a DAG-based BAB asks for more on top of this, since the thresholds that govern DAG construction and the randomized perfect coin that drives ordering have to be reconfigured as well; while staying consistent with the underlying reliable broadcast.

\textbf{In this paper}, we present a seamless reconfiguration procedure for DAG-based BAB.
To allow consistent reconfiguration of reliable broadcast, perfect coin and DAG construction, we extend the idea of multi-configuration operation to the DAG itself.
Every block records the configuration its broadcaster has installed, the reconfiguration it has pending, and the configurations enabled in its round.
Each correct old member enables both configurations for four rounds
after the round in which it installs the transition. Retained processes
continue producing blocks throughout the transition, and already
DAG-delivered blocks remain in the DAG.
Our procedure therefore does not require processes to agree on the exact round in which a new configuration is enabled, while ensuring that new processes take up block production when needed, and that the consensus layer cannot produce conflicting leaders from different configurations.

To the best of our knowledge, we provide the first provably correct \emph{seamless reconfiguration protocol} for DAG-based Byzantine atomic broadcast in a fully asynchronous network.
In total, we make the following contributions:
\begin{itemize}
  \item We define \emph{Dynamic DAG-structured Byzantine Broadcast}, a reconfigurable dissemination abstraction whose integrity, consistency, causal closure and per-round quorum progress hold even while correct processes sit at different stages of a transition (\Cref{sec:problem-definition}). Being independent of the ordering layer, it can carry other ordering rules.
  \item We define \emph{Dynamic DAG Byzantine Atomic Broadcast} on top, and make explicit which guarantees of the static problem must be given up under a changing membership, and why (\Cref{sec:problem-definition}).
  \item We give a seamless reconfiguration protocol implementing both, from configuration-scoped reliable broadcast, the four-round transition window, and a state transfer admitting joining processes without pausing dissemination (\Cref{sec:our-solution}).
  \item We adapt DAG-Rider's randomized leader election to a changing membership by scoping its perfect coin to a single configuration, and prove that the composition satisfies our definitions (\Cref{sec:proof-overview}).
\end{itemize}

%% file: new-sections/Background.tex
\section{Background}
\label{sec:background}

DAG-based protocols for Byzantine atomic broadcast decouple transaction dissemination from consensus~\cite{narwhal, dagrider}.
For illustration, we informally recall the seminal protocol of that line, DAG-Rider~\cite{dagrider}. 
DAG-Rider builds a DAG structured in \emph{logical rounds}, with each process contributing at most one block per round. The logical rounds are further segmented into waves. In DAG-Rider, every subsequent four rounds form a wave. 
Blocks carry a payload and reference at least a quorum of the preceding round's blocks as their \emph{parents}.
Processes share blocks using reliable broadcast~\cite{bracharbc}, this ensures non-equivocation. Since correct processes eventually complete their reliable broadcasts, there will always exist sufficient blocks to proceed. 

The DAG itself only provides a partial order to the blocks. A separate ordering layer is then utilized to obtain a total order from the DAG. 
DAG-Rider achieves this by identify certain blocks as \textit{leaders}. Leaders are identified by a \textit{perfect coin} constructed by obtaining sufficient signature shares from blocks belonging to the last round of the wave. 
This allows ensures there exist a canonical list of leaders ordered by the wave number; and every correct process will eventually derive a prefix of this canonical list. 
Each leader in the list can have it's causal history ordered in the same manner to obtain a total order of blocks that exist in the DAG. Crucially, this ordering layer does not require any additional communication; all information required has already been disseminated. 

In this work, we utilize DAG-Rider as a framework to build a seamless reconfigurable DAG Byzantine Atomic Broadcast algorithm. We make changes to the reliable broadcast, DAG construction as well as leader election components to enable seamless addition and removal of processes.

%% file: new-sections/Problem_def.tex
\section{Problem Definition and Primitives}
\label{sec:problem-definition}
\subsection{System Model}
\label{sec:model}

We consider a message-passing system whose processes are identified by
public keys. Corruptions are adaptive and permanent. A process is
\emph{uncorrupted at an event} if it has not been corrupted before that
event, and is \emph{correct} if it is never corrupted. An uncorrupted
process follows the protocol; after corruption, it becomes \emph{byzantine} and may behave arbitrarily.
The quantitative fault bound is configuration-specific and is stated below.

A process is \emph{active} while it executes at least one prescribed
protocol task. An active correct process continues those tasks and processes every delivered protocol message.

The network is asynchronous: message delays are finite but unbounded. Any message sent to a recipient is eventually delivered. 

\subsection{Configurations and membership}
In an execution, a reconfigurable system can evolve through multiple configurations. 
We write $R_l$ for a coinfiguration with identifier $l$ and $\Pi_l$ for the membership set of $R_l$. 
$R_1$ is the initial configuration of the system. 
For any configuration $R_l$ with processes $\Pi_l$ let 
\[
  n_l=|\Pi_l|\geq4,
  \qquad
  f_l=\left\lfloor\frac{n_l-1}{3}\right\rfloor,
  \qquad
  q_l=n_l-f_l.
\]
Thus
\[
  n_l\geq3f_l+1,
  \qquad
  q_l\geq2f_l+1,
  \qquad
  q_l-f_l=n_l-2f_l\geq f_l+1.
\]

%



We assume that at most $f_l$ identities in $\Pi_l$ are ever
corrupted during an execution. This bounds all identities corrupted over
the execution, not merely simultaneous corruptions. The corruption schedule
and membership policy preserve this bound for every configuration.

For our randomized perfect coin we assume that each $\Pi_l$ is represented as a duplicate-free list sorted by canonical public-key encoding, so all correct processes assign the same position to
each member. Consecutive configurations differ by the addition or removal
of at most one process. A removal is permitted only if $n_l\geq5$, so the
resulting configuration still contains at least four processes. A removed process will not be added back in any future configurations. 

\subsection{Authentication}
\label{sec:authentication}

We assume an unforgeable digital-signature scheme. 
For process $p_i$ and message $x$, $\operatorname{sign}_i(x)$ is $p_i$'s signature and $\operatorname{verify}_i(\sigma,x)$ returns $\mathsf{true}$ exactly when
$\sigma$ is a valid signature by $p_i$ on $x$. 
We also assume a public collision-resistant hash function $\mathsf{Hash}(\cdot)$.

For every configuration $R_l$, each $p_i\in\Pi_l$ receives a key share for each of two independent, domain-separated threshold-signature schemes:
\[
  (q_l,n_l)=(n_l-f_l,n_l)
  \qquad\text{and}\qquad
  (f_l+1,n_l).
\]
The first produces quorum certificates for Byzantine reliable broadcast (BRB);
the second supplies the shares used to select one member of $\Pi_l$ in each wave, as specified in
\cref{sec:configuration-scoped-coin}. 
In a $(t,n_l)$ scheme, any $t$ valid shares on the same message combine into a unique signature verifiable under that scheme's public key, while fewer than $t$ shares cannot produce a valid signature. 
The two schemes have distinct public keys and domains.

We write $\mathsf{ShareSign}^{\mathsf{BRB}}_{i,l}(x)$ and
$\mathsf{ValidShare}^{\mathsf{BRB}}_{i,l}(\sigma,x)$ for BRB share
generation and verification, and
$\mathsf{ShareSign}^{\mathsf{COIN}}_{i,l}(x)$ and
$\mathsf{ValidShare}^{\mathsf{COIN}}_{i,l}(\sigma,x)$ for their coin
counterparts. Note that each threshold-signature is scoped to a particular configuration $R_l$.

\subsection{Rounds and DAG notation}
\label{sec:round-dag-notation}
\label{sec:dag-bb-definition}

The protocol is organized into logical rounds $r=1,2,\ldots$. 
A round number is a protocol index and imposes no timing bound in the asynchronous
network. Four consecutive rounds form a wave. We write $\mathsf{first}(w)=4w-3$ and $\mathsf{last}(w)=4w$ for the first and last round in a wave $w$.


\label{sec: blocks and causal}
For process $p_i$, $DAG_i[r]$ is the round-$r$ layer of its local directed acyclic graph, and $DAG_i=\bigcup_{r\geq1}DAG_i[r]$. 
A stored round-$r$ block has the form
\[
  B=\langle m,r,\mu,C\rangle,
  \qquad
  \mu=\langle h,R,P,E,ID\rangle,
\]
where $m$ is the payload and $C$ is a quorum certificate
from configuration $R$, that certifies delivery by a
configuration-scoped Byzantine reliable broadcast (BRB) instance. 
The structure $\mu$ contains metadata. 
The symbol $\bot$ to denotes absence.
The field $h$ is $\bot$
in round $1$ and otherwise contains references to parent blocks in round
$r-1$; $R$ is the declared configuration; and $ID$ is the broadcaster identity. 
Fields $P$ and $E$ hold reconfiguration related information, detailed below in \cref{sec:configuration-authorization}.
$P$ is a pending reconfiguration message or $\bot$; $E$ is a set of enabled configurations. 
We write
\[
\begin{gathered}
  \mathsf{parents}(\mu)=h,\quad
  \mathsf{local}(\mu)=R,\quad
  \mathsf{pending}(\mu)=P,\quad
  \mathsf{enabled}(\mu)=E,\quad
  \mathsf{ID}(\mu)=ID.
\end{gathered}
\]

The pair $(r,\mathsf{ID}(\mu))$ is the block's
\emph{round-broadcaster slot}. A parent is referenced directly by $h$;
an ancestor is reached by following one or more parent references. A set
of blocks is \emph{ancestor-closed} if it contains every ancestor of each
block it contains. For $B,B'\in DAG_i$, write $B\leadsto_i B'$ if $B=B'$
or $B'$ is an ancestor of $B$. Thus $\leadsto_i$ is reflexive and points
from a block toward its causal history. We omit the subscript when the
local process is immaterial.

\subsection{Reconfiguration Model}

\label{sec:configuration-authorization}
Processes receive reconfiguration messages of the form $R_l=\langle\mathsf{RECONFIG},\Pi_l,l,\delta_l\rangle .$
Here $\mathsf{RECONFIG}$ is a fixed domain tag, $\Pi_l$ is the membership set, $l$ is the configuration identifier, and $\delta_l$ is the earliest logical round in which $R_l$ may occur in a block's pending field.
We write
\[
  \mathsf{members}(R_l)=\Pi_l,
  \qquad
  \mathsf{number}(R_l)=l,
  \qquad
  \mathsf{earliest}(R_l)=\delta_l.
\]

The process to determine the membership of a new configuration is out of scope. 
An externally authorized message may be produced by Byzantine agreement among current members or reliably broadcasted by a trusted reconfiguration dealer; 
for simplicity, we assume the latter.
The trusted dealer needs to ensure that: 
(1) reconfigurations form a sequence, using consecutive numbers as identifiers, 
(2) reconfiguration messages carry unforgeable evidence authenticating its complete tuple, 
(3) each reconfiguration adds or removes at most one process, 
(4) earliest rounds $\delta_l$ ensure configurations are \emph{well-spaced} and do not overlap, and finally 
(5) each execution contains only finitely many reconfigurations.

The passage from $R_l$ to $R_{l+1}$ is a \emph{transition};
\emph{reconfiguration} denotes the overall mechanism. A process in
$\Pi_l\cap\Pi_{l+1}$ is \emph{retained}, one in
$\Pi_{l+1}\setminus\Pi_l$ is \emph{joining}, and one in
$\Pi_l\setminus\Pi_{l+1}$ is \emph{removed}.


We write $\mathsf{ValidReconfiguration}(R_l,R_{l+1})$ when
the reconfiguration messages are correctly authenticated and follow the assumptions above. 
We call $R_l\rightarrow R_{l+1}$ a \emph{valid transition} exactly when $\mathsf{ValidReconfiguration}(R_l,R_{l+1})$ holds.
The distribution mechanism sends each valid $R_{l+1}$ to every correct process in $\Pi_l\cup\Pi_{l+1}$ to which the transition applies. 
Eventual delivery ensures that each such recipient receives it.

For every $l\geq2$, the authenticated message $R_l$ also carries the
identities in $\bigcup_{c=1}^{l-1}\Pi_c$. We omit this list from the tuple
notation. The procedure $\texttt{UpdateKnown}_i$ validates $R_l$ and this identity
list. If $R_l$ is the highest-numbered authenticated reconfiguration message
processed by $\texttt{UpdateKnown}_i$, then $\mathit{known}_i=\bigcup_{c=1}^{l}\Pi_c$. Thus $\mathit{known}_i$ contains the identities of all processes in
$R_1,\ldots,R_l$.



Point (4) above ensures a correct member of $\Pi_l$ starts the transition to $R_{l+1}$ only while installed in $R_l$. A correct joining process may authenticate the reconfiguration message $R_{l+1}$ before installing $R_l$, but invokes the transition only after it has installed $R_l$, as enforced by $\delta_{l+1}$.
In (\cref{sec: config begin}) we provide an issuance rule, that can be used by the dealer to set $\delta_l$.
Disjointness of transitions is not an assumption, but follows from the issuance rule, proved in
\cref{lemma:transition-window-disjointness}.

\paragraph*{Enabled Configurations and Transition Windows.}
A process $p_i$ stores a current configuration $\mathit{local}_i$, the potentially empty next pending authenticated reconfiguration message $\mathit{pending}_i$ and a map $\mathit{enabled\text{-}configs}_i[r]$ with a set of configurations for every round.

Whenever $\mathit{enabled\text{-}configs}_i[r]$ is assigned at a correct process $p_i$, it is assigned only once and is thereafter immutable. 
Its value has one of the following forms: $\{R_c\}$ or $\{R_l,R_{l+1}\},$ where the second form records the two configurations of an authenticated
valid transition $R_l\rightarrow R_{l+1}$. 
We call $\{R_c\}$ a \emph{singleton enabled-configs}. 
The transition protocol determines which of these forms a process assigns in each round. 
Because processes advance asynchronously, an enabled-configs value is local to a process and logical round; 
the occurrence of the pair at one correct process does not imply that every correct process has assigned the pair in that round.
When $p_i$ broadcasts a block in round $r$, it includes $\mathit{local}_i$, $\mathit{pending}_i$, and $\mathit{enabled\text{-}configs}_i[r]$ as $R,P,E$ in the blocks metadata $\mu$.

We say a round $r$ is \emph{singleton-enabled for $R_c$ at $p_i$} when $\mathit{enabled\text{-}configs}_i[r]=\{R_c\}.$
Similarly, a block $B=\langle m,r,\mu,C\rangle$ with $\mathsf{local}(\mu)=R_c$ and $\mathsf{enabled}(\mu)=\{R_c\}$ is \emph{singleton-enabled for $R_c$}.
The \emph{transition window} of $R_l\rightarrow R_{l+1}$ is the set of logical rounds $r$ for which some correct process $p_i$ satisfies $\mathit{enabled\text{-}configs}_i[r]=\{R_l,R_{l+1}\}$.

\paragraph*{Finite Reconfiguration and Stability.}
We assume that each execution contains only finitely many reconfiguration
messages, although their number is not bounded a priori, while logical
rounds remain infinite. Together with $R_1$, these messages define a finite
sequence $R_1,\ldots,R_L$ with final configuration $R_L$. We further assume
that at least $q_L$ correct members of $\Pi_L$ remain active and continue the
prescribed protocol tasks.

A wave is \emph{stable under $R_L$} if all four of its rounds are singleton-enabled for $R_L$ at every correct process in $\Pi_L$ that remains active. 
If all finitely many reconfigurations complete, there exists a round $r_*$ where all waves that follow are \emph{stable}. 


\subsection{Configuration-scoped Byzantine Reliable Broadcast}
\label{sec:configuration-brb}
We use a configuration scoped variant of Byzantine Reliable Broadcast\cite{bracharbc} (BRB) to distribute blocks. 
For round $r$ and broadcaster $p_k$, a call
$\texttt{BRB-Broadcast}_k(m,r,\mu)$ proposes
$\langle m,r,\mu\rangle$ for slot $(r,p_k)$. 
The metadata $\mu$ is the same as in a DAG block, explained in \cref{sec:round-dag-notation}.
A corresponding output at $p_i$ is $\texttt{BRB-Deliver}_i(m,r,\mu,C),$ where $C$ combines $q_l$ distinct valid BRB vote shares and verifies under the BRB public key of $R_l=\mathsf{local}(\mu)$.
The voters, threshold, and key are scoped to the declared configuration, but the slot remains $(r,p_k)$ and is not configuration-indexed.

Configuration-scoped BRB differs from standard BRB 
in that correct processes only participate after validating the metadata $\mu$, and that delivery creates a certificate. 
A correct process $p_i$ participates in a BRB instance $\langle m,r,\mu\rangle$ only if $p_i\in\Pi_l$ for $\mathsf{local}(\mu)=R_l$, and 
$R_l\in \mathit{enabled\text{-}configs}_i[r]$.
Additionally, $\mu$ must be well-formed according to the validity rules given below. 
The certificates enable processes to share valid evidence of BRB delivery, allowing an alternative for processes to add blocks into their DAG (see \cref{sec: state avail}). 
Configuration-scoped BRB provides the following safety properties.
\begin{description}
  \item [Agreement.]
  If two correct processes BRB-deliver values for the same
  round-broadcaster slot, they BRB-deliver the same value.

  \item [Integrity.]
  A correct process BRB-delivers at most one value for a
  round-broadcaster slot. Every delivered value has a valid certificate
  for its declared configuration; if the broadcaster is correct, the
  value is the one supplied to its prescribed
  $\texttt{BRB-Broadcast}$ invocation.
\end{description}
\noindent
We used a modified version of Bracha's\cite{bracharbc} BRB.
\Cref{lemma:brb-integrity,lemma:brb-safety} prove that \cref{alg:brb} satisfies these properties, including when a slot is certified during a transition.



\paragraph*{Configuration-authorized Block Validity.}
A correct process echoes or votes for a BRB instance
$\langle m,r,\mu\rangle$ only after checking the metadata $\mu$ as
follows. Because a certificate requires $q_l$ shares and at most $f_l$
members of the declared configuration are ever corrupted, Byzantine
members alone cannot certify a block that fails any of these checks.

The metadata $\mu$ is valid if the corresponding conditions hold:
The recorded field $\mathsf{enabled}(\mu)$ must be the singleton $\{\mathsf{local}(\mu)\}$ or the authenticated pair of a valid transition, and must contain $\mathsf{local}(\mu)$, and the broadcaster must belong to $\mathsf{members}(\mathsf{local}(\mu))$. 
The pending field must be $\bot$ or a valid reconfiguration message for the successor of $\mathsf{local}(\mu)$ whose earliest inclusion round has been reached:
the authorization for $R_{l+1}$ binds $\delta_{l+1}$, so a block in a round $r<\delta_{l+1}$ whose pending field contains $R_{l+1}$ is rejected.

Except in round $1$, the block must also carry exactly $q_c$ parent references, where $R_c=\mathsf{local}(\mu)$ is its declared configuration. 
The sole exception is the joining process's
\emph{transition block} (\cref{sec:our-solution}): during an addition, a block declared under $R_{l+1}$ by the member in $\Pi_{l+1}\setminus\Pi_l$ may carry $q_l$ references instead of $q_{l+1}$. 
Validators accept the smaller count only when the transition $R_l\rightarrow R_{l+1}$ is valid and both $\mathsf{enabled}(\mu)$ and their own $\mathit{enabled\text{-}configs}_i[r]$ equal $\{R_l,R_{l+1}\}$. 

Certification checks only the number of parent
references, never the referenced parents themselves; parents are validated during DAG delivery (\cref{sec:dynamic-dag-bb}). 
A valid certificate thus attests all of these checks wherever the block is later obtained.
The complete predicate appears in \cref{alg:brb-variables,alg:brb}.


\subsection{Dynamic DAG-structured Byzantine Broadcast}
\label{sec:dynamic-dag-bb}
Our protocol implements Dynamic DAG-structured Byzantine Broadcast (Dynamic DAG-BB).
Dynamic DAG-BB is a repeated, round-indexed broadcast abstraction over the finite configuration sequence $R_1,\ldots,R_L$. 
In round $r$, a process $p_i$ may invoke $\texttt{DAG-Append}_i(m,r)$. 
Consequently, a processes $p_j$ may \emph{DAG-deliver} a block $B=\langle m,r,\mu,C\rangle$ with the same message $m$ and $\mathsf{ID}(\mu)=i$, adding it to a local datastructure $DAG_j$ at level $r$ ($DAG_j[r]$).

A block contains parent and configuration metadata, as defined in \cref{sec:round-dag-notation}.
A referenced parent with metadata $\mu'$ is \emph{authorized} only if
its broadcaster belongs to a configuration in $\mathsf{enabled}(\mu')$
and its declared configuration is no newer than the child's.
Hence a certified block from a member removed by $R_{l+1}$ that records
$\{R_l,R_{l+1}\}$ may serve as a parent and count toward parent quorums,
while its broadcaster cannot broadcast, echo, or vote under $R_{l+1}$.
A block is \emph{admissible} exactly when its certificate is valid, its
metadata passes the instance-validity rules of
\cref{sec:configuration-brb}, and its referenced parents are admissible
and authorized and satisfy the applicable parent-count rule, including
the round-$1$ rule and the joining-process exception.
A \emph{parent-quorum condition} is
the protocol's round-specific requirement that a process hold enough
admissible blocks of the preceding layer to advance.

\begin{definition}[Dynamic DAG-structured Byzantine Broadcast]
\label{def:dynamic-dag-structured-broadcast}
A Dynamic DAG-structured Byzantine Broadcast satisfies the following
properties in every execution, including while correct processes are at
different stages of a transition.
\begin{enumerate}
  \item \textbf{DAG-Integrity.}
  For each round-broadcaster slot, a correct process DAG-delivers at most
  one block, and every DAG-delivered block is admissible. If the broadcaster is correct, the blocks payload $m$ originates from the append operation prescribed for that process and
  round.
  \item \textbf{DAG-Consistency.}
  If two correct processes DAG-deliver blocks $B$ and $B'$ for the same
  round-broadcaster slot, then $B=B'$.
  \item \textbf{Causal Closure.}
  Before a correct process DAG-delivers a block, it has already
  DAG-delivered every referenced parent and, recursively, every ancestor
  of those parents.
  \item \textbf{Per-Round Quorum Progress.}
    Consider a round $r$ with applicable quorum threshold $q_c$.
    If at least $q_c$ active correct members for which the protocol
    prescribes a round-$r$ append enter the round and invoke that
    operation, then every active correct process whose advancement
    depends on this quorum eventually DAG-delivers sufficient admissible
    blocks from distinct broadcaster slots to satisfy an applicable
    round-advancement condition and advance to round $r+1$.
    During $R_l\rightarrow R_{l+1}$, the applicable quorum threshold
    may be either $q_l$ or $q_{l+1}$.
  \item \textbf{Reconfiguration Completion.}
  Every valid transition eventually completes at every correct retained or
  joining process. Every correct removed process eventually reaches the
  point at which it may leave.
\end{enumerate}
\end{definition}
Per-Round Quorum Progress includes dissemination of the blocks needed for
advancement. 
For a protocol satisfying Reconfiguration Completion, the
finite reconfiguration sequence and the continued execution of the
correct processes in $\Pi_L$ that remain active imply an infinite suffix
of stable waves.

\paragraph*{BRB to DAG}
Upon an invocation of $\texttt{DAG-Append}_i(m,r)$, the process snapshots
$\mathit{local}_i$, $\mathit{pending}_i$, and
$\mathit{enabled\text{-}configs}_i[r]$. It then waits for the required
parents, constructs $\mu$ using their references and the snapshotted
values, and starts the configuration-scoped BRB instance for
$\langle m,r,\mu\rangle$. A successful instance produces the certificate $C$.

Similarly, BRB delivery alone does not insert $B=\langle m,r,\mu,C\rangle$ into the DAG. 
A correct process first validates the certificate and metadata, obtains every referenced parent and missing ancestor, and only then inserts $B$ into $DAG_i[r]$; we then say that it \emph{DAG-delivers} $B$.
A correct joining process preserves Causal Closure by
validating its ancestor-closed transferred state before inserting its
first block, and its two transition blocks (\cref{sec:our-solution})
supply the additional broadcaster slot needed for progress when a correct
process joins.

\paragraph*{Why Dynamic DAG-BB Does Not Require Validity or Totality.}
Traditional Validity requires every value broadcast by a correct process
to be delivered by every correct process; Totality requires a value
delivered by one correct process to be delivered by all. Both presume a
fixed membership: a joining process did not participate in earlier
broadcasts, and a removed process must eventually be allowed to stop.
Indeed, the strong, departure-oblivious forms of Validity and Totality
are impossible in an asynchronous dynamic system, even with reliable
links and a single possible crash
failure~\cite[Appendix~A.6]{dynamicbrb}. Dynamic DAG-BB therefore demands
only what progress and safety need: Per-Round Quorum Progress supplies
every active correct process whose advancement depends on the applicable
quorum with sufficient admissible blocks to satisfy an applicable
round-advancement condition, while DAG-Integrity, DAG-Consistency, and Causal Closure
provide safety. Universal delivery of every block is unnecessary; the
atomic-broadcast layer separately ensures agreement on the finalized
leader sequence. However, once the stable wave is reached, then we may ensure totality and validity once more.

\subsection{Dynamic DAG Byzantine Atomic Broadcast}
As in DAG-Rider~\cite{dagrider}, we define a rule to consistently select blocks from the DAG, based on a global perfect coin included in the last round of a wave.
However, our coin is scoped to configuration.
Blocks causally located between selected blocks can also be ordered deterministically as in previous work.

\subsubsection{Configuration-scoped Perfect Coin}
\label{sec:configuration-scoped-coin}
As with DAG-Rider's global perfect coin, each coin instance $(R_l,w)$ 
selects one process of $\Pi_l$. 
Unlike there, every
instance is scoped to one configuration: instances of different
configurations use separate keys, and their shares are never combined.
In the threshold implementation used by the protocol, a member invokes an
instance by releasing its valid share for that instance. For security
parameter $\lambda$, let $\operatorname{negl}(\lambda)$ denote a
negligible function. The coin satisfies the following four properties.
\begin{description}
  \item[Agreement.]
  If two correct processes obtain an output for $(R_l,w)$, they obtain the
  same process.
  \item[Termination.]
  If at least $f_l+1$ correct processes invoke the coin for $(R_l,w)$, every
  invocation by a correct process eventually returns.
  \item[Unpredictability.]
  As long as fewer than $f_l+1$ members have invoked the coin for
  $(R_l,w)$, its output, conditioned on the adversary's current view, is
  computationally indistinguishable from a uniformly random member of
  $\Pi_l$, except with probability $\operatorname{negl}(\lambda)$.
  \item[Fairness.]
  Every process in $\Pi_l$ is selected with probability $1/n_l$.
\end{description}
The threshold-share implementation and the rule that maps the selected
member to a first-round block are described in \cref{sec:our-solution}.
Coin Unpredictability is used only to prove Probabilistic progress; safety uses only Agreement and signature uniqueness.
For each correct process's block $B=\langle m,r,\mu,C\rangle,$ where $r$ is the last round of a wave, $m$ includes a $\mathsf{local}(\mu)$-scoped coin share.

\subsubsection{Atomic-broadcast Abstraction}
To realize an atomic-broadcast abstraction, we define a rule to consistently select blocks from the DAG.
These blocks are called leaders.
For an active correct process $p_i$, $\mathit{finalized\text{-}values}_i$
is its local sequence of DAG blocks selected by the ordering layer; these
blocks are its \emph{finalized leaders}. A block is \emph{finalized} at
$p_i$ when it is appended to this sequence, and
$\mathit{finalized\text{-}values}_i[k]$ denotes the $k$-th block in the
sequence.

DAG-Rider states Byzantine Atomic Broadcast using reliable-broadcast
integrity, agreement, and validity together with total
order~\cite{dagrider}. A changing membership forces two relaxations. Our
ordering primitive does not claim that every proposed block is eventually
finalized, so we omit proposal validity. A removed process must be
allowed to leave, so eventual agreement is required only of the correct
members of the final configuration that remain active. The dynamic DAG Byzantine Atomic Broadcast is the goal of this work. 

\begin{definition}[Dynamic DAG Byzantine Atomic Broadcast]
\label{def:dynamic-dag-bab}
A leader-level Dynamic DAG Byzantine Atomic Broadcast consists of an
underlying Dynamic DAG-structured Byzantine Broadcast satisfying
\cref{def:dynamic-dag-structured-broadcast} and, at each process, a sequence
of finalized leader blocks satisfying the following properties.
\begin{enumerate}
  \item \textbf{Integrity.}
  For every active correct process $p_i$ and wave $w$,
  $\mathit{finalized\text{-}values}_i$ contains at
  most one block from $w$, and no block occurs more than once. Every
  finalized block has already been inserted into $DAG_i$.
  \item \textbf{Total Order.}
  At every point in an execution, the finite sequences
  $\mathit{finalized\text{-}values}_i$ of active correct processes are
  pairwise comparable by prefix. Consequently, if one active correct
  process finalizes $B$ before $B'$, no active correct process finalizes
  $B'$ without first finalizing $B$.
  \item \textbf{Probabilistic Progress.}
  Let $R_L$ be the final configuration of the execution. With
  probability one, every correct process $p_j\in\Pi_L$ that remains
  active finalizes infinitely many blocks.
  
\end{enumerate}
\end{definition}
After a correct process is removed and leaves, no requirement is imposed
on further extensions of its $\mathit{finalized\text{-}values}$ sequence.

Our rule, specified in \cref{sec: leader elect} selects at most one block per wave.
\Cref{lemma:dynamic-dag-structured-broadcast} proves that the broadcast,
DAG, state-transfer, and reconfiguration algorithms implement Dynamic
DAG-BB. Leader Integrity, Total Order, and Probabilistic Progress
are established in \cref{corollary:leader-integrity,%
theorem:leader-total-order,theorem:leader-probabilistic-progress}, and
\cref{theorem:dynamic-dag-bab-correctness} composes them into Dynamic DAG
Byzantine Atomic Broadcast. A proof sketch appears in
\cref{sec:proof-overview}.

%% file: new-sections/Our_sol.tex
\section{Seamless DAG Reconfiguration}
\label{sec:our-solution}
Our protocol extends ordinary round-based DAG production with a four-round
configuration transition. During a transition, blocks declared under the
installed configuration may carry the pending reconfiguration, and every
block records the configurations enabled in its round. This lets retained
processes continue producing blocks while a joining process connects the
existing causal history to blocks declared under the next configuration. The
ordering layer counts coin shares only in singleton-enabled rounds, so a
transition preserves Causal Closure and Per-Round Quorum Progress without
combining shares from different configurations. We describe the protocol in
words below; its complete
validation predicates, message handlers, and retry loops appear in
\cref{sec:full-algorithms}.

\subsection{Ordinary DAG Operation}

At the beginning of round $r$, process $p_i$ uses its installed configuration $\mathit{local}_i$, pending reconfiguration $\mathit{pending}_i$, and $\mathit{enabled\text{-}configs}_i[r]$, which is immutable once assigned, to construct a block for $DAG_i[r]$.




Suppose $\mathit{local}_i=R_l$ when $p_i$ enters round $r>1$. At the start
of its append invocation, it snapshots $\mathit{local}_i$,
$\mathit{pending}_i$, and $\mathit{enabled\text{-}configs}_i[r]$. It then
waits until the preceding DAG layer contains the admissible blocks required
by the applicable parent-quorum condition, selects their hashes and
broadcaster identities as the parent set $h$, and constructs $ \mu=\left\langle h,R_l,P,E,p_i\right\rangle,$
where $R_l$, $P$, and $E$ are the snapshotted values. A reconfiguration
received while the process waits for the applicable parent-quorum
condition therefore cannot change this block's pending field or enabled
set. If $P\neq\bot$ but
$r<\mathsf{earliest}(P)$, the process does not broadcast the block; correct
BRB echoers and voters enforce the same check.

It then invokes the $R_l$-scoped BRB of
\cref{sec:configuration-brb} on $\langle m,r,\mu\rangle$ by calling $\texttt{DAG-Append}_i(m,r)$.
A successful BRB instance produces a certificate $C$, but this certificate
alone does not make the certified object a DAG block. A recipient validates
the block according to the admissibility rules of
\cref{sec:problem-definition}, recursively obtains every referenced parent
and missing ancestor, and only then inserts
$B=\langle m,r,\mu,C\rangle$ into its local DAG and DAG-delivers it.
Consequently, every block DAG-delivered by a correct process is admissible
and is accompanied by its complete causal history.


A process enters round $r+1$ once $DAG_i[r]$ contains enough admissible
round-$r$ blocks to satisfy the applicable parent-quorum condition.
Thus, a correct process may advance even if it has not yet delivered the block itself proposed.
In round $r$, a process installed in $R_l$ ordinarily waits for $q_l$
admissible round-$r$ blocks before entering round $r+1$. If $r=r_0+3$
is its last round before completing a transition installed in round $r_0$,
it instead waits for $q_{l+1}$ admissible round-$r$ blocks that are eligible
as parents of the next block declared under $R_{l+1}$.
\paragraph*{State Availability}
\label{sec: state avail}
The procedures mentioned here further enable processes to include blocks in their DAG indirectly.
Three actions leads to a process $p_i$ sending a block $B\in DAG_i$ to every identity in $\mathit{known}_i$ without first receiving a request. 
(1) The BRB vote-receive handler sends one when it combines $q_c$ valid
shares into a certificate, (2) when a process adds a $B$ into its local DAG, or (3) when $\texttt{FloodDAG}$ is called and one of every block in the local DAG is broadcasted to $known_i$. 
When a new configuration message $R$ is received, and $\texttt{UpdateKnown}_i$ updates $\mathit{known}_i$; $\texttt{FloodDAG}$ will be invoked. 
Once a process receives a block in these manners, it is buffered until it becomes admissible.



\subsection{Beginning and Completing a Configuration Change}
\label{sec: config begin}


In this section, we consider one transition from $R_l$ to $R_{l+1}$.
By well-spacing (\cref{sec:configuration-authorization}), different transitions happen in non-overlapping logical rounds. 
If a correct process records $\{R_l,R_{l+1}\}$ in one
round, no other correct process records the pair for a different transition in
that round.
To ensure well-spacing, a dealer can "listen" to the DAG. 
In practice, the dealer may repeatedly request complete,
ancestor-closed DAGs from members of $\Pi_l$ and validate every reply.
Eventual delivery ensures that it receives a complete DAG from a correct
member.
It waits until the downloaded state contains $q_l$ certified blocks
declared under $R_l$ in some round $s$, and then issues
$R_{l+1}$ with $\delta_{l+1}=s+2$. 
We show in \cref{lemma:no-premature-pending,lemma:well-spaced-predecessor-completion,%
lemma:well-spaced-base-round,lemma:predecessor-parent-slots,%
lemma:earliest-install-lower-bound} that this leads to \emph{well-spacing}.

When a member of $\Pi_l$ receives $R_{l+1}$, it authenticates and buffers
the message. Once it has installed $R_l$ and reached the authenticated
earliest inclusion round, it records $R_{l+1}$ in
$\mathit{pending}_i$; only append operations started after that point copy
the message into block metadata. A block is \emph{marked by $R_{l+1}$} if
its own pending field contains $R_{l+1}$ or an ancestor is marked by
$R_{l+1}$. Thus the mark propagates through the DAG's causal history.

When a process with $\mathit{local}_i=R_l$ completes a round $r$ containing
at least $f_l+1$ marked blocks from distinct broadcaster slots in $\Pi_l$,
it invokes $\texttt{InstallTransition}_i(R_l,R_{l+1},r)$. We call $r$ its
\emph{install transition round}. The process enables both configurations
in the next four rounds:
\[
  \mathit{enabled\text{-}configs}_i[s]=\{R_l,R_{l+1}\},
  \qquad s\in\{r+1,r+2,r+3,r+4\}.
\]
It remains installed in $R_l$ through round $r+3$ and installs $R_{l+1}$
at the beginning of round $r+4$, before ordinary block production. Correct
members of $\Pi_l$ need not choose the same install transition round, but
\cref{corollary:install-transition-round-bound} shows that their choices
differ by at most one round.

\begin{figure}[t]
  \centering
  \makebox[\linewidth][c]{%
    \resizebox{0.8\linewidth}{!}{%
      \input{figures/dag-example2}%
    }%
  }
  \caption{Adding $p_5$ when
  $\Pi_l=\{p_1,p_2,p_3,p_4\}$ and
  $\Pi_{l+1}=\Pi_l\cup\{p_5\}$, with Byzantine $p_4$. For each block, $\{\cdot\}$ denotes its $\mathit{enabled\text{-}configs}$ in its $\mu$.
  $p_5$ begins $\texttt{DAG-Transition-Append}$ in round $r+5$ before starting normal
  operations in round $r+7$. Blocks from $p_5$ in rounds $r+5,r+6$ are \emph{transition blocks}; having $q_l$ parents instead of $q_{l+1}$.}
  \label{fig:dag-example}
\end{figure}
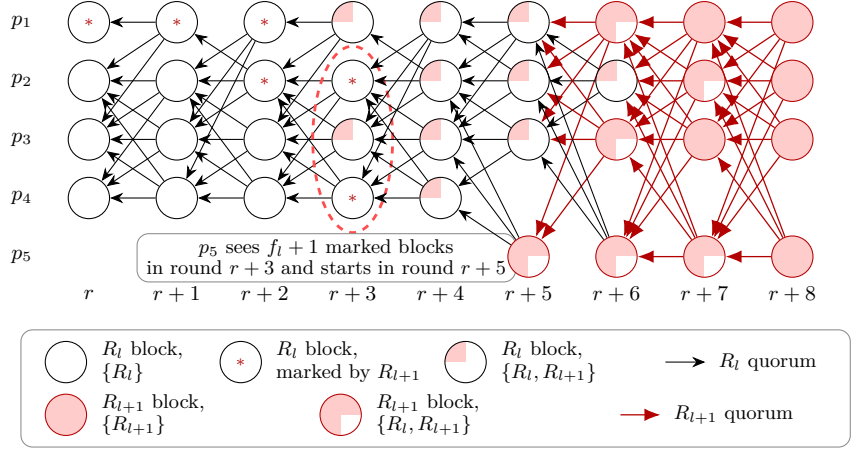

\paragraph*{Joining Procedure}

Let $p_j\in\Pi_{l+1}\setminus\Pi_l$. After receiving the authenticated pair
$\langle R_l,R_{l+1}\rangle$, the joining process repeatedly requests
admissible blocks from every identity in $\mathit{known}_j$ and sends another
request after accepting a reply or when $\mathit{known}_j$ changes. 
An active correct process that receives a request
satisfying the state-request handler's conditions replies with a finite, ancestor-closed set of blocks from its local DAG. The joining process accepts
the set only if every block has valid certified metadata, and every referenced parent and its ancestors are included in the~set.

The joining process also inserts blocks it accepts into its local DAG in
increasing round order. Before inserting a block, it validates the block and
ensures that every referenced parent and ancestor has already been inserted.
It applies the same rule to blocks received individually, requesting any missing
parents and ancestors from the identities in $\mathit{known}_j$. The joining
process continues requesting blocks until its local DAG contains a round with
at least $q_l$ admissible blocks in distinct broadcaster slots belonging to
$\Pi_l$, including at least $f_l+1$ blocks marked by $R_{l+1}$.


Per \cref{sec: state avail}, processes will occasionally send their blocks as a broadcast to other processes. 
Messages sent to a joining process before it
activates remain pending and are processed after activation. Consequently,
obtaining the required causal history does not depend on any particular
process remaining active until the joining process later requests it.

As soon as the preceding condition holds, the joining process chooses the
smallest qualifying round currently present in its local DAG and denotes it
by $u$. It uses $u$ as its install transition round and continues obtaining
blocks until it has $q_l$ admissible round-$(u+1)$ blocks in distinct
broadcaster slots belonging to $\Pi_l$. It continues receiving and validating
additional blocks throughout the rest of the joining procedure. In round
$u+1$, it remains passive: it neither produces a block nor sends a BRB echo
or vote share.





In rounds $u+2$ and $u+3$, the joining process invokes
$\texttt{DAG-Transition-Append}$ and produces a transition block
declared under $R_{l+1}$. The round-$(u+2)$ block uses $q_l$ parents from
the downloaded round-$(u+1)$ layer. The round-$(u+3)$ block also uses
$q_l$ admissible parents and may reference transition blocks in the
preceding layer.

By \cref{lemma:joining-transition-round-bound}, every correct member of
$\Pi_l$ enabled-configs are $\{R_l,R_{l+1}\}$ in rounds $u+2$ and $u+3$, thus transition blocks obtain $R_{l+1}$-scoped BRB certificates. These blocks contribute the additional broadcaster slot that may be
needed by processes with an early install transition round; when $q_{l+1}>q_l$. 
 At the beginning of round $u+4$, the joining
process installs $R_{l+1}$ and starts its ordinary protocol tasks.
Processes accept transition blocks only during the rounds, when both configurations in the transition are enabled.
Hence, a Byzantine joining process cannot reuse the exception in a later singleton-enabled round.






\paragraph*{Leaving Procedure}

\begin{figure}[t]
  \centering
  \makebox[\linewidth][c]{%
    \resizebox{0.8\linewidth}{!}{%
      \input{figures/dag_example_leave}%
    }%
  }
  \caption{Removing $p_5$ when
  $\Pi_l=\{p_1,p_2,p_3,p_4,p_5\}$ and
  $\Pi_{l+1}=\Pi_l\setminus\{p_5\}$, with Byzantine $p_4$. $p_5$ completes \texttt{InstallTransition} at round r+6 but does not leave (it continues
helping complete reliable broadcasts and responding to requests for
DAG blocks) until sufficient blocks exist in round r+7.}
  \label{fig:dag-example-leave}
\end{figure}
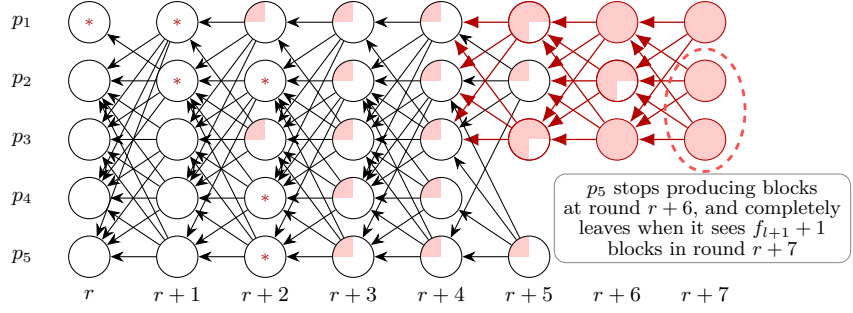

A leaving process $p_i\in\Pi_l\setminus\Pi_{l+1}$ installs the transition as other nodes in $\Pi_l$. 
However, when installing the transition in round $r$, $p_i$ does not produce blocks in round $r+4$ or later. 
Process $p_i$ continues helping complete reliable broadcasts and responding to requests for
DAG blocks; while its departure condition remains
pending, $p_i$ continues processing BRB messages and receiving, validating,
and inserting blocks and their missing ancestors.
Once it has inserted $f_{l+1}+1$ admissible
round-$(r+5)$ blocks from distinct broadcasters in $\Pi_{l+1}$, it invokes one final $\texttt{FloodDAG}$ before stopping its remaining tasks; finally proceeding to leave. 
This enables a lagging process to eventually obtain the required blocks to progress. 




Let $r_0$ be process $p_j$'s install transition round for
$R_l\rightarrow R_{l+1}$. If $p_j$ is in round
$r_j\in\{r_0+2,r_0+3\}$, it may also advance when its local DAG
contains at least $q_{l+1}$ admissible round-$(r_j+1)$ blocks eligible
as parents of a block declared under $R_{l+1}$.
These blocks need not all be declared under $R_{l+1}$ themselves.
The process can therefore advance even if its own round-$r_j$ block
has not yet been certified. The same condition is available during
additions and removals.
This behavior is formalized in
\cref{lemma:transition-next-layer-advancement,%
lemma:stopping-process-transition-state,%
lemma:transition-advancement-alternatives}.

\subsection{Configuration-scoped Leader Election}
\label{sec: leader elect}

After explaining how block production crosses a configuration change, we turn to ordering the resulting DAG. 
We explain how a sequence of leaders (anchor blocks) is created.
Blocks located between these leaders can be ordered deterministically, as done in previous work~\cite{dagrider}.
As in DAG-Rider, four consecutive rounds form a wave: first-round blocks are leader candidates, while last-round blocks provide causal support and carry the shares used to reconstruct the wave's threshold coin. 
Because a wave may intersect a configuration change, all coin shares and support thresholds are scoped to one configuration.

A last-round block $B=\langle m,r,\mu,C\rangle$ is an
$R_l$-typed \emph{leader-voting block} for wave $w$ when
\[
  r=\mathsf{last}(w),\qquad
  \mathsf{local}(\mu)=R_l,\qquad
  \mathsf{enabled}(\mu)=\{R_l\},
\]
and $B$ carries a valid $R_l$-scoped coin share for $w$. 
This coin share is distinct from a BRB vote share. 
A block with $\mathsf{enabled}(\mu)=\{R_l,R_{l+1}\}$ does not count as a
leader-voting block, but may be a leader candidate.
Further, shares from different configurations (e.g. $R_l$ and $R_{l+1}$-typed) are never combined.

Every broadcaster that is uncorrupted when producing a leader-voting block
fixes the block's parent set before generating and releasing its coin share.
For stable waves under $R_L$, this ordering lets the proof fix the common
core before $f_L+1$ distinct valid shares have been released, as shown in
\cref{lemma:stable-wave-reachability}.
From $f_l+1$ distinct valid shares, $\texttt{WaveLeader}(R_l,w)$ reconstructs the unique $R_l$-scoped coin. 
If the coin selects $p_j\in\Pi_l$, $\texttt{WaveLeader}(R_l,w)$ returns
$p_j$'s round-$\mathsf{first}(w)$ block declared under $R_l$, if that block
is in the local DAG; otherwise it returns $\bot$.

A process may receive certified leader-voting blocks declared under
different configurations in the same last round. 
However, as shown in \cref{lemma:same-round-enabled-configs}, all correct broadcasters whose blocks are singleton-enabled in that round use the same configuration. 
For every other configuration $R_c$, only Byzantine broadcasters can produce leader-voting blocks of that type, yielding at most $f_c$ distinct slots.
This is fewer than the $f_c+1$ shares needed to reconstruct the
$R_c$-scoped coin. Consequently, at most one configuration can have at
least $f_c+1$ leader-voting blocks in a wave, even when a reconfiguration is ongoing.

For configuration $R_l$ and wave $w$, let $B=\texttt{WaveLeader}(R_l,w)\neq\bot$ be the first-round block selected by the reconstructed coin. 
Process $p_i$ identifies $B$ as a \emph{direct anchor} when at least $q_l$ distinct $R_l$-typed leader-voting blocks $X$ in round $\mathsf{last}(w)$ satisfy $X\leadsto_i B.$
Here $X\leadsto_i B$ means that $B$ is in $X$'s causal history; 
the relation follows parent references toward earlier blocks. 
The process checks this condition on completing the wave and whenever it later DAG-delivers another leader-voting block for that wave. 
Rechecking matters because correct processes may receive the same last-round blocks in different orders.

Once a process identifies a direct anchor $B^\star$ in wave $k$, it scans
the ancestors of $B^\star$ for anchors from earlier waves that it has not
previously processed, in reverse wave order.
When it examines a wave $w<k$, it accepts that wave's candidate $B'$ only
when, for some configuration $R_l$, at least $f_l+1$ distinct $R_l$-typed
leader-voting blocks $X$ for wave $w$ satisfy
$B^\star\leadsto_i X\leadsto_i B'.$
Recall that $\leadsto_i$ points toward causal history:
the later anchor $B^\star$ contains $X$ in its ancestry, and $X$ in turn
contains $B'$. Only the voting blocks of $B'$'s own wave are counted,
although voting blocks of later waves may also lie between $B^\star$ and
$B'$.
An accepted $B'$ becomes the scan anchor for the next iteration and is
prepended to the result; a wave with no acceptable candidate is skipped.


The full backward scan rooted at a direct anchor determines its \emph{canonical leader prefix}: a prefix of the total order of leader blocks that all active correct processes will eventually derive (\cref{lemma:deterministic-leader-scans}). 
The serialized local procedure $\texttt{Elect}$ performs this scan and atomically appends the missing suffix to
$\mathit{finalized\text{-}values}_i$. 
When one process identifies a direct anchor, delivering $q_l$ many $R_l$-typed leader-voting blocks in one round, the causal history of any direct anchor in a later wave will include $f_l+1$ of these blocks. 
Thus other processes will accept the same anchor in a later backward scan.
The detailed helper definitions and election procedure appear in \cref{alg:leader-election-helpers,alg:leader-election_2}.

Together, transition blocks preserve Per-Round Quorum Progress across transitions, while requiring the leader-voting blocks for each decision to declare one configuration prevents a decision from combining configurations.
Transitions may prevent a leader from being elected in some waves. 
To ensure the liveness of leader election, we rely on the assumption that reconfigurations eventually cease.

%% file: figures/dag-example2.tex
\providecommand{\dagXScale}{1.32cm}%
\providecommand{\dagLegendXScale}{1.50cm}%
\providecommand{\dagYScale}{0.88cm}%
\providecommand{\dagEllipseYRadius}{1.40cm}%
\providecommand{\dagColumnLabelY}{-0.65}%
\providecommand{\dagLegendRowOne}{-1.80}%
\providecommand{\dagLegendRowTwo}{-2.70}%
\providecommand{\dagArrowBlackY}{-1.80}%
\providecommand{\dagArrowRedY}{-2.70}%
\providecommand{\dagAnnotationFont}{\footnotesize}%

\colorlet{bDraw}{red!70!black}%
\colorlet{bFill}{red!20}%
\colorlet{bMark}{red!70!black}%

\newcommand{\drawtoPrev}[4]{%
    \pgfmathtruncatemacro{\prevcol}{#1-1}%
    \ifnum#1>1
        \ifx&#3&\else
            \foreach \targetRow in {#3} {
                \ifcsname pgf@sh@ns@b-\targetRow-\prevcol\endcsname
                    \draw[#4] (b-#2-#1) -- (b-\targetRow-\prevcol);
                \fi
            }
        \fi
    \fi
}%

\newcommand{\placeA}[2]{\node[blockA] (b-#2-#1) at ({#1},{5-#2}) {};}%
\newcommand{\placeMarked}[3]{\node[blockA] (b-#2-#1) at ({#1},{5-#2}) {\textcolor{#3}{\scriptsize $\ast$}};}%
\newcommand{\placeAB}[2]{\node[blockAB] (b-#2-#1) at ({#1},{5-#2}) {};}%
\newcommand{\placeBA}[2]{\node[blockBA] (b-#2-#1) at ({#1},{5-#2}) {};}%
\newcommand{\placeB}[2]{\node[blockB] (b-#2-#1) at ({#1},{5-#2}) {};}%

\newcommand{\blockA}[3]{\placeA{#1}{#2}\drawtoPrev{#1}{#2}{#3}{arrowA}}%
\newcommand{\blockMarked}[4]{\placeMarked{#1}{#2}{#4}\drawtoPrev{#1}{#2}{#3}{arrowA}}%
\newcommand{\blockAB}[3]{\placeAB{#1}{#2}\drawtoPrev{#1}{#2}{#3}{arrowA}}%
\newcommand{\blockBA}[4]{\placeBA{#1}{#2}\drawtoPrev{#1}{#2}{#3}{#4}}%
\newcommand{\blockB}[4]{\placeB{#1}{#2}\drawtoPrev{#1}{#2}{#3}{#4}}%

\begin{tikzpicture}[x=\dagXScale,y=\dagYScale,every node/.style={font=\small}]
    \tikzset{
        blockBase/.style={draw, circle, minimum size=0.62cm, inner sep=0pt, line width=0.45pt},
        blockA/.style={blockBase, draw=black, fill=white},
        blockAB/.style={
            blockBase,
            draw=black,
            fill=white,
            path picture={\fill[bFill] (path picture bounding box.center) rectangle (path picture bounding box.north west);}
        },
        blockBA/.style={
            blockBase,
            draw=bDraw,
            fill=bFill,
            path picture={\fill[white] (path picture bounding box.center) rectangle (path picture bounding box.south east);}
        },
        blockB/.style={blockBase, draw=bDraw, fill=bFill},
        legendText/.style={anchor=west, align=left, font=\footnotesize, inner sep=0pt},
        groupBox/.style={
            draw=black!45,
            fill=white,
            rounded corners=5pt,
            line width=0.45pt
        },
        arrowA/.style={draw=black, line width=0.45pt, -{Stealth[length=2.0mm,width=1.6mm]}},
        arrowB/.style={draw=bDraw, line width=0.55pt, -{Latex[length=2.4mm,width=1.9mm]}}
    }

    \blockMarked{1}{1}{}{bMark} 
    \blockA{1}{2}{} 
    \blockA{1}{3}{} 
    \blockA{1}{4}{}

    \blockMarked{2}{1}{1,2,3}{bMark}
    \blockA{2}{2}{2,3,4}
    \blockA{2}{3}{2,3,4}{}
    \blockA{2}{4}{2,3,4}

    \blockMarked{3}{1}{2,3,4}{bMark}
    \blockMarked{3}{2}{1,2,3}{bMark}
    \blockA{3}{3}{2,3,4}
    \blockA{3}{4}{2,3,4}{}

    \blockAB{4}{1}{1,2,3}
    \blockMarked{4}{2}{2,3,4}{bMark}
    \blockAB{4}{3}{2,3,4}
    \blockMarked{4}{4}{2,3,4}{bMark}

    \draw[very thick, dashed, rounded corners=12pt, draw=red!70] 
        (4.0, 2) ellipse [x radius=0.6cm, y radius=\dagEllipseYRadius];
    \node[
        groupBox,
        anchor=north,
        align=center,
        font=\dagAnnotationFont\linespread{0.94}\selectfont,
        inner xsep=4pt,
        inner ysep=2.5pt
    ] at (3.70, 0.40) {$p_5$ sees $f_l+1$ marked blocks\\ in round $r+3$ and starts in round $r+5$};
    
    \blockAB{5}{1}{2,3,4}
    \blockAB{5}{2}{1,2,3}
    \blockAB{5}{3}{2,3,4}
    \blockAB{5}{4}{2,3,4}

    \blockAB{6}{1}{1,2,3}
    \blockAB{6}{2}{1,2,3}
    \blockAB{6}{3}{2,3,4}
    \blockBA{6}{5}{2,3,4}{arrowA}

    \blockBA{7}{1}{1,2,3,5}{arrowB}
    \blockAB{7}{2}{1,2,3}
    \blockBA{7}{3}{1,2,3,5}{arrowB}
    \blockBA{7}{5}{1,2,3}{arrowA}

    \blockB{8}{1}{1,2,3,5}{arrowB}
    \blockBA{8}{2}{1,2,3,5}{arrowB}
    \blockB{8}{3}{1,2,3,5}{arrowB}
    \blockBA{8}{5}{1,2,3,5}{arrowB}

    \blockB{9}{1}{1,2,3,5}{arrowB}
    \blockB{9}{2}{1,2,3,5}{arrowB}
    \blockB{9}{3}{1,2,3,5}{arrowB}
    \blockB{9}{5}{1,2,3,5}{arrowB}

    \foreach \c in {1,...,9} {
        \pgfmathtruncatemacro{\roundOffset}{\c-1}
        \node[font=\small] at ({\c}, {\dagColumnLabelY})
            {\ifnum\roundOffset=0 $r$\else $r+\roundOffset$\fi};
    }

    \foreach \r in {1,...,5} {
        \node[anchor=east] at ({0.45}, {5-\r}) {$p_{\r}$};
    }

    \begin{scope}[x=\dagLegendXScale]
    \path[groupBox]
        (0.20,{\dagLegendRowTwo-0.55}) rectangle
        (8.50,{\dagLegendRowOne+0.55});

    \node[blockA] at (0.65,\dagLegendRowOne) {};
    \node[legendText, text width=1.55cm] at (1.00,\dagLegendRowOne) {$R_l$ block,\\$\{R_l\}$};

    \node[blockA] at (2.40,\dagLegendRowOne) {\textcolor{bMark}{\scriptsize $\ast$}};
    \node[legendText, text width=2.35cm] at (2.75,\dagLegendRowOne) {$R_l$ block,\\marked by $R_{l+1}$};

    \node[blockAB] at (4.65,\dagLegendRowOne) {};
    \node[legendText, text width=2.25cm] at (5.00,\dagLegendRowOne) {$R_l$ block,\\$\{R_l,R_{l+1}\}$};

    \draw[arrowA] (6.65,\dagArrowBlackY) -- (7.05,\dagArrowBlackY);
    \node[legendText, text width=1.65cm] at (7.18,\dagArrowBlackY) {$R_l$ quorum};

    \node[blockB] at (0.65,\dagLegendRowTwo) {};
    \node[legendText, text width=1.80cm] at (1.00,\dagLegendRowTwo) {$R_{l+1}$ block,\\$\{R_{l+1}\}$};

    \node[blockBA] at (3.40,\dagLegendRowTwo) {};
    \node[legendText, text width=2.25cm] at (3.75,\dagLegendRowTwo) {$R_{l+1}$ block,\\$\{R_l,R_{l+1}\}$};

    \draw[arrowB] (6.20,\dagArrowRedY) -- (6.60,\dagArrowRedY);
    \node[legendText, text width=1.95cm] at (6.73,\dagArrowRedY) {$R_{l+1}$ quorum};
    \end{scope}
\end{tikzpicture}%

%% file: figures/dag_example_leave.tex
\providecommand{\dagXScale}{1.32cm}%
\providecommand{\dagYScale}{0.88cm}%
\providecommand{\dagEllipseYRadius}{1.40cm}%
\providecommand{\dagColumnLabelY}{-0.65}%
\providecommand{\dagAnnotationFont}{\footnotesize}%

\colorlet{bDraw}{red!70!black}%
\colorlet{bFill}{red!20}%
\colorlet{bMark}{red!70!black}%

\newcommand{\drawtoPrev}[4]{%
    \pgfmathtruncatemacro{\prevcol}{#1-1}%
    \ifnum#1>1
        \ifx&#3&\else
            \foreach \targetRow in {#3} {
                \ifcsname pgf@sh@ns@b-\targetRow-\prevcol\endcsname
                    \draw[#4] (b-#2-#1) -- (b-\targetRow-\prevcol);
                \fi
            }
        \fi
    \fi
}%

\newcommand{\placeA}[2]{\node[blockA] (b-#2-#1) at ({#1},{5-#2}) {};}%
\newcommand{\placeMarked}[3]{\node[blockA] (b-#2-#1) at ({#1},{5-#2}) {\textcolor{#3}{\scriptsize $\ast$}};}%
\newcommand{\placeAB}[2]{\node[blockAB] (b-#2-#1) at ({#1},{5-#2}) {};}%
\newcommand{\placeBA}[2]{\node[blockBA] (b-#2-#1) at ({#1},{5-#2}) {};}%
\newcommand{\placeB}[2]{\node[blockB] (b-#2-#1) at ({#1},{5-#2}) {};}%

\newcommand{\blockA}[3]{\placeA{#1}{#2}\drawtoPrev{#1}{#2}{#3}{arrowA}}%
\newcommand{\blockMarked}[4]{\placeMarked{#1}{#2}{#4}\drawtoPrev{#1}{#2}{#3}{arrowA}}%
\newcommand{\blockAB}[3]{\placeAB{#1}{#2}\drawtoPrev{#1}{#2}{#3}{arrowA}}%
\newcommand{\blockBA}[4]{\placeBA{#1}{#2}\drawtoPrev{#1}{#2}{#3}{#4}}%
\newcommand{\blockB}[4]{\placeB{#1}{#2}\drawtoPrev{#1}{#2}{#3}{#4}}%

\begin{tikzpicture}[x=\dagXScale,y=\dagYScale,every node/.style={font=\small}]
    \tikzset{
        blockBase/.style={draw, circle, minimum size=0.62cm, inner sep=0pt, line width=0.45pt},
        blockA/.style={blockBase, draw=black, fill=white},
        blockAB/.style={
            blockBase,
            draw=black,
            fill=white,
            path picture={\fill[bFill] (path picture bounding box.center) rectangle (path picture bounding box.north west);}
        },
        blockBA/.style={
            blockBase,
            draw=bDraw,
            fill=bFill,
            path picture={\fill[white] (path picture bounding box.center) rectangle (path picture bounding box.south east);},
            postaction={draw=bDraw, line width=0.45pt}
        },
        blockB/.style={blockBase, draw=bDraw, fill=bFill},
        groupBox/.style={
            draw=black!45,
            fill=white,
            rounded corners=5pt,
            line width=0.45pt
        },
        arrowA/.style={draw=black, line width=0.45pt, -{Stealth[length=2.0mm,width=1.6mm]}},
        arrowB/.style={draw=bDraw, line width=0.55pt, -{Latex[length=2.4mm,width=1.9mm]}}
    }

    \blockMarked{1}{1}{}{bMark}
    \blockA{1}{2}{}
    \blockA{1}{3}{}
    \blockA{1}{4}{}
    \blockA{1}{5}{}

    \blockMarked{2}{1}{2,3,4,5}{bMark}
    \blockMarked{2}{2}{1,2,3,4}{bMark}
    \blockA{2}{3}{2,3,4,5}
    \blockA{2}{4}{2,3,4,5}
    \blockA{2}{5}{2,3,4,5}

    \blockAB{3}{1}{1,2,3,4}
    \blockMarked{3}{2}{2,3,4,5}{bMark}
    \blockAB{3}{3}{1,2,3,4}
    \blockMarked{3}{4}{2,3,4,5}{bMark}
    \blockMarked{3}{5}{2,3,4,5}{bMark}

    \blockAB{4}{1}{1,2,3,4}
    \blockAB{4}{2}{2,3,4,5}
    \blockAB{4}{3}{1,2,3,4}
    \blockAB{4}{4}{2,3,4,5}
    \blockAB{4}{5}{2,3,4,5}

    \blockAB{5}{1}{1,2,3,4}
    \blockAB{5}{2}{2,3,4,5}
    \blockAB{5}{3}{1,2,3,4}
    \blockAB{5}{4}{2,3,4,5}
    \blockAB{5}{5}{2,3,4,5}

    \blockBA{6}{1}{1,2,3}{arrowB}
    \blockAB{6}{2}{1,2,3,5}
    \blockBA{6}{3}{1,2,3}{arrowB}
    \blockAB{6}{5}{2,3,4,5}

    \blockB{7}{1}{1,2,3}{arrowB}
    \blockBA{7}{2}{1,2,3}{arrowB}
    \blockB{7}{3}{1,2,3}{arrowB}

    \blockB{8}{1}{1,2,3}{arrowB}
    \blockB{8}{2}{1,2,3}{arrowB}
    \blockB{8}{3}{1,2,3}{arrowB}

    \draw[very thick, dashed, rounded corners=12pt, draw=red!70]
        (8.0, 2.5) ellipse [x radius=0.6cm, y radius=0.92cm];

    \node[
        groupBox,
        anchor=north west,
        align=center,
        font=\dagAnnotationFont\linespread{0.92}\selectfont,
        inner xsep=4pt,
        inner ysep=2.5pt
    ] at (6.28, 1.42) {$p_5$ stops producing blocks\\ at round $r+6$, and completely\\ leaves when it sees $f_{l+1}+1$\\ blocks in round $r+7$};

    \foreach \c in {1,...,8} {
        \pgfmathtruncatemacro{\roundOffset}{\c-1}
        \node[font=\small] at ({\c}, {\dagColumnLabelY})
            {\ifnum\roundOffset=0 $r$\else $r+\roundOffset$\fi};
    }

    \foreach \r in {1,...,5} {
        \node[anchor=east] at ({0.45}, {5-\r}) {$p_{\r}$};
    }
\end{tikzpicture}%

%% file: new-sections/Proofs.tex
\section{Correctness}
\label{sec:proof-overview}

We sketch the main arguments; complete proofs appear in
\cref{sec:proofs}.

\begin{restatable}
  [Correctness of Dynamic DAG Byzantine Atomic Broadcast]
  {theorem}
  {dynamicdagbabcorrectness}
\label{theorem:dynamic-dag-bab-correctness}
Under the model of \cref{sec:problem-definition}, suppose that the dealer
follows the issuance rule of \cref{sec: config begin}. For Probabilistic Progress, also assume the conditional
probability requirement of \cref{sec:probabilistic-progress}.
For every finite sequence of reconfiguration messages, with no a priori bound
on its length, the protocol of \cref{sec:our-solution} implements the
leader-level Dynamic DAG Byzantine Atomic Broadcast of
\cref{def:dynamic-dag-bab}.
\end{restatable}

\begin{proof}[Proof sketch.]
\emph{Broadcast safety and reconfiguration progress.}
Certificate validation, authenticated proposal and echo checks, and one
delivery per slot give BRB Integrity. In each round, certificates can use
only one configuration or the two configurations of one valid transition.
Their quorums intersect in a correct process that votes at most once in
the slot, so the slot cannot have certificates for different values or
under both configurations
(\cref{lemma:brb-integrity,lemma:brb-aligned-safety,lemma:brb-safety}).
BRB Integrity and Agreement give DAG-Integrity and DAG-Consistency, while
obtaining and validating every parent and missing ancestor before inserting
any block, including a transferred block, gives Causal Closure
(\cref{lemma:dynamic-dag-structured-broadcast}).

Fix a valid transition $R_l\rightarrow R_{l+1}$.
By \cref{lemma:eventual-first-installation}, let $t$ be the earliest install
transition round among correct processes in $\mathsf{members}(R_l)$.
A correct process that installs in round $t$ has DAG-delivered
$f_l+1$ blocks marked by $R_{l+1}$.
Every round-$(t+1)$ block from a correct broadcaster in
$\mathsf{members}(R_l)$ is marked: either its pending field contains
$R_{l+1}$, or its $q_l$ round-$t$ parent slots intersect those $f_l+1$
marked slots. Every round-$(t+1)$ parent quorum counted under $R_l$ uses
distinct broadcaster slots in $\mathsf{members}(R_l)$, and therefore
contains at least $q_l-f_l\geq f_l+1$ blocks from correct broadcasters
(\cref{lemma:marked-block-propagation,lemma:predecessor-parent-slots}).

Block broadcasts after DAG delivery and $\texttt{FloodDAG}$ let correct
lagging and joining processes obtain and validate certified,
ancestor-closed history
(\cref{lemma:known-coverage,lemma:certified-history-dissemination}).
Using this dissemination and the marking argument above, every correct
process in $\mathsf{members}(R_l)$ installs in round $t$ or $t+1$.
A correct joining process eventually broadcasts its two transition blocks
(\cref{corollary:install-transition-round-bound,%
lemma:eventual-joining-discovery,%
lemma:transition-predecessor-quorum-progress}).

Because install rounds may differ by one, $\texttt{AdvanceRound}$ also accepts
$q_{l+1}$ admissible following-round blocks in the last two transition rounds.
If a required correct process stops before the ordinary
condition holds, its final $\texttt{FloodDAG}$ supplies blocks satisfying the
ordinary or following-round condition, so the wait terminates without
skipping a round
(\cref{lemma:stopping-process-transition-state,%
lemma:transition-lookahead-delivery,corollary:transition-round-progress}).

Even if an added process is Byzantine or a process is removed, the fault
bound leaves enough correct processes to supply the required parent quorums.
Thus the protocol satisfies Per-Round Quorum Progress and Reconfiguration
Completion
(\cref{lemma:reconfiguration-completion,%
lemma:dynamic-dag-structured-broadcast}).
Moreover, a following configuration's earliest inclusion round exceeds
$t+5$, and no installation precedes that inclusion round. Hence distinct
transition windows do not overlap
(\cref{lemma:transition-window-disjointness}); well-spacing is derived rather
than assumed.

\emph{Integrity and Total Order.}
An $R_c$-typed leader-voting block records
$\mathsf{enabled}(\mu)=\{R_c\}$ and carries a valid $R_c$-scoped coin share.
If $R_a$ and $R_b$ respectively supplied $f_a+1$ and $f_b+1$ such blocks
in distinct broadcaster slots, each collection would contain a correct
broadcaster. By \cref{lemma:same-round-enabled-configs}, those broadcasters
cannot record different singleton enabled-configs values in the same
logical round. Thus at most one configuration can supply the $f_c+1$
leader-voting blocks needed to reconstruct a coin for a wave.
Threshold-signature uniqueness and DAG-Consistency then make the selected
leader unique
(\cref{lemma:unique-accepted-leader}).

A direct anchor identified from $R_l$-typed leader-voting blocks is reached
by at least $q_l$ such blocks in the last round of its wave.
Every later-wave block DAG-delivered by a correct process reaches at least
$2f_l+1$ $R_l$-declared blocks in that round.
Both collections use distinct broadcaster slots from
$\mathsf{members}(R_l)$, so DAG-Consistency gives at least
$q_l+(2f_l+1)-n_l=f_l+1$ common blocks.
These blocks satisfy the backward-scan acceptance condition. Hence a later
backward scan accepts the earlier direct anchor when it reaches that wave
(\cref{lemma:later-block-last-round-ancestry,%
corollary:direct-anchor-scan-acceptance}).

Each scan depends only on its starting anchor's immutable causal history.
The finalized leader sequences of active correct processes are therefore
nested prefixes, giving Integrity and Total Order
(\cref{lemma:deterministic-leader-scans,lemma:nested-leader-scans,%
corollary:leader-integrity,theorem:leader-total-order}).

\emph{Probabilistic Progress.}
After the finitely many transitions, the final configuration $R_L$ remains
installed and every correct $R_L$-member that remains active advances through
infinitely many stable waves. Each such wave $w$ has a common core~\cite{Canetti1995StudiesIS,canettiFirstGather}
$\mathcal K_w$ of $q_L$ first-round blocks in the causal history
of every last-round block DAG-delivered by a correct process. Every correct
$R_L$-member that remains active eventually DAG-delivers all blocks in
$\mathcal K_w$ and at least $q_L$ distinct $R_L$-typed leader-voting blocks,
each of which reaches every block in $\mathcal K_w$
(\cref{corollary:eventual-stable-suffix,%
corollary:stable-suffix-round-progress,lemma:stable-wave-reachability,%
corollary:stable-wave-availability}).

In every stable wave, the coin selects the broadcaster of a block in
$\mathcal K_w$ with probability bounded below by a fixed positive constant,
regardless of the outcomes of earlier waves
(\cref{sec:probabilistic-progress}).
When this happens, every correct $R_L$-member that remains active
eventually identifies that block as a direct anchor, unless it has already
identified a direct anchor in a later wave
(\cref{proposition:stable-wave-direct-anchor}).

The probability of seeing no further selections from a common core is
therefore zero, regardless of how earlier waves ended.
Consequently, every correct $R_L$-member that remains active identifies
direct anchors in arbitrarily late stable waves with probability one
(\cref{lemma:infinitely-many-direct-anchors}).
Processing these anchors extends the nested finalized leader sequences
with blocks from arbitrarily late waves. Hence every such process
finalizes infinitely many blocks with probability one
(\cref{lemma:nested-leader-scans}).
\end{proof}

%% file: new-sections/Conclusion.tex
\section{Conclusion and Open Questions}

In this work we present the first seamless reconfiguration protocol for DAG-based Byzantine atomic broadcast in a fully asynchronous network. 
Membership and quorum thresholds change while dissemination continues. Reconfiguration occurs seamlessly by processes locally interpreting their view of the DAG. 

We conclude this paper with several open questions that we leave to future work: Firstly, is it possible to for reconfigurations to be seamless and non-blocking when transition windows are allowed to overlap. Furthermore, is it possible to extend reconfigurations to not preclude removal and addition of more than one processes at a time. Lastly, if a downgrade to a partially-synchronous model would invite even more elegant solutions.

%% file: Appendix/Full-algos.tex
\section{Full Algorithms}
\label{sec:full-algorithms}

Throughout the algorithms, $R_a,R_b,R_c$ denote configurations;
$r,s,u$ denote logical rounds; $m$ denotes a block payload, whose optional
coin-share field is $\mathsf{coin}(m)$; $\mu$ denotes
block metadata; $C$ denotes a BRB certificate; $B$ denotes a certified
block; $ID$ denotes a process identity; and $\bot$ denotes absence.

\subsection{Local Variables and Threshold Functions}

At process $p_i$, $\mathit{local}_i$ is the installed configuration and
$\mathit{pending}_i$ is either the authenticated next reconfiguration
message currently recorded by the process or $\bot$.
The Boolean $\mathit{joining}_i$ is true while $p_i$ is performing the
joining operation. The set $\mathit{passive\text{-}rounds}_i$ contains
rounds in which $p_i$ must not broadcast a block, echo, or vote.
The immutable value
$\mathit{enabled\text{-}configs}_i[r]$ records either a singleton set
$\{R_c\}$ or an authenticated pair $\{R_a,R_b\}$ enabled by $p_i$ in round
$r$.
The value $r_i$ is the current logical round. After the round-advancement
task starts, its serialized $\texttt{AdvanceRound}_i$ loop is the only task
that changes this value and always changes it from $r$ to $r+1$. The value
$ID_i$ is the identity of $p_i$, and $DAG_i[r]$ is the set of blocks that
$p_i$ has DAG-delivered in round $r$.

The set $\mathit{known}_i$ contains the identities of the members of every
configuration covered by the highest-numbered authenticated
reconfiguration message processed by $p_i$. It is used only to choose the
recipients of protocol messages; it does not change membership,
eligibility, or a quorum threshold.

For a configuration $R_c$, the projections
$\mathsf{members}(R_c)$, $\mathsf{number}(R_c)$, and
$\mathsf{earliest}(R_c)$ return its membership, configuration number, and
earliest inclusion round. The values $n_c$, $f_c$, and $q_c$ are the
membership size, fault bound, and quorum threshold of $R_c$.
For an enabled-configs value $E$, $\mathsf{Recipients}(E)$ is the union of
the members of the configurations in $E$.

$\mathsf{ValidReconfiguration}(R_l,R_{l+1})$ validates the authorization
evidence, consecutive configuration numbers, earliest inclusion round,
and membership difference. $\texttt{UpdateKnown}_i(R_c)$ authenticates
$R_c$ and the identities carried by it and adds those identities and the
members of $R_c$ to $\mathit{known}_i$. If this enlarges
$\mathit{known}_i$, the procedure invokes $\texttt{FloodDAG}_i$ so that
the newly known identities are sent the blocks currently stored by $p_i$.

$\mathsf{ValidSign}(C,m,r,\mu)$ validates $C$ against the complete BRB value and the
threshold public key of the configuration declared in that value.
$\mathsf{ValidEnabled}(\mu)$ checks that the enabled field is either the
declared configuration alone or the pair belonging to a valid transition.
$\mathsf{ValidPending}(r,\mu)$ checks that a non-$\bot$ pending field is
the valid successor of the declared configuration and is not used before
its earliest inclusion round.
$\mathsf{JoiningMetadata}(\mu,R_a,R_b)$ checks that $R_a\rightarrow R_b$ is
a valid transition, the block declares $R_b$, its enabled field is
$\{R_a,R_b\}$, and its broadcaster belongs to
$\mathsf{members}(R_b)\setminus\mathsf{members}(R_a)$.

$\mathsf{ValidParentCount}_i(s,\mu)$ is evaluated by a correct process
before it echoes or votes. Ordinarily it requires $q_b$ parents for a block
declared under $R_b$. For metadata satisfying
$\mathsf{JoiningMetadata}(\mu,R_a,R_b)$, the $q_a$-parent exception is
accepted only when the process's immutable enabled-configs value for
round $s$ equals the pair $\{R_a,R_b\}$ recorded by the block. This permits
the exception in any round in which that value is assigned; the metadata
does not prove the added process's locally chosen install transition round. The
correct joining process uses the exception only for its two prescribed
transition appends. A $q_b$ parent count remains the ordinary alternative
and does not depend on the exception.

$\mathsf{CertifiedParentCount}(s,\mu)$ is used when validating an already
certified block. It does not consult the receiving process's local
enabled-configs value, because at least one correct voter applied
$\mathsf{ValidParentCount}$ before the certificate was formed. It checks
the ordinary $q_b$ count or the permitted joining-process count recorded
in the certified metadata.

\input{Algos/variables}

\subsection{Configuration-Scoped Byzantine Reliable Broadcast}

$ECHOS_i[m,r,\mu]$ stores one valid echo signature per sender for the
specified BRB value. $VOTES_i[m,r,\mu]$ similarly stores valid BRB vote
shares. The Boolean values
$\mathit{echoed}_i[r,ID]$, $\mathit{voted}_i[r,ID]$, and
$\mathit{delivered}_i[r,ID]$ ensure that $p_i$ echoes, votes, and delivers
at most once for a round-broadcaster slot. A message for a round whose
enabled-configs value is not yet assigned is buffered until that assignment
is made.

$\mathsf{ValidInstance}_i(r,\mu)$ checks that the declared configuration
is enabled at $p_i$ in round $r$, that the broadcaster belongs to that
configuration, and that the enabled, pending, and parent fields are valid.
$\mathsf{MaySendBRB}_i(R_c,r)$ checks whether $p_i$ may broadcast, echo,
or vote under $R_c$ in round $r$: $R_c$ must be enabled, $p_i$ must be a
member, and the round must not be passive.

$\texttt{BRB-Broadcast}_i(m,r,\mu)$ sends the proposed value to the
members of the configurations recorded in the enabled field.
Let $R_c=\mathsf{local}(\mu)$. After collecting $q_c$ valid echoes,
$p_i$ calls $\texttt{SendVote}_i$. Receipt of $f_c+1$ valid votes also
causes a process that has not yet voted to call $\texttt{SendVote}_i$.
Receipt of $q_c$ valid votes produces the certificate $C$ and certified
block $B$. The complete certified block is then sent to every identity
in $\mathit{known}_i$, and the algorithm triggers
$\texttt{BRB-Deliver}_i$.

The temporary values $e$ and $v$ are the exact echo and vote messages that
are signed. The temporary $R_c$ and $ID$ are the declared configuration
and broadcaster identity obtained from $\mu$.

\input{Algos/brb}

\subsection{DAG Append and DAG Delivery}

For a block $B$, $R_B$, $E_B$, and $ID_B$ denote its declared
configuration, enabled field, and broadcaster identity.
$\mathsf{ref}(B)$ is the hash-and-identity reference stored in a child
block. The set $\mathit{scheduled}_i$ prevents $p_i$ from starting more
than one concurrent fetch operation for the same block.

$\mathsf{ParentBlocks}_i(R,r)$ contains the stored round-$r$ blocks that
may be selected as parents of a block whose declared configuration is
$R$. Every such block comes from an identity permitted by its own enabled
field, and its declared configuration number is no greater than that of
$R$. Every block in this set is certified and admissible because only an
admissible block is inserted in $DAG_i$.

For a supplied block set $D$, let $D[v]$ denote its round-$v$ blocks.
$\mathsf{ValidParentSet}_i(B,D)$ checks that every parent reference of
$B$ names a matching block in round $r-1$ of the supplied block set $D$
and that each named parent satisfies the configuration checks used by
$\mathsf{ParentBlocks}$. For round $1$, it instead requires the parent
field to be $\bot$.

$\mathsf{ValidCertifiedMetadata}_i(B)$ validates the certificate,
configuration evidence, broadcaster membership, enabled field, pending
field, and certified parent count.
$\mathsf{Admissible}_i(B)$ additionally requires all of $B$'s referenced
parents to be present and valid in $DAG_i$.

$\texttt{DAG-Append}_i(m,r)$ records the installed, pending, and
enabled-configs values at the beginning of the invocation. For $r>1$, it
waits for the applicable parent quorum, selects the blocks in the temporary
set $S$, and places their references in $h$. Only after $h$ is fixed, a
singleton-enabled last-round invocation sets $\mathsf{coin}(m)$ to the
configuration-scoped coin share; every other correct invocation sets it to
$\bot$. The procedure then constructs $\mu$ and invokes BRB.

$\texttt{ScheduleFetch}_i(B)$ starts
$\texttt{FetchAndDeliver}_i(B)$ if the block is neither stored nor already
being fetched. $\texttt{FetchAndDeliver}_i(B)$ first validates the
certificate and metadata. It then requests every missing parent from all
identities currently in $\mathit{known}_i$, repeating the request when
$\mathit{known}_i$ changes. Once all parents are present, it calls
$\texttt{AddAndDeliver}_i(B)$.

$\texttt{AddAndDeliver}_i(B)$ prevents two blocks from occupying the same
round-broadcaster slot, inserts an admissible block into $DAG_i$, sends
the certified block to every identity in $\mathit{known}_i$, and outputs
$\texttt{DAG-Deliver}_i$.

$\texttt{FloodDAG}_i$ takes the finite set of blocks currently stored in
$DAG_i$ and sends each certified block to every identity in
$\mathit{known}_i$. An active correct process continues invoking this
procedure while it remains active. The temporary set $K$ used by a fetch
operation is a copy of $\mathit{known}_i$ that lets the operation detect
when new identities have been added.

A \textsf{dag-reply} is sent in response to a \textsf{dag-request}.
In contrast, a \textsf{certified-block} message is sent without waiting
for a request when a certificate is completed, when a block is inserted,
or when $\texttt{FloodDAG}$ is invoked. Both message types are passed to
the same validation and fetching procedure.

\input{Algos/DAG-bb}

\subsection{Reconfiguration}

$\mathit{transition}_i$ is either $\bot$ or
$\langle R_l,R_{l+1},r\rangle$, where $r$ is the install transition round
chosen by $p_i$.

$\mathsf{LocalBlocks}_i(R_l,r)$ contains the round-$r$ blocks stored at
$p_i$ whose broadcaster identities belong to $\Pi_l$. This definition
selects blocks by broadcaster identity; it does not require every selected
block to be declared under $R_l$.

For a valid transition $R_l\rightarrow R_{l+1}$,
$\mathsf{Marked}_i(R_l,R_{l+1},B)$ holds when $B$ or one of its ancestors
contains $R_{l+1}$ in its pending field.

$\mathsf{MarkedBlocks}_i(R_l,R_{l+1},r)$ counts marked blocks in
$\mathsf{LocalBlocks}_i(R_l,r)$.

$\mathsf{AdmissibleIDs}_i(R,u)$ contains the distinct identities in
$\mathsf{members}(R)$ for which $p_i$ stores an admissible round-$u$
block.

On receiving $R_{l+1}$, a member of $R_l$ authenticates it, updates
$\mathit{known}_i$, waits until it has installed $R_l$ and reached the
earliest inclusion round, and then records $R_{l+1}$ in
$\mathit{pending}_i$.

$\texttt{AdvanceRound}_i$ is the only task that changes $r_i$ after round
advancement starts, and its loop processes one value $r=r_i$ at a time.
Normally it waits for the parent quorum required in round $r$ and then sets
$r_i$ to $r+1$. The temporary $R_c$ identifies the configuration whose
quorum is required: it is normally the installed configuration, but it
becomes $R_{l+1}$ in the last round before the transition completes.

There is one additional way to advance by one round while
$\mathit{transition}_i=\langle R_l,R_{l+1},r_0\rangle$. If
$r\in\{r_0+2,r_0+3\}$ and the process has at least $q_{l+1}$ blocks in
$\mathsf{ParentBlocks}_i(R_{l+1},r+1)$, it may set $r_i$ to $r+1$ without
first satisfying the ordinary round-$r$ condition. These are certified,
admissible blocks because $\mathsf{ParentBlocks}_i$ selects blocks already
inserted in $DAG_i$. This rule lets a transitioning process use a later DAG
layer that it received from another process. It never assigns a round beyond
$r+1$: after the increment, the same serialized loop evaluates the rule
again for the new current round. In particular, reaching $r_0+4$ invokes
$\texttt{FinishTransition}_i$ through the beginning-of-round rule before
further advancement can occur. If the additional condition is used in
round $r_0+2$, the same stored round-$(r_0+3)$ set satisfies the ordinary
condition in round $r_0+3$ on the next iteration; the two increments are
separate and no round is skipped.

$\texttt{UpdateConfigs}_i(r)$ first learns a pending reconfiguration that
appears in a marked block, if necessary. If at least $f_l+1$ blocks are
marked, it calls $\texttt{InstallTransition}_i$; otherwise it assigns the
old singleton enabled-configs value for the following round.

$\texttt{InstallTransition}_i(R_l,R_{l+1},r)$ records the transition and
assigns $\{R_l,R_{l+1}\}$ in rounds $r+1$ through $r+4$.
$\texttt{FinishTransition}_i(R_{l+1})$ installs $R_{l+1}$ and clears the
pending and transition values. The beginning-of-round rule calls it before
ordinary block production in round $r+4$, so an ordinary append in that
round snapshots $R_{l+1}$ rather than $R_l$.


After a removed process finishes the transition, it stops starting new
block-production operations and stops round advancement, leader-election,
and finalization tasks. While the departure wait is pending, every
block-production operation invoked in an earlier round continues until it
reaches $\texttt{BRB-Broadcast}$. The process continues its remaining BRB
and DAG tasks, including processing received blocks, obtaining missing
ancestors, and repeatedly invoking $\texttt{FloodDAG}$. It waits until

$$
  |\mathsf{AdmissibleIDs}_i(R_{l+1},r+5)|
  \geq f_{l+1}+1,
$$

invokes one final $\texttt{FloodDAG}$, and then stops its remaining tasks.

\input{Algos/reconfig}

\subsection{Joining Process}

\subsubsection{Transition Append}

$\texttt{DAG-Transition-Append}_i$ is used only by a joining process in
rounds $r_0+2$ and $r_0+3$, where $r_0$ is its install transition round.
It waits for $q_l$ admissible parents, stores their references in $h$,
declares $R_{l+1}$, records the enabled pair
$\{R_l,R_{l+1}\}$, sets $\mathsf{coin}(m)=\bot$, and invokes the
$R_{l+1}$-scoped BRB.

\subsubsection{State Transfer}

$\mathit{accepted\text{-}state\text{-}replies}_i$ counts state replies that
have satisfied the state-reply handler's conditions and have been merged. It
lets $\texttt{ContinueDownload}_i$ send another request after each accepted reply;
the temporary $a$ is the counter value recorded before the current request.

$\mathsf{ValidState}_i(D,R_{l+1})$ validates a finite set $D$ received in
a state reply. It rejects conflicting round-broadcaster slots, conflicts
with $DAG_i$, invalid certificates or metadata, and missing parents. For
each round $v$, $D[v]$ denotes the round-$v$ blocks in that finite set. For
every $B=\langle m,u,\mu,C\rangle\in D$, it also requires
$\mathsf{number}(\mathsf{local}(\mu))\leq\mathsf{number}(R_{l+1})$, every
configuration in $\mathsf{enabled}(\mu)$ to have number at most
$\mathsf{number}(R_{l+1})$, and any non-$\bot$ $\mathsf{pending}(\mu)$ to
have number at most $\mathsf{number}(R_{l+1})$.
$\texttt{MergeState}_i(D)$ processes the rounds of a valid set in
increasing order so that parents are inserted before their children.

$\texttt{DiscoverBegin}_i(R_l,R_{l+1})$ waits until some local DAG layer
contains blocks from at least $q_l$ broadcaster identities in $\Pi_l$,
including at least $f_l+1$ blocks marked by $R_{l+1}$, and returns the
least such round.

The procedure $\texttt{ContinueDownload}_i$, with arguments $R_l$ and
$R_{l+1}$, repeatedly sends a state request to every identity in
$\mathit{known}_i$ while $p_i$ is joining.
It repeats after processing a state reply that satisfies the state-reply
handler's conditions or after $\mathit{known}_i$ changes. The temporary
$K$ is the value of $\mathit{known}_i$ used for the current request. The
received set is written as $D_j$ when it was sent by $p_j$; $D_i$ is the
set sent by $p_i$ when it answers another process's request.

The state-request handler runs at every active process. A request satisfying
its guard---a valid transition and a requester added by that transition---
causes the recipient to send its finite ancestor-closed stored blocks for
which the declared configuration, every configuration in the enabled
field, and any non-$\bot$ pending configuration have number at most
$\mathsf{number}(R_{l+1})$. The state-reply handler and the other state-transfer operations run
only while the receiving process is joining.

\subsubsection{Joining Operation}

$\mathsf{HistoryTasks}_i$ contains certified-block receipt, DAG
request/reply and state-reply processing, fetching, DAG delivery, and
repeated $\texttt{FloodDAG}$ invocations.
$\mathsf{CatchUpTasks}_i$ contains the BRB proposal, echo, and vote receive
handlers and $\texttt{AdvanceRound}$.
$\mathsf{PostJoinTasks}_i$ contains the remaining protocol tasks, including
ordinary block production. The first set starts before state discovery,
the second starts after $\texttt{DiscoverBegin}$ has returned $u$ and the
joining process has obtained the required round-$(u+1)$ blocks, and the
third starts after the transition completes.

\input{Algos/join_process}

\subsection{Leader Election and Finalization}

$\mathsf{VotingBlocks}_i(R_l,w)$ contains the singleton-enabled
last-round blocks for wave $w$ that are declared under $R_l$ and carry a
valid $R_l$-scoped coin share.
$\mathsf{WaveConfigs}_i(w)$ contains the configurations for which at least
one such block is present.
$\mathsf{Shares}_i(R_l,w)$ extracts one valid coin share per broadcaster
from those blocks.

$\mathsf{Support}_i(B,R_l,w)$ counts the $R_l$-typed leader-voting blocks
whose causal history contains $B$.
$\mathsf{HistorySupport}_i(B^\star,B,R_l,w)$ counts the leader-voting
blocks that lie in the causal history of $B^\star$ and whose own causal
history contains $B$.

$\texttt{WaveLeader}(R_l,w)$ combines exactly $f_l+1$ valid coin shares.
The resulting configuration-scoped coin selects one member of $\Pi_l$.
The function returns that member's locally stored round-$\mathsf{first}(w)$
block declared under $R_l$, and returns $\bot$ otherwise. The temporary $S_{l,w}$ is the
combined coin signature.

$\mathit{finalized\text{-}values}_i$ is the sequence of elected leader
blocks finalized by $p_i$.
$\mathit{processed\text{-}wave}_i$ is the greatest wave through which its
backward scan has completed.

$\texttt{TryElect}_i(w)$ constructs $A$, the set of locally available
coin-selected candidates having support at least $q_l$. It calls
$\texttt{Elect}_i$ only when this set contains exactly one block.

$\texttt{Elect}_i(B,k)$ scans the unprocessed waves preceding $k$ in
reverse order. The temporary $v$ is the previously processed wave, $L$ is
the sequence being constructed, and $B^\star$ is the most recently
accepted block used for the next causal-history check. For each earlier
wave, $E_w$ contains the candidates having history support at least
$f_l+1$. A unique such candidate is prepended to $L$. The procedure then
appends $L$ to $\mathit{finalized\text{-}values}_i$ and sets
$\mathit{processed\text{-}wave}_i$ to $k$.

Invocations of $\texttt{TryElect}_i$ and $\texttt{Elect}_i$ at one process
are serialized.

\input{Algos/leader_election2}

%% file: Algos/variables.tex
\begin{algorithm}[H]
\caption{Local variables and threshold functions at process $p_i$}
\label{alg:brb-variables}
\scriptsize
\begin{algorithmic}[1]

\State $\mathit{local}_i\gets R_1$;
       $\mathit{pending}_i\gets\bot$;
       $\mathit{joining}_i\gets\mathsf{false}$;
       $\mathit{passive\text{-}rounds}_i\gets\varnothing$
\State $\mathit{known}_i\gets\mathsf{members}(R_1)$
       \Comment{Known identities}
\State $\mathit{enabled\text{-}configs}_i[1]\gets\{R_1\}$ and
       $\mathit{enabled\text{-}configs}_i[r]\gets\bot$ for every $r\geq2$
\State $ID_i\gets p_i$; $r_i\gets1$

\State For $R_l=\langle\mathsf{RECONFIG},\Pi_l,l,\delta_l\rangle$, let
$\mathsf{members}(R_l)=\Pi_l$, $\mathsf{number}(R_l)=l$, and
$\mathsf{earliest}(R_l)=\delta_l$.
\State For every $R_c$, let $n_c=|\mathsf{members}(R_c)|$,
$f_c=\lfloor(n_c-1)/3\rfloor$, and $q_c=n_c-f_c$.

\Statex $\begin{aligned}
\mathsf{ValidReconfiguration}(R_l,R_{l+1})\iff{}&
 (R_l=R_1\ \lor\ R_l\text{ has valid authorization evidence})\\[-0.2ex]
&{}\land R_{l+1}\text{ has valid authorization evidence}\\[-0.2ex]
&{}\land\text{that evidence covers the exact value }
             \mathsf{earliest}(R_{l+1})\\[-0.2ex]
&{}\land \mathsf{earliest}(R_{l+1})>1
 \land \mathsf{number}(R_{l+1})=\mathsf{number}(R_l)+1\\[-0.2ex]
&{}\land
 |\mathsf{members}(R_l)\mathbin{\triangle}\mathsf{members}(R_{l+1})|\leq1.
\end{aligned}$
\Statex Use the block-metadata tuple and projections from
\cref{sec:problem-definition}.
\State $\mathsf{Recipients}(E)=\bigcup_{R\in E}\mathsf{members}(R)$.
\Statex $\mathit{known}_i$ is used only for communication; it does not
alter membership, eligibility, thresholds, or $\mathsf{Recipients}(E)$.
\Statex Receipt of an authenticated $R_c$ also provides
dealer-authenticated evidence for the identities that appeared in
$R_1,\ldots,R_{c-1}$.

\Statex
\Procedure{UpdateKnown$_i$}{$R_c$}
    \If{$R_c$ or the identities carried by it from
        $R_1,\ldots,R_{c-1}$ fail authentication}
        \State \Return
    \EndIf
    \State $K\gets\mathit{known}_i$
    \State add every member of $R_c$ and every identity carried by $R_c$
           to $\mathit{known}_i$
    \If{$\mathit{known}_i\neq K$}
        \State \Call{FloodDAG$_i$}{}
    \EndIf
\EndProcedure

\State $DAG_i[r]\gets\varnothing$ for every $r\geq1$.
\State $\mathsf{ValidSign}(C,m,r,\mu)$ holds iff $C$ is a valid BRB quorum
certificate for $\langle m,r,\mu\rangle$ under the threshold public key of
$\mathsf{local}(\mu)$.

\Statex $\begin{aligned}
\mathsf{ValidEnabled}(\mu)\iff{}&
 \mathsf{enabled}(\mu)=\{\mathsf{local}(\mu)\}\ \lor\\[-0.2ex]
&\exists R_a,R_b:\ \mathsf{enabled}(\mu)=\{R_a,R_b\}\\[-0.2ex]
&\hspace{7em}\land\mathsf{ValidReconfiguration}(R_a,R_b),\\[-0.2ex]
\mathsf{ValidPending}(r,\mu)\iff{}&
 \mathsf{pending}(\mu)=\bot\ \lor\\[-0.2ex]
&\bigl(\mathsf{ValidReconfiguration}(\mathsf{local}(\mu),
       \mathsf{pending}(\mu))\\[-0.2ex]
&\hspace{7em}\land
       r\geq\mathsf{earliest}(\mathsf{pending}(\mu))\bigr),\\[-0.2ex]
\mathsf{JoiningMetadata}(\mu,R_a,R_b)\iff{}&
 \mathsf{local}(\mu)=R_b\land
 \mathsf{enabled}(\mu)=\{R_a,R_b\}\\[-0.2ex]
&{}\land\mathsf{ValidReconfiguration}(R_a,R_b)\\[-0.2ex]
&{}\land
 \mathsf{ID}(\mu)\in\mathsf{members}(R_b)\setminus\mathsf{members}(R_a).
\end{aligned}$

\Statex
\Function{ValidParentCount$_i$}{$s,\mu$}
    \State $R_b\gets\mathsf{local}(\mu)$;
           $E\gets\mathsf{enabled}(\mu)$
    \If{$s>1$ and $\mathit{enabled\text{-}configs}_i[s]=E$ and
        $\exists R_a:\mathsf{JoiningMetadata}(\mu,R_a,R_b)$}
        \State \Return $|\mathsf{parents}(\mu)|\in\{q_a,q_b\}$
    \EndIf
    \State \Return $|\mathsf{parents}(\mu)|=q_b$
\EndFunction

\Statex
\Function{CertifiedParentCount}{$s,\mu$}
    \If{$s=1$}
        \State \Return $\mathsf{parents}(\mu)=\bot$
    \EndIf
    \If{$\mathsf{parents}(\mu)=\bot$ or its broadcaster IDs are not
        pairwise distinct}
        \State \Return $\mathsf{false}$
    \EndIf
    \State $R_b\gets\mathsf{local}(\mu)$
    \If{$\exists R_a:\mathsf{JoiningMetadata}(\mu,R_a,R_b)$}
        \State \Return $|\mathsf{parents}(\mu)|\in\{q_a,q_b\}$
               \Comment{Joining-process exception}
    \EndIf
    \State \Return $|\mathsf{parents}(\mu)|=q_b$
           \Comment{A correct voter checked \texttt{ValidParentCount}$_i$}
\EndFunction
\end{algorithmic}
\end{algorithm}

%% file: Algos/brb.tex
\begin{algorithm}[tbp]
\caption{Configuration-scoped BRB at process $p_i$}
\label{alg:brb}
\scriptsize
\begin{algorithmic}[1]
\State $ECHOS_i[m,r,\mu]\gets\varnothing$;
       $VOTES_i[m,r,\mu]\gets\varnothing$ on first access
\State $\mathit{echoed}_i[r,ID],\mathit{voted}_i[r,ID],
       \mathit{delivered}_i[r,ID]\gets\mathsf{false}$ for every $(r,ID)$
\Statex Future-round messages are buffered until
$\mathit{enabled\text{-}configs}_i[r]\neq\bot$, then processed normally.

\Statex $\begin{aligned}
\mathsf{ValidInstance}_i(r,\mu)\iff{}&
 \mathit{enabled\text{-}configs}_i[r]\neq\bot
 \land\mathsf{local}(\mu)\in\mathit{enabled\text{-}configs}_i[r]\\[-0.2ex]
&{}\land\mathsf{local}(\mu)\in\mathsf{enabled}(\mu)
 \land\mathsf{ID}(\mu)\in\mathsf{members}(\mathsf{local}(\mu))\\[-0.2ex]
&{}\land\mathsf{ValidEnabled}(\mu)
 \land\mathsf{ValidPending}(r,\mu)\\[-0.2ex]
&{}\land\bigl(r=1\ \lor\
 \Call{ValidParentCount\ensuremath{_i}}{r,\mu}\bigr),\\[0.2ex]
\mathsf{MaySendBRB}_i(R_c,r)\iff{}&
 \mathit{enabled\text{-}configs}_i[r]\neq\bot
 \land R_c\in\mathit{enabled\text{-}configs}_i[r]\\[-0.2ex]
&{}\land p_i\in\mathsf{members}(R_c)
 \land r\notin\mathit{passive\text{-}rounds}_i.
\end{aligned}$

\Statex
\Procedure{\texttt{BRB-Broadcast}$_i$}{$m,r,\mu$}
    \If{$\neg\mathsf{MaySendBRB}_i(\mathsf{local}(\mu),r)$}
        \State \Return
    \EndIf
    \State \textbf{broadcast }$\langle m,r,\mu\rangle$ to
        $\mathsf{Recipients}(\mathsf{enabled}(\mu))$
        \Comment{The BRB proposal is configuration-scoped}
\EndProcedure

\Statex \textbf{upon receiving }$\langle m,r,\mu\rangle$
\textbf{ from }$p_j$
    \State $R_c\gets\mathsf{local}(\mu)$;
           $ID\gets\mathsf{ID}(\mu)$
    \If{$\mathsf{ValidInstance}_i(r,\mu)\land p_j=ID\land
         \mathsf{MaySendBRB}_i(R_c,r)\land
         \neg\mathit{echoed}_i[r,ID]$}
        \State $e\gets\langle\mathsf{echo},m,r,\mu\rangle$;
               $\mathit{echoed}_i[r,ID]\gets\mathsf{true}$
        \State \textbf{broadcast }
        $\langle\mathsf{echo},\operatorname{sign}_i(e),m,r,\mu\rangle$
        to $\mathsf{Recipients}(\mathsf{enabled}(\mu))$
    \EndIf

\Statex \textbf{upon receiving }
$\langle\mathsf{echo},\sigma_j,m,r,\mu\rangle$
\textbf{ from }$p_j$
    \State $R_c\gets\mathsf{local}(\mu)$;
           $e\gets\langle\mathsf{echo},m,r,\mu\rangle$
    \If{$\mathsf{ValidInstance}_i(r,\mu)\land
         p_j\in\mathsf{members}(R_c)\land
         \operatorname{verify}_j(\sigma_j,e)$}
        \State $ECHOS_i[m,r,\mu][j]\gets\sigma_j$
        \If{$|ECHOS_i[m,r,\mu]|\geq q_c$}
            \State \Call{SendVote$_i$}{$m,r,\mu$}
        \EndIf
    \EndIf

\Statex
\Procedure{SendVote$_i$}{$m,r,\mu$}
    \State $R_c\gets\mathsf{local}(\mu)$;
           $ID\gets\mathsf{ID}(\mu)$
    \If{$\mathsf{MaySendBRB}_i(R_c,r)\land
        \neg\mathit{voted}_i[r,ID]$}
        \State $v\gets\langle\mathsf{vote},m,r,\mu\rangle$;
               $\mathit{voted}_i[r,ID]\gets\mathsf{true}$
        \State \textbf{broadcast }
        $\langle\mathsf{vote},
        \mathsf{ShareSign}^{\mathsf{BRB}}_{i,c}(v),m,r,\mu\rangle$
        to $\mathsf{Recipients}(\mathsf{enabled}(\mu))$
    \EndIf
\EndProcedure

\Statex \textbf{upon receiving }
$\langle\mathsf{vote},\sigma_j,m,r,\mu\rangle$
\textbf{ from }$p_j$
    \State $R_c\gets\mathsf{local}(\mu)$;
           $ID\gets\mathsf{ID}(\mu)$;
           $v\gets\langle\mathsf{vote},m,r,\mu\rangle$
    \If{$\mathsf{ValidInstance}_i(r,\mu)\land
         p_j\in\mathsf{members}(R_c)\land
         \mathsf{ValidShare}^{\mathsf{BRB}}_{j,c}(\sigma_j,v)$}
        \State $VOTES_i[m,r,\mu][j]\gets\sigma_j$
        \If{$|VOTES_i[m,r,\mu]|\geq f_c+1$}
            \State \Call{SendVote$_i$}{$m,r,\mu$}
        \EndIf
        \If{$|VOTES_i[m,r,\mu]|\geq q_c\land
             \neg\mathit{delivered}_i[r,ID]$}
            \State $C\gets$ the threshold signature from exactly $q_c$
            entries of $VOTES_i[m,r,\mu]$
            \State $B\gets\langle m,r,\mu,C\rangle$
            \State \textbf{broadcast }
                $\langle\mathsf{certified\text{-}block},B\rangle$
                to $\mathit{known}_i$
                \Comment{Flood the complete certificate}
            \State $\mathit{delivered}_i[r,ID]\gets\mathsf{true}$
            \State \textbf{trigger }
                $\texttt{BRB-Deliver}_{i}(m,r,\mu,C)$
        \EndIf
    \EndIf
\end{algorithmic}
\end{algorithm}

%% file: Algos/DAG-bb.tex
\begin{algorithm}[tbp]
\caption{DAG-Append and DAG-Deliver at process $p_i$}
\label{alg:dag-append-deliver}
\scriptsize
\begin{algorithmic}[1]

\State $\mathit{scheduled}_i\gets\varnothing$
\Statex \Comment{Certified blocks received before activation are queued.}

\Statex For $B=\langle m,r,\mu,C\rangle$, write
$R_B=\mathsf{local}(\mu)$, $E_B=\mathsf{enabled}(\mu)$,
$ID_B=\mathsf{ID}(\mu)$, and
\Statex \hfill
$\mathsf{ref}(B)=\langle\mathsf{Hash}(\langle m,r,\mu\rangle),ID_B\rangle$.

\Statex
$\displaystyle
\mathsf{ParentBlocks}_i(R,r)=
\{B\in DAG_i[r]\mid
ID_B\in\mathsf{Recipients}(E_B)\land
\mathsf{number}(R_B)\leq\mathsf{number}(R)\}.$

\Statex For a candidate DAG $D$,
\State
$\begin{aligned}
\mathsf{ValidParentSet}_i(B,D)\iff{}&
(r=1\land\mathsf{parents}(\mu)=\bot)\ \lor\\[-0.2ex]
&\bigl(r>1\land
\forall\rho\in\mathsf{parents}(\mu),\ \exists B'\in D[r-1]:\\[-0.2ex]
&\qquad \mathsf{ref}(B')=\rho\land
ID_{B'}\in\mathsf{Recipients}(E_{B'})\land
\mathsf{number}(R_{B'})\leq\mathsf{number}(R_B)\bigr).
\end{aligned}$

\State
$\begin{aligned}
\mathsf{ValidCertifiedMetadata}_i(B)\iff{}&
r\geq1\land\mathsf{ValidSign}(C,m,r,\mu)\land
(R_B=R_1\ \lor\ R_B\text{ has valid authorization evidence})\land\\[-0.2ex]
&ID_B\in\mathsf{members}(R_B)\land R_B\in E_B\land
\mathsf{ValidEnabled}(\mu)\land\\[-0.2ex]
&\mathsf{ValidPending}(r,\mu)\land\\[-0.2ex]
&\Call{CertifiedParentCount}{r,\mu}.
\end{aligned}$

\State $\mathsf{Admissible}_i(B)\iff
\mathsf{ValidCertifiedMetadata}_i(B)\land
\mathsf{ValidParentSet}_i(B,DAG_i)$
\Statex \Comment{Before certification, a correct echoer or voter uses
$\Call{ValidParentCount\ensuremath{_i}}$ in \cref{alg:brb}.}

\Statex
\Procedure{\texttt{DAG-Append}$_i$}{$m,r$}
    \State $(R_c,P_c,E_c)\gets
    (\mathit{local}_i,\mathit{pending}_i,
    \mathit{enabled\text{-}configs}_i[r])$
        \Comment{Freeze the round metadata}
    \If{$P_c\neq\bot$ and $r<\mathsf{earliest}(P_c)$}
        \State \Return
    \EndIf
    \If{$r=1$}
        \State $h\gets\bot$
    \Else
        \State \textbf{await }
        $|\mathsf{ParentBlocks}_i(R_c,r-1)|\geq q_c$
        \State $S\gets$ any $q_c$ blocks in
        $\mathsf{ParentBlocks}_i(R_c,r-1)$
        \State $h\gets\{\mathsf{ref}(B'):B'\in S\}$
    \EndIf
    \If{$r=\mathsf{last}(w)$ for some $w$ and $E_c=\{R_c\}$}
        \State set $\mathsf{coin}(m)\gets
        \mathsf{ShareSign}^{\mathsf{COIN}}_{i,c}(\langle w\rangle)$
        \Comment{The parent set is already fixed}
    \Else
        \State set $\mathsf{coin}(m)\gets\bot$
    \EndIf
    \State $\mu\gets\langle h,R_c,P_c,E_c,ID_i\rangle$
    \State \Call{\texttt{BRB-Broadcast}$_i$}{$m,r,\mu$}
\EndProcedure

\Statex \textbf{upon }$\texttt{BRB-Deliver}_{i}(m,r,\mu,C)$ \textbf{ do}
    \State \Call{ScheduleFetch$_i$}{$\langle m,r,\mu,C\rangle$}

\Statex \textbf{upon receiving }
$\langle\mathsf{certified\text{-}block},B\rangle$ \textbf{ or }
$\langle\mathsf{dag\text{-}reply},B\rangle$ \textbf{ do}
    \State \Call{ScheduleFetch$_i$}{$B$}

\Statex
\Procedure{ScheduleFetch$_i$}{$B$}
    \State let $B=\langle m,r,\mu,C\rangle$
    \If{$B\notin DAG_i[r]$ and $B\notin\mathit{scheduled}_i$}
        \State $\mathit{scheduled}_i\gets\mathit{scheduled}_i\cup\{B\}$
        \State \textbf{spawn }\Call{FetchAndDeliver$_i$}{$B$}
    \EndIf
\EndProcedure

\Statex
\Procedure{FetchAndDeliver$_i$}{$B$}
    \State let $B=\langle m,r,\mu,C\rangle$
    \If{$\neg\mathsf{ValidCertifiedMetadata}_i(B)$}
        \State \Return
    \EndIf
    \If{$r>1$}
        \ForAll{$\rho=\langle\alpha,ID'\rangle
        \in\mathsf{parents}(\mu)$}
            \While{$\nexists B'\in DAG_i[r-1]:\mathsf{ref}(B')=\rho$}
                \State $K\gets\mathit{known}_i$
                \State \textbf{broadcast }
                $\langle\mathsf{dag\text{-}request},r-1,ID'\rangle$
                to every process in $K$
                \State \textbf{await }
                $\exists B'\in DAG_i[r-1]:\mathsf{ref}(B')=\rho$
                \textbf{ or }$\mathit{known}_i\neq K$
            \EndWhile
        \EndFor
    \EndIf
    \If{$\mathsf{Admissible}_i(B)$}
        \State \Call{AddAndDeliver$_i$}{$B$}
    \EndIf
\EndProcedure

\Statex \textbf{upon receiving }
$\langle\mathsf{dag\text{-}request},s,ID\rangle$
\textbf{ from }$p_j$ \textbf{ do}
    \If{$\exists B\in DAG_i[s]:ID_B=ID$}
        \State \textbf{send }
        $\langle\mathsf{dag\text{-}reply},B\rangle$ to $p_j$
    \EndIf

\Statex
\Procedure{AddAndDeliver$_i$}{$B$}
    \State let $B=\langle m,r,\mu,C\rangle$
    \If{$B\notin DAG_i[r]$ and
    $\nexists B'\in DAG_i[r]:ID_{B'}=ID_B$}
        \State $DAG_i[r]\gets DAG_i[r]\cup\{B\}$
        \State $\mathit{scheduled}_i\gets\mathit{scheduled}_i\setminus\{B\}$
        \State \textbf{broadcast }
        $\langle\mathsf{certified\text{-}block},B\rangle$
        to every process in $\mathit{known}_i$
        \State \textbf{output }$\texttt{DAG-Deliver}_{i}(m,r,\mu)$
    \EndIf
\EndProcedure

\Statex
\Procedure{FloodDAG$_i$}{}
    \State $D\gets\bigcup_{u\geq1}DAG_i[u]$
        \Comment{Finite snapshot}
    \ForAll{$B\in D$}
        \State \textbf{broadcast }
        $\langle\mathsf{certified\text{-}block},B\rangle$
        to every process in $\mathit{known}_i$
    \EndFor
\EndProcedure


\end{algorithmic}
\end{algorithm}

%% file: Algos/reconfig.tex
\begin{algorithm}[tbp]
\caption{Reconfiguration procedure at process $p_i$}
\label{alg:reconfiguration}
\scriptsize
\begin{algorithmic}[1]

\State $\mathit{transition}_i\gets\bot$

\State $\displaystyle
\mathsf{LocalBlocks}_i(R_l,r)=
\{\langle m,r,\mu,C\rangle\in DAG_i[r]\mid
\mathsf{ID}(\mu)\in\mathsf{members}(R_l)\}$.

\State For $B=\langle m,r,\mu,C\rangle$, define
\Statex $\begin{aligned}
\mathsf{Marked}_i(R_l,R_{l+1},B)\iff{}&
\mathsf{ValidReconfiguration}(R_l,R_{l+1})\ \land\\[-0.2ex]
&\exists r',\,B'=\langle m',r',\mu',C'\rangle\in DAG_i[r']:\\[-0.2ex]
&\qquad B\leadsto_i B'\land\mathsf{pending}(\mu')=R_{l+1}.
\end{aligned}$

\State $\displaystyle
\mathsf{MarkedBlocks}_i(R_l,R_{l+1},r)=
|\{B\in\mathsf{LocalBlocks}_i(R_l,r)\mid
\mathsf{Marked}_i(R_l,R_{l+1},B)\}|$.

\State $\displaystyle
\mathsf{AdmissibleIDs}_i(R,u)=
\{\mathsf{ID}(\mu)\mid
B=\langle m,u,\mu,C\rangle\in DAG_i[u],
\mathsf{ID}(\mu)\in\mathsf{members}(R),
\mathsf{Admissible}_i(B)\}$.

\Upon{\shortstack[l]{receiving
$R_{l+1}=\langle\mathsf{RECONFIG},\Pi_{l+1},l+1,\delta_{l+1}\rangle$,\\
where $R_l$ is its unique authenticated predecessor already known to $p_i$,\\
with $p_i\in\mathsf{members}(R_l)$ and
$\mathsf{ValidReconfiguration}(R_l,R_{l+1})$}}
    \State \Call{UpdateKnown$_i$}{$R_{l+1}$}
    \If{$\mathsf{number}(\mathit{local}_i)>\mathsf{number}(R_l)$}
        \State \Return
    \EndIf
    \State \textbf{await }$\mathit{local}_i=R_l$ and
        $r_i\geq\mathsf{earliest}(R_{l+1})$
    \State $\mathit{pending}_i\gets R_{l+1}$
\EndUpon

\Statex
\Procedure{\texttt{AdvanceRound}$_i$}{}
    \While{\textbf{true}}
        \State $r\gets r_i$
        \State \textbf{await }
        $\mathit{local}_i\in\mathit{enabled\text{-}configs}_i[r]$ and
        $\bigl(\mathit{transition}_i=\bot\ \lor\
        (\mathit{transition}_i=\langle R_l,R_{l+1},r_0\rangle
        \land r\leq r_0+3)\bigr)$
        \State $R_l\gets\mathit{local}_i$
        \If{$\mathit{transition}_i=\langle R_l,R_{l+1},r_0\rangle$
            and $r=r_0+3$}
            \State $R_c\gets R_{l+1}$
        \Else
            \State $R_c\gets R_l$
        \EndIf
        \State \textbf{await either }
            $|\mathsf{ParentBlocks}_i(R_c,r)|\geq q_c$
        \Statex \hspace{2.5em}\textbf{or }
            $\mathit{transition}_i=\langle R_l,R_{l+1},r_0\rangle$,
            $r\in\{r_0+2,r_0+3\}$, and
        \Statex \hspace{2.5em}
            $|\mathsf{ParentBlocks}_i(R_{l+1},r+1)|\geq q_{l+1}$
        \If{$\mathit{transition}_i=\bot$}
            \State \Call{\texttt{UpdateConfigs}$_i$}{$r$}
        \EndIf
        \State $r_i\gets r+1$
    \EndWhile
\EndProcedure

\Statex
\Procedure{FinishTransition$_i$}{$R_{l+1}$}
    \State $\mathit{local}_i\gets R_{l+1}$;
    $\mathit{pending}_i\gets\bot$;
    $\mathit{transition}_i\gets\bot$
\EndProcedure

\Statex
\Procedure{\texttt{UpdateConfigs}$_i$}{$r$}
    \State $R_l\gets\mathit{local}_i$
    \If{$\mathit{pending}_i=\bot$ and there exist $R_{l+1}$ and
        $B\in\mathsf{LocalBlocks}_i(R_l,r)$ with
        $\mathsf{Marked}_i(R_l,R_{l+1},B)$}
        \State \Call{UpdateKnown$_i$}{$R_{l+1}$}
        \State $\mathit{pending}_i\gets R_{l+1}$
    \EndIf
    \State $R_{l+1}\gets\mathit{pending}_i$
    \If{$R_{l+1}\neq\bot$ and
        $\mathsf{MarkedBlocks}_i(R_l,R_{l+1},r)\geq f_l+1$}
        \State \Call{InstallTransition$_i$}{$R_l,R_{l+1},r$}
    \Else
        \State $\mathit{enabled\text{-}configs}_i[r+1]\gets\{R_l\}$
    \EndIf
\EndProcedure

\Statex
\Procedure{InstallTransition$_i$}{$R_l,R_{l+1},r$}
    \State $\mathit{transition}_i\gets\langle R_l,R_{l+1},r\rangle$
    \State $\mathit{enabled\text{-}configs}_i[s]
        \gets\{R_l,R_{l+1}\}$ for every
        $s\in\{r+1,r+2,r+3,r+4\}$
\EndProcedure

\Upon{the beginning of round $r+4$, before ordinary block production, with
$\mathit{transition}_i=\langle R_l,R_{l+1},r\rangle$ and
$p_i\in\mathsf{members}(R_l)$}
    \State \Call{FinishTransition$_i$}{$R_{l+1}$}
    \If{$p_i\notin\mathsf{members}(R_{l+1})$}
        \State \textbf{stop starting new block-production operations, and
        stop round advancement, leader-election, and finalization tasks at }$p_i$
        \State \textbf{while the departure wait below is pending, continue
        every previously invoked }\texttt{DAG-Append}
        \textbf{ or }\texttt{DAG-Transition-Append}
        \textbf{ until it reaches }\texttt{BRB-Broadcast}
        \State \textbf{await }
$|\mathsf{AdmissibleIDs}_i(R_{l+1},r+5)|
 \geq f_{l+1}+1$
        \State \Call{FloodDAG$_i$}{}
        \State \textbf{stop all tasks at }$p_i$
    \EndIf
\EndUpon

\end{algorithmic}
\end{algorithm}

%% file: Algos/join_process.tex
\begin{algorithm}[tbp]
\caption{Joining-process validation and transition append at $p_i$}
\label{alg:joining-transition-validation}
\scriptsize
\begin{algorithmic}[1]
\Procedure{\texttt{DAG-Transition-Append}$_i$}
    {$m,s,R_l,R_{l+1},r_0$}
    \If{$\mathit{joining}_i=\mathsf{false}$ or
        $\mathit{transition}_i\neq\langle R_l,R_{l+1},r_0\rangle$ or
        $s\notin\{r_0+2,r_0+3\}$}
        \State \Return
    \EndIf
    \State \textbf{await }
        $|\mathsf{ParentBlocks}_i(R_{l+1},s-1)|\geq q_l$
    \State $S\gets$ any $q_l$ blocks in
        $\mathsf{ParentBlocks}_i(R_{l+1},s-1)$
    \State $h\gets
        \{\langle\mathsf{Hash}(\langle m',s-1,\mu'\rangle),
        \mathsf{ID}(\mu')\rangle\mid
        \langle m',s-1,\mu',C'\rangle\in S\}$
    \State set $\mathsf{coin}(m)\gets\bot$
    \State $\mu\gets
        \langle h,R_{l+1},\bot,
        \mathit{enabled\text{-}configs}_i[s],ID_i\rangle$
    \State \Call{\texttt{BRB-Broadcast}$_i$}{$m,s,\mu$}
\EndProcedure
\end{algorithmic}
\end{algorithm}

\begin{algorithm}[tbp]
\caption{Complete, ancestor-closed DAG state transfer}
\label{alg:joining-state-transfer}
\scriptsize
\begin{algorithmic}[1]
\State $\mathit{accepted\text{-}state\text{-}replies}_i\gets0$
\Statex The state-request handler below runs at every active process. All
other actions in this algorithm run while $p_i$ is joining.
\Function{ValidState$_i$}{$D,R_{l+1}$}
    \State For every round $v$, let
        $D[v]=\{\langle m,v,\mu,C\rangle\in D\}$
    \If{$D$ is not finite, contains distinct blocks in the same
        round-broadcaster slot, or contains a block distinct from the block
        in the same round-broadcaster slot of $DAG_i$}
        \State \Return $\mathsf{false}$
    \EndIf
    \State \Return $\displaystyle\bigwedge_{B=\langle m,u,\mu,C\rangle\in D}
        \Bigl(\mathsf{ValidCertifiedMetadata}_i(B)\land
        \mathsf{ValidParentSet}_i(B,D)$
    \Statex \hspace{5em}$\land\ 
        \mathsf{number}(\mathsf{local}(\mu))
        \leq\mathsf{number}(R_{l+1})$
    \Statex \hspace{5em}$\land\ 
        \forall R\in\mathsf{enabled}(\mu):
        \mathsf{number}(R)\leq\mathsf{number}(R_{l+1})$
    \Statex \hspace{5em}$\land\ 
        \bigl(\mathsf{pending}(\mu)=\bot\ \lor\
        \mathsf{number}(\mathsf{pending}(\mu))
        \leq\mathsf{number}(R_{l+1})\bigr)\Bigr)$
\EndFunction

\Statex
\Procedure{MergeState$_i$}{$D$}
    \ForAll{rounds $u$ represented in $D$, in increasing order}
        \ForAll{$B=\langle m,u,\mu,C\rangle\in D$}
            \State \Call{AddAndDeliver$_i$}{$B$}
        \EndFor
    \EndFor
\EndProcedure

\Statex
\Function{DiscoverBegin$_i$}{$R_l,R_{l+1}$}
    \State \textbf{await }there exists $u\geq1$ such that
        $|\mathsf{LocalBlocks}_i(R_l,u)|\geq q_l$ and
        $\mathsf{MarkedBlocks}_i(R_l,R_{l+1},u)\geq f_l+1$
    \State \Return the least such $u$
\EndFunction

\Statex
\Procedure{ContinueDownload$_i$}{$R_l,R_{l+1}$}
    \While{$\mathit{joining}_i=\mathsf{true}$}
        \State $K\gets\mathit{known}_i$
        \State $a\gets\mathit{accepted\text{-}state\text{-}replies}_i$
        \State \textbf{broadcast }
        $\langle\mathsf{state\text{-}request},
        R_l,R_{l+1}\rangle$ to every process in $K$
        \State \textbf{await either }$\mathit{joining}_i=\mathsf{false}$,
        $\mathit{accepted\text{-}state\text{-}replies}_i>a$, or
        $\mathit{known}_i\neq K$
    \EndWhile
\EndProcedure

\Upon{\shortstack[l]{receiving
$\langle\mathsf{state\text{-}reply},R_l,R_{l+1},D_j\rangle$
from $p_j\in\mathit{known}_i$\\
with $\mathit{joining}_i=\mathsf{true}$,
$\mathit{pending}_i=R_{l+1}$,\\
$\mathsf{ValidReconfiguration}(R_l,R_{l+1})$, and
$\mathsf{ValidState}_i(D_j,R_{l+1})$}}
    \State \Call{MergeState$_i$}{$D_j$}
    \State $\mathit{accepted\text{-}state\text{-}replies}_i\gets
        \mathit{accepted\text{-}state\text{-}replies}_i+1$
\EndUpon

\Upon{receiving
$\langle\mathsf{state\text{-}request},
R_l,R_{l+1}\rangle$ from $p_j$ with
$\mathsf{ValidReconfiguration}(R_l,R_{l+1})$ and
$p_j\in\mathsf{members}(R_{l+1})\setminus\mathsf{members}(R_l)$}
    \State \Call{UpdateKnown$_i$}{$R_{l+1}$}
    \State $D_i\gets$ the maximal finite ancestor-closed subset of
        $\bigcup_{u\geq1}DAG_i[u]$ for which every
        $B=\langle m,u,\mu,C\rangle$ satisfies
    \Statex \hspace{1em}$\mathsf{number}(\mathsf{local}(\mu))
        \leq\mathsf{number}(R_{l+1})$,
        $\mathsf{number}(R)\leq\mathsf{number}(R_{l+1})$ for every
        $R\in\mathsf{enabled}(\mu)$, and
    \Statex \hspace{1em}$\mathsf{pending}(\mu)=\bot$ or
        $\mathsf{number}(\mathsf{pending}(\mu))
        \leq\mathsf{number}(R_{l+1})$
    \State \textbf{send }
        $\langle\mathsf{state\text{-}reply},R_l,R_{l+1},D_i\rangle$
        to $p_j$
\EndUpon
\end{algorithmic}
\end{algorithm}

\begin{algorithm}[tbp]
\caption{Joining operation at process $p_i$}
\label{alg:joining-process}
\scriptsize
\begin{algorithmic}[1]
\State Let $\mathsf{HistoryTasks}_i$ consist of certified-block receipt,
DAG request/reply and state-reply handlers, fetch/deliver workers, and repeated
\texttt{FloodDAG}.
\State Let $\mathsf{CatchUpTasks}_i$ consist of the BRB proposal, echo, and
vote receive handlers and \texttt{AdvanceRound}.
\State $\mathsf{PostJoinTasks}_i=$ all remaining non-joining protocol
tasks, including ordinary block production.
\Statex The task sets are pairwise disjoint. The state-request handler is
not in these sets and runs whenever $p_i$ is active. History tasks may
insert and DAG-deliver certified blocks but never invoke
\texttt{DAG-Append} or \texttt{DAG-Transition-Append}, echo, vote, or
advance a round.

\Upon{receiving
the authenticated pair $\langle R_l,R_{l+1}\rangle$, where
$R_{l+1}=\langle\mathsf{RECONFIG},\Pi_{l+1},l+1,\delta_{l+1}\rangle$, such that
$\mathsf{ValidReconfiguration}(R_l,R_{l+1})$,
$p_i\notin\mathsf{members}(R_l)$, and
$p_i\in\mathsf{members}(R_{l+1})$}
    \State \Call{UpdateKnown$_i$}{$R_{l+1}$}
    \State \textbf{stop all protocol tasks except the state-request handler
        and this joining operation}
    \State \textbf{start }$\mathsf{HistoryTasks}_i$
    \State $\mathit{pending}_i\gets R_{l+1}$;
        $\mathit{local}_i\gets R_1$;
        $\mathit{enabled\text{-}configs}_i[1]\gets\{R_1\}$
    \State $\mathit{joining}_i\gets\mathsf{true}$
    \State \textbf{spawn }\Call{ContinueDownload$_i$}{$R_l,R_{l+1}$}
    \State $r\gets{}$\Call{DiscoverBegin$_i$}{$R_l,R_{l+1}$}
    \State $\mathit{local}_i\gets R_l$
    \State $\mathit{passive\text{-}rounds}_i
        \gets\mathit{passive\text{-}rounds}_i\cup\{r+1\}$
    \State \Call{InstallTransition$_i$}{$R_l,R_{l+1},r$}
    \State \textbf{await }
        $|\mathsf{LocalBlocks}_i(R_l,r+1)|\geq q_l$
    \State $r_i\gets r+1$
    \State \textbf{start }$\mathsf{CatchUpTasks}_i$
\EndUpon

\Upon{the beginning of round $s$ with input $m$ and
$\mathit{transition}_i=\langle R_l,R_{l+1},r\rangle$ and
$s\in\{r+2,r+3\}$}
    \State \Call{\texttt{DAG-Transition-Append}$_i$}
        {$m,s,R_l,R_{l+1},r$}
\EndUpon

\Upon{the beginning of round $r+4$, before ordinary block production, with
$\mathit{joining}_i=\mathsf{true}$ and
$\mathit{transition}_i=\langle R_l,R_{l+1},r\rangle$}
    \State \Call{FinishTransition$_i$}{$R_{l+1}$}
    \State $\mathit{joining}_i\gets\mathsf{false}$
    \State \textbf{start }$\mathsf{PostJoinTasks}_i$
\EndUpon

\end{algorithmic}
\end{algorithm}

%% file: Algos/leader_election2.tex
\begin{algorithm}[tbp]
\caption{Leader-selection helpers at process $p_i$}
\label{alg:leader-election-helpers}
\scriptsize
\begin{algorithmic}[1]

\Statex Use the four-round wave boundaries $\mathsf{first}(w)$ and
$\mathsf{last}(w)$ defined in \cref{sec:problem-definition}.

\Statex We use $\mathsf{LocalBlocks}_i$ from
\cref{alg:reconfiguration} and $\leadsto_i$ from
\cref{sec:dag-bb-definition}.

\Statex $\begin{aligned}
\mathsf{VotingBlocks}_i(R_l,w)=\bigl\{B\mid{}&
B=\langle m,\mathsf{last}(w),\mu,C\rangle,\\[-0.2ex]
&B\in\mathsf{LocalBlocks}_i(R_l,\mathsf{last}(w)),\\[-0.2ex]
&\mathsf{local}(\mu)=R_l,\\[-0.2ex]
&\mathsf{enabled}(\mu)=\{R_l\},\\[-0.2ex]
&\mathsf{coin}(m)\neq\bot\ \land\
\mathsf{ValidShare}^{\mathsf{COIN}}_{\mathsf{ID}(\mu),l}
(\mathsf{coin}(m),\langle w\rangle)\bigr\}.
\end{aligned}$
\Statex \Comment{Blocks with both configurations enabled remain in the DAG
but do not participate in leader election}

\Statex $\mathsf{WaveConfigs}_i(w)=
\{R_l\mid\mathsf{VotingBlocks}_i(R_l,w)\neq\varnothing\}$.

\Statex Every broadcaster $p_j$ that is uncorrupted at production fixes the
parent set of its singleton-enabled last-round block before generating and
releasing its share.
\Statex The share is
$\tau_{j,l,w}=\mathsf{ShareSign}^{\mathsf{COIN}}_{j,l}(\langle w \rangle)$.

\Statex $\begin{aligned}
\mathsf{Shares}_i(R_l,w)=
\bigl\{\langle j,\tau_{j,l,w}\rangle\;\bigm|\;{}&
\exists B=\langle m,\mathsf{last}(w),\mu,C\rangle:\\[-0.2ex]
&B\in\mathsf{VotingBlocks}_i(R_l,w),\\[-0.2ex]
&j=\mathsf{ID}(\mu)
\text{ and }\tau_{j,l,w}=\mathsf{coin}(m)
\bigr\}.
\end{aligned}$

\Statex $\begin{aligned}
\mathsf{Support}_i(B,R_l,w)=
\Bigl|\bigl\{\mathsf{ID}(\mu_X)\;\bigm|\;{}&
X=\langle m_X,\mathsf{last}(w),\mu_X,C_X\rangle
\in\mathsf{VotingBlocks}_i(R_l,w),\\[-0.2ex]
&X\leadsto_i B
\bigr\}\Bigr|.
\end{aligned}$

\Statex $\begin{aligned}
\mathsf{HistorySupport}_i(B^\star,B,R_l,w)=
\Bigl|\bigl\{\mathsf{ID}(\mu_X)\;\bigm|\;{}&
X=\langle m_X,\mathsf{last}(w),\mu_X,C_X\rangle
\in\mathsf{VotingBlocks}_i(R_l,w),\\[-0.2ex]
&B^\star\leadsto_i X
\;\land\;
X\leadsto_i B
\bigr\}\Bigr|.
\end{aligned}$

\Function{WaveLeader}{$R_l,w$}
    \Statex \Comment{Coin and candidate under $R_l$}

    \If{$|\mathsf{Shares}_i(R_l,w)|<f_l+1$}
        \State \Return $\bot$
    \EndIf

    \State $S_{l,w}\gets$ the threshold signature generated from
        exactly $f_l+1$ entries of $\mathsf{Shares}_i(R_l,w)$
    \State $p_j\gets$ the process selected by the fair
        configuration-scoped coin reconstructed from $S_{l,w}$
        \Comment{The coin selects this first-round slot}

    \State \Return the block
    $B_j=\langle m_j,\mathsf{first}(w),\mu_j,C_j\rangle$
    \Statex \hspace{1em}in
    $\mathsf{LocalBlocks}_i(R_l,\mathsf{first}(w))$ with
    $\mathsf{local}(\mu_j)=R_l$ and $\mathsf{ID}(\mu_j)=j$,
    \Statex \hspace{1em}or $\bot$ if no such block exists
\EndFunction

\end{algorithmic}
\end{algorithm}
\begin{algorithm}[tbp]
\caption{Leader election and finalization at process $p_i$}
\label{alg:leader-election_2}
\scriptsize
\begin{algorithmic}[1]
\State $\mathit{finalized\text{-}values}_i\gets\langle\rangle$;
$\mathit{processed\text{-}wave}_i\gets0$
\Comment{Store elected leader blocks in wave order}
\Statex The event handler and its calls to \texttt{TryElect} and
\texttt{Elect} execute as one serialized local task; a later trigger waits
until the current invocation returns.

\Statex \textbf{upon completing wave }$w$, or upon DAG-delivering a new
block in $\mathsf{VotingBlocks}_i(R_l,w)$ for some $R_l$,\textbf{ do}
    \State \Call{TryElect$_i$}{$w$}

\Statex
\Procedure{TryElect$_i$}{$w$}
    \If{$w\leq\mathit{processed\text{-}wave}_i$}
        \State \Return
    \EndIf
    \Statex $\begin{aligned}
    A\gets\bigl\{B\;\bigm|\;{}&
    \exists R_l\in\mathsf{WaveConfigs}_i(w):
    B=\Call{WaveLeader}{R_l,w}
    \;\land\;B\neq\bot,\\[-0.2ex]
    &\mathsf{Support}_i(B,R_l,w)\geq q_l
    \bigr\}.
    \end{aligned}$
    \Comment{All sufficiently supported candidates}

    \If{$A=\{B\}$}
        \State \Call{Elect$_i$}{$B,w$}
    \EndIf
\EndProcedure

\Statex
\Procedure{Elect$_i$}{$B,k$}
    \State $v\gets\mathit{processed\text{-}wave}_i$
        \Comment{Scan unprocessed waves backward}

    \If{$k\leq v$}
        \State \Return
    \EndIf

    \State $L\gets\langle B\rangle$
    \State $B^\star\gets B$
        \Comment{Current scan anchor}

    \For{$w\gets k-1$ \textbf{ downto } $v+1$}
        \Statex $\begin{aligned}
        E_w\gets\bigl\{B'\;\bigm|\;{}&
        \exists R_l\in\mathsf{WaveConfigs}_i(w):
        B'=\Call{WaveLeader}{R_l,w}
        \;\land\;B'\neq\bot,\\[-0.2ex]
        &\mathsf{HistorySupport}_i(B^\star,B',R_l,w)\geq
        f_l+1
        \bigr\}.
        \end{aligned}$
        \Statex \Comment{All candidates supported through $B^\star$}

        \If{$E_w=\{B'\}$}
            \State prepend $B'$ to $L$
            \State $B^\star\gets B'$
                \Comment{Continue from the accepted earlier leader}
        \EndIf
    \EndFor

    \State atomically append the blocks of $L$, in order, to
        $\mathit{finalized\text{-}values}_i$
    \State $\mathit{processed\text{-}wave}_i\gets k$
        \Comment{Ambiguous skipped waves are not reconsidered}
\EndProcedure

\end{algorithmic}
\end{algorithm}

%% file: Appendix/Full-proofs.tex
\section{Proofs}
\label{sec:proofs}

\paragraph*{Proof Conventions.}
We use the system model, configuration notation, definitions of active and
correct processes, and definitions of configuration-scoped BRB, Dynamic
DAG-structured Byzantine Broadcast, and Dynamic DAG Byzantine Atomic
Broadcast from \cref{sec:problem-definition}. In particular, $n_l$, $f_l$,
and $q_l$ have the meanings fixed there.
The following consequences will be used repeatedly:
\[
  n_l\geq3f_l+1,
  \qquad
  q_l\geq2f_l+1,
  \qquad
  q_l-f_l=n_l-2f_l\geq f_l+1.
\]

Throughout, the \emph{issuance-rule assumption} states that the dealer
follows the issuance rule of \cref{sec: config begin}: before issuing
$R_{l+1}$, it observes, in correct ancestor-closed DAG snapshots, $q_l$
certified blocks from distinct broadcaster slots, all declared under
$R_l$, in round $\mathsf{earliest}(R_{l+1})-2$.

For clarity, a process is \emph{active} while it executes at least one
protocol task. An active correct
process continues every task prescribed at that stage and processes every
delivered protocol message. A correct member of an installed configuration
remains active until the reconfiguration protocol permits it to leave. A
correct joining process becomes active for state-transfer when it receives
the authenticated pair $\langle R_l,R_{l+1}\rangle$ that adds it. A joining process
that is temporarily passive with respect to block production and BRB may
therefore still be active because it continues its prescribed state-transfer
tasks. Once a removed process is permitted to leave and stops all prescribed
tasks, it is no longer active. When a process enters its
transition-completion round, the $\texttt{FinishTransition}$
beginning-of-round action executes before that round's ordinary block
production snapshots its metadata.
For a removed process, stopping block production means that it starts no new
block-production operation. While its departure wait is pending, every
$\texttt{DAG-Append}$ or $\texttt{DAG-Transition-Append}$ operation invoked
in an earlier round continues until it reaches $\texttt{BRB-Broadcast}$.
If the departure condition becomes true first, the process performs its
final $\texttt{FloodDAG}$ and may then stop the still-pending operation with
its other remaining tasks.

For a fixed transition $R_l\rightarrow R_{l+1}$, an \emph{old member} is
a process in $\mathsf{members}(R_l)$. Thus, a \emph{correct old member} is
a correct process in $\mathsf{members}(R_l)$, whether it is retained in
$R_{l+1}$ or removed by the transition. The term \emph{old} describes
membership relative to the transition, not whether the process is active;
activity is stated separately when required.

\paragraph*{Configuration-Level Induction Hypothesis.}
For a valid transition $R_l\rightarrow R_{l+1}$, let
$\delta_{l+1}=\mathsf{earliest}(R_{l+1})$. The
\emph{configuration-level induction hypothesis for $R_l$} states that:

\begin{enumerate}
    \item Every correct member of $R_l$ has installed $R_l$;
    \item Every correct process removed by the preceding transition has
    reached the point at which it may leave;
    \item BRB Agreement holds in every round before $\delta_{l+1}$; and
    \item There is a logical round that every active correct member of $R_l$
eventually enters with $R_l$ installed and in which at least $q_l$ active
correct members invoke the prescribed $R_l$-declared
$\texttt{DAG-Append}$ operation and reach its call to
$\texttt{BRB-Broadcast}$.
\end{enumerate}


For $R_1$, the clauses concerning a
preceding transition are vacuous. This hypothesis is a proof invariant, not
an additional assumption on the execution. It is established for every
issued transition by \cref{lemma:configuration-level-induction}; its
induction step from $R_l$ to $R_{l+1}$ is
\cref{lemma:configuration-step-closure}.

Within the proof for $R_l\rightarrow R_{l+1}$, ``the conclusions for
preceding transitions'' means the conclusions already established for
transitions ending at $R_l$, including the bounds on install transition
rounds and the conclusions about transition completion, permission for
removed correct processes to leave, and BRB Agreement.

We use the install transition round and marked-block terminology of
\cref{sec:our-solution}.
For a process with install transition round $a_i$, round $a_i+3$ is its
final transition round and round $a_i+4$ is its transition-completion round.

Throughout, an \emph{enabled-configs value} is the process-local value
$\mathit{enabled\text{-}configs}_i[r]$. For a round-$r$ block with metadata
$\mu$, $\mathsf{enabled}(\mu)$ is the enabled-configs value recorded by the
broadcaster. If a correct process $p_i$ produces the block, then
\[
  \mathsf{enabled}(\mu)
  =
  \mathit{enabled\text{-}configs}_i[r].
\]
For a Byzantine broadcaster, the recorded field need not represent genuine
process-local state. By contrast, $\mathsf{local}(\mu)$ is the block's
single declared configuration and determines the configuration under which
its BRB certificate is formed.

Except for the joining process's
$\texttt{DAG-Transition-Append}$, a \emph{parent-quorum condition} at process
$p_i$ uses the same configuration $R_c$ both to select eligible round-$s$
parents and to set the required threshold:
\[
  |\mathsf{ParentBlocks}_i(R_c,s)|\geq q_c.
\]
In the exceptional append, the joining process selects eligible parents from
$\mathsf{ParentBlocks}_i(R_{l+1},s)$ but needs only $q_l$ of them. Thus,
$R_{l+1}$ determines which parents are eligible, whereas the predecessor
threshold $q_l$ determines how many are required.

For an $\texttt{AdvanceRound}_i$ iteration with frozen current round $r$ and
selected configuration $R_c$, we call
$|\mathsf{ParentBlocks}_i(R_c,r)|\geq q_c$ the \emph{ordinary parent-quorum
condition}. When
$\mathit{transition}_i=\langle R_l,R_{l+1},r_0\rangle$ and
$r\in\{r_0+2,r_0+3\}$, we call
$|\mathsf{ParentBlocks}_i(R_{l+1},r+1)|\geq q_{l+1}$ the
\emph{additional transition-round condition}.
For a correct process, an \emph{ordinary block} is a block produced by its
prescribed $\texttt{DAG-Append}$ operation, rather than by
$\texttt{DAG-Transition-Append}$.

Whenever a result below says that BRB Agreement holds for a
round-broadcaster slot, we prove and use the following
certificate-level statement: any two valid quorum certificates for that
slot bind the same complete BRB value. This statement implies the BRB
Agreement property of \cref{sec:configuration-brb}, and it also rules out a
conflicting certified block received through a certified-block message,
state reply, or DAG reply even when the recipient did not first
BRB-deliver that block. The certificate-level statement is established by
\cref{lemma:brb-aligned-safety} and is preserved in every later result that
invokes that lemma to establish BRB Agreement.
We say that BRB Agreement holds in round $r$ when it holds for every
round-broadcaster slot in round $r$, and that it holds through round $r$
when it holds for every round-broadcaster slot whose round index is at most
$r$.

For a correct process $p_i$ and blocks in its DAG, we use the causal
relation $Y\leadsto_i X$ from \cref{sec:problem-definition}; thus $X$ is
in the causal history of $Y$. When a statement has already fixed the
correct process whose DAG is under discussion, we may omit the subscript
and write $Y\leadsto X$.

For any block $Y=\langle m_Y,r_Y,\mu_Y,C_Y\rangle$ and any
round $s\leq r_Y$, let
\[
\mathcal H_{Y,i}(s)=
\left\{
X=\langle m_X,s,\mu_X,C_X\rangle\in DAG_i[s]
\;\middle|\;
Y\leadsto_i X
\right\}
\]
denote the round-$s$ layer of the causal history of $Y$ at $p_i$. After
fixing $p_i$, we abbreviate this set by $\mathcal H_Y(s)$.
\paragraph*{Proof Roadmap.}
The proof proceeds in six steps.
\begin{enumerate}
  \item \emph{BRB Integrity and Agreement.}
  Under the configuration-level induction hypothesis, the issuance rule
  gives the starting conditions for each transition
  (\cref{lemma:well-spaced-predecessor-completion,lemma:well-spaced-base-round}).
  \Cref{lemma:brb-integrity} proves Integrity, and
  \cref{lemma:brb-aligned-safety} proves Agreement when certificates use one
  configuration or the two configurations of one valid transition.
  \Cref{lemma:brb-agreement-transition-rounds} establishes alignment and
  Agreement in every round in which a correct old member's assigned
  enabled-configs value contains $R_l$.

  \item \emph{DAG delivery and round advancement.}
  \Cref{lemma:local-dag-ancestor-closure} proves that every correct process's
  DAG contains the ancestors of each stored block.
  \Cref{lemma:certified-history-dissemination,lemma:correct-block-dissemination}
  establish dissemination under their stated Agreement, validation, and
  activity conditions.
  \Cref{lemma:transition-lookahead-delivery,corollary:round-advancement}
  give the corresponding round-advancement results.

  \item \emph{Transition progress, completion, and departure.}
  \Cref{lemma:eventual-first-installation,lemma:eventual-joining-discovery}
  establish the earliest correct old-member install transition round $t$
  and, for a correct joining process, the discovery round $u$.
  \Cref{corollary:install-transition-round-bound,lemma:joining-install-round-cases}
  bound the install transition rounds, and
  \cref{lemma:brb-safety-joining-round} covers the remaining joining-process
  Agreement case.
  \Cref{corollary:transition-round-progress,lemma:retained-joining-completion,lemma:removed-process-departure}
  prove transition progress, completion, and permission to leave.
  \Cref{lemma:next-outer-step-round} supplies the round required for the
  next configuration-level induction step.

  \item \emph{Configuration-level induction and Dynamic DAG-structured
  Byzantine Broadcast.}
  \Cref{lemma:configuration-step-closure,lemma:configuration-level-induction}
  establish the induction hypothesis throughout the authenticated
  configuration sequence for every configuration with an issued successor.
  \Cref{lemma:reconfiguration-completion,lemma:brb-safety,lemma:dynamic-dag-structured-broadcast}
  then establish Reconfiguration Completion and Dynamic DAG-structured
  Byzantine Broadcast.

  \item \emph{Leader Integrity and Total Order.}
  \Cref{lemma:transition-window-disjointness} proves that transition windows
  are disjoint.
  \Cref{lemma:unique-accepted-leader,lemma:deterministic-leader-scans,lemma:nested-leader-scans}
  establish consistent leader choices and nested leader sequences.
  \Cref{corollary:leader-integrity,theorem:leader-total-order} conclude
  Leader Integrity and Total Order.

  \item \emph{Probabilistic Progress.}
  After the final reconfiguration,
  \cref{corollary:eventual-stable-suffix,corollary:stable-suffix-round-progress}
  establish continuing round progress under $R_L$.
  \Cref{lemma:stable-wave-reachability,corollary:stable-wave-availability}
  establish the common core and its availability.
  Under the conditional probability requirement,
  \cref{proposition:stable-wave-direct-anchor,lemma:infinitely-many-direct-anchors,theorem:leader-probabilistic-progress}
  prove that every correct $R_L$-member that remains active finalizes
  infinitely many blocks with probability one.
\end{enumerate}
\Cref{theorem:dynamic-dag-bab-correctness} combines these results to establish
Dynamic DAG Byzantine Atomic Broadcast.

\subsection{Dynamic DAG-structured Byzantine Broadcast Progress and Reconfiguration}
\label{sec:dynamic-dag-progress}

\paragraph*{Assumptions and Process Operations.}
We use the system and reconfiguration assumptions of
\cref{sec:problem-definition}. Correct processes invoke the prescribed
$\texttt{DAG-Append}$ or
$\texttt{DAG-Transition-Append}$ operation. Each authenticated $R_c$ carries
dealer-authenticated evidence for every identity that appeared in
$R_1,\ldots,R_{c-1}$. On accepting $R_c$, a correct process $p_i$ invokes
$\texttt{UpdateKnown}_i(R_c)$. If $R_c$ is the highest-numbered reconfiguration
message it has processed, then
$\mathit{known}_i=\bigcup_{d=1}^{c}\mathsf{members}(R_d)$, and immediately
$\texttt{UpdateKnown}_i(R_c)$ invokes one finite $\texttt{FloodDAG}$ whenever
this set grows. Every later insertion sends its
complete certified block to every identity in $\mathit{known}_i$, every
active correct process repeatedly invokes
$\texttt{FloodDAG}$, and a removed correct process performs one final finite
flood before stopping its remaining DAG tasks. Certified-block messages
addressed to an identity before it activates are queued and processed when it
activates. While active, a correct process executes the prescribed
certified-block receive, fetch, $\texttt{FloodDAG}$, and applicable
state-request and state-reply handlers. These are
protocol actions, not an assumption that a particular process storing a block remains
available until every lagging process catches up.
When $R_{l+1}$ adds a process, the dealer sends it the authenticated pair
$\langle R_l,R_{l+1}\rangle$. The joining process therefore knows the full
membership and thresholds of $R_l$ without receiving all earlier
reconfiguration messages.
Every valid quorum certificate for a complete value
$\langle m,s,\mu\rangle$ contains a vote from a correct process $p_v$ that
checked $\mathsf{ValidParentCount}_v(s,\mu)$ when the certificate was formed.
A process $p_i$ later validating a certified block
$B=\langle m,s,\mu,C\rangle$ may therefore apply
$\mathsf{CertifiedParentCount}(s,\mu)$; it need not reconstruct its own
enabled-configs value for the earlier round. The predicate
$\mathsf{ValidCertifiedMetadata}_i(B)$ requires the declared
configuration to be $R_1$ or to carry valid authorization evidence, and
applies the corresponding checks to every configuration named by the
metadata.
In particular, suppose a certified block
$B=\langle m,s,\mu,C\rangle$ satisfies
$\mathsf{JoiningMetadata}(\mu,R_l,R_{l+1})$ and uses $q_l$ parents while
declaring $R_{l+1}$. If $q_l<q_{l+1}$, a correct process $p_v$ accepts
this parent count only if
$\mathit{enabled\text{-}configs}_v[s]=\{R_l,R_{l+1}\}$ and
$\mathsf{ValidReconfiguration}(R_l,R_{l+1})$ hold.
If $q_l=q_{l+1}$, the parent count satisfies the ordinary rule and does
not require the voter's enabled-configs value to equal the pair.
In either case, a process that later receives $B$ applies
$\mathsf{CertifiedParentCount}(s,\mu)$ to the immutable metadata and does
not require its own enabled-configs value for round $s$. The analysis
uses the issuance-rule assumption; \cref{lemma:transition-window-disjointness}
derives from it that transition windows of distinct transitions are
disjoint.

A block is \emph{admissible} at process $p_i$ if it passes all checks
required for insertion into $p_i$'s DAG: its certificate and metadata are
valid, and its referenced parents satisfy the applicable identity,
configuration-number, and parent-count checks. In particular, if $\mu'$ is
the metadata of a
parent, that parent's broadcaster belongs to
$\mathsf{members}(R_c)$ for some
$R_c\in\mathsf{enabled}(\mu')$. During a transition,
$\mathsf{enabled}(\mu')$ may be the pair $\{R_l,R_{l+1}\}$ comprising the
two configurations of a valid transition; hence a certified block from a
process removed by
$R_{l+1}$ may still be used as a parent during that transition. The parent
must also satisfy
\[
  \mathsf{number}(\mathsf{local}(\mu'))
  \leq
  \mathsf{number}(\mathsf{local}(\mu)),
\]
where $\mu$ is the child metadata. A valid round-$r$ certificate declared
under $R_c$ contains a share from at least one correct voter $p_v$ satisfying
$R_c\in\mathit{enabled\text{-}configs}_v[r]$. The lemmas below therefore
determine which configurations may be declared by valid round-$r$
certificates. Correct processes apply the same metadata checks to
parents selected by Byzantine broadcasters.

Accordingly, at process $p_i$, $B\in\mathsf{LocalBlocks}_i(R_l,r)$ means that
$B$ is already
in $DAG_i[r]$ and hence has passed the applicable DAG-insertion predicate,
and that its broadcaster slot belongs to $\mathsf{members}(R_l)$. It does
not assert that the block is declared under $R_l$, i.e., it need not satisfy
$\mathsf{local}(\mu)=R_l$.

The prescribed $\texttt{DAG-Append}$ invocation by $p_i$ snapshots
$\mathit{local}_i$, $\mathit{pending}_i$, and
$\mathit{enabled\text{-}configs}_i[r]$ before waiting for parents and records
those values in the block metadata. Moreover, a correct process rejects
a pending reconfiguration in any round earlier than the authenticated
earliest-inclusion round $\mathsf{earliest}(R_c)$ of that pending
configuration $R_c$. The authenticated reconfiguration message covers that
round. Hence neither a delayed correct append nor a
Byzantine broadcaster can obtain a certificate that places a later message
in a block's pending field before that round.

For a fixed transition $R_l\rightarrow R_{l+1}$, we say that the round-$r$
enabled-configs values assigned by correct processes in
$\mathsf{members}(R_l)$ are \emph{aligned} if every assigned value is one of
\[
  \{R_l\},\qquad
  \{R_l,R_{l+1}\},\qquad
  \{R_{l+1}\},
\]
and there do not exist correct
$p_i,p_j\in\mathsf{members}(R_l)$ such that
\[
  \mathit{enabled\text{-}configs}_i[r]=\{R_l\}
  \qquad\text{and}\qquad
  \mathit{enabled\text{-}configs}_j[r]=\{R_{l+1}\}.
\]
Here and below, an ``assigned'' value means a value different from $\bot$.
Thus the pair may coexist with either
singleton, but the old and new singletons cannot coexist. Here the old
singleton is $\{R_l\}$ and the new singleton is $\{R_{l+1}\}$. A pair
belonging to another transition in the same round is excluded within the
alignment proofs by the predecessor-completion and earliest-inclusion
bounds, and in general by \cref{lemma:transition-window-disjointness}.

Whenever one of these correct processes $p_i$ starts a round-$r$ append with
metadata $\mu$, it records
\[
  \mathsf{enabled}(\mu)
  =\mathit{enabled\text{-}configs}_i[r].
\]
Consequently, the $\mathsf{enabled}(\mu)$ fields recorded by round-$r$ blocks
produced by these correct processes are drawn from the aligned process-local
values. The alignment condition is stated for process-local enabled-configs
values
because a correct process may assign a value for a round without producing a
block in that round.

\paragraph*{BRB Integrity and Agreement when Certificates Use One
Configuration or the Two Configurations of One Valid Transition.}
The BRB algorithm supplies Integrity independently of reconfiguration.
Agreement is proved separately after restricting the configurations that
can certify one slot.

\begin{lemma}[BRB Integrity]
\label{lemma:brb-integrity}
The configuration-scoped BRB algorithm in \cref{alg:brb} satisfies BRB
Integrity. Moreover, every valid certificate in a correct broadcaster's
round-broadcaster slot binds the complete value supplied to that
broadcaster's prescribed BRB invocation.
\end{lemma}

\begin{proof}
A correct process BRB-delivers at most once per slot and only after verifying
a certificate under the declared configuration's threshold key. Under the
BRB algorithm, a correct process first votes for a value declared under
$R_c$ only after receiving either $q_c$ valid echoes or $f_c+1$ valid votes
for that value. When the declared broadcaster is correct, consider the earliest
correct process that votes for the certified value. This process cannot use
the $f_c+1$-vote relay rule: among those votes at least one is correct,
contradicting the choice of the earliest correct voter. It must therefore
have received $q_c$ valid echoes, including a correct echo. A correct
process sends an echo only after receiving the authenticated initial message
from the declared broadcaster. Hence a value attributed to a correct
broadcaster originates in that broadcaster's prescribed BRB invocation.
\end{proof}

\begin{lemma}[Two valid certificates for one round-broadcaster slot cannot
bind different values within one configuration or one valid transition]
\label{lemma:brb-aligned-safety}
Fix a round-broadcaster slot $(r,ID)$. Any two valid BRB quorum
certificates for this slot under the same declared configuration bind the
same complete BRB value $\langle m,r,\mu\rangle$ and, by
threshold-signature uniqueness, are the same certificate. For any valid
transition $R_l\rightarrow R_{l+1}$, a valid certificate for this slot
declared under $R_l$ and a valid certificate for this slot declared under
$R_{l+1}$ cannot both be formed. Consequently, if certificates for the slot
can be formed only under one configuration or under the two configurations
of one valid transition, any two valid certificates for the slot bind the
same complete BRB value. In particular, any two correct processes that
BRB-deliver for the slot deliver that same value.
\end{lemma}

\begin{proof}
If two certificates are declared under the same configuration $R_c$, their
voter sets
intersect in at least
\[
  2q_c-n_c\geq f_c+1
\]
members and hence in a correct process. If one certificate is declared under
$R_l$ and the other under $R_{l+1}$, an addition gives a voter-set
intersection of at least
\[
  q_l+q_{l+1}-n_{l+1}\geq f_{l+1}+1,
\]
whereas a removal gives an intersection of at least
\[
  q_l+q_{l+1}-n_l\geq f_l+1.
\]
If the memberships are unchanged, then
$n_{l+1}=n_l$, $f_{l+1}=f_l$, and $q_{l+1}=q_l$, so the intersection has
size at least
\[
  2q_l-n_l\geq f_l+1.
\]
In all three cases, the intersection contains a correct process. A correct
process casts at most one vote in a fixed round-broadcaster slot, and that
vote binds the complete value $\langle m,r,\mu\rangle$. Therefore,
certificates under one configuration bind the same tuple and, by
threshold-signature uniqueness under that configuration's key, are the same
certificate. Certificates declared under the two configurations of the
valid transition
cannot both be formed, because a correct process in their voter-set
intersection would have to vote twice in the slot.
\end{proof}

\subsubsection{Consequences of the Issuance Rule for Earliest-Inclusion and
Install Transition Rounds}

\paragraph*{Consequences of the Issuance Rule.}
For the lemmas in this subsection, fix a transition
$R_l\rightarrow R_{l+1}$, let
$\delta_{l+1}=\mathsf{earliest}(R_{l+1})$, and suppose that, before issuing
$R_{l+1}$, the dealer observes, as its issuance rule in
\cref{sec: config begin} prescribes, $q_l$ certified blocks from distinct
broadcaster slots, all declared under $R_l$, in round $\delta_{l+1}-2$.

\begin{lemma}[No certified block below $\delta_{l+1}$ contains $R_{l+1}$ in
its pending field]
\label{lemma:no-premature-pending}
No certified block below round $\delta_{l+1}$ has pending field $R_{l+1}$.
\end{lemma}

\begin{proof}
Every correct echoer and voter rejects a block in round
$s<\delta_{l+1}$ whose pending field contains $R_{l+1}$. At most $f_c$
members of any declared configuration $R_c$ are Byzantine, and $q_c>f_c$.
Byzantine members alone therefore cannot form a quorum certificate.
\end{proof}

\begin{lemma}[Before acting in round $\delta_{l+1}-1$, an active correct
$R_l$-member has installed $R_l$ and completed any transition into it]
\label{lemma:well-spaced-predecessor-completion}
Fix a valid transition $R_l\rightarrow R_{l+1}$ in an execution satisfying
the issuance-rule assumption, let
$\delta_{l+1}=\mathsf{earliest}(R_{l+1})$, and assume the
configuration-level induction hypothesis for $R_l$ and the conclusions
already established for preceding transitions.
Every correct process
$p_i\in\mathsf{members}(R_l)$ that is still active when it reaches logical
round $\delta_{l+1}-1$ has already installed $R_l$. If $l>1$, it has also
completed the transition into $R_l$ before performing any
round-$(\delta_{l+1}-1)$ protocol action.
\end{lemma}

\begin{proof}
The claim is immediate for the initial configuration $R_1$. Otherwise, let $t'$
be the earliest correct old-member install transition round for the preceding
transition into $R_l$. The conclusions already established for the
preceding transition imply that every correct process retained in $R_l$ installs
in round $t'$ or $t'+1$ and
completes the transition at the beginning of round $t'+4$ or $t'+5$. A
correct process added by that transition has install transition round in
$\{t'-1,t',t'+1\}$ and also completes by round $t'+5$.

Before round $t'+4$, no correct process in
$\mathsf{members}(R_{l-1})$ can produce a block declared under $R_l$. Such
blocks can occupy at most the $f_l$ Byzantine slots and, for an addition,
the joining process's slot. Since $f_l+1<q_l$, the observed $q_l$ blocks in
round $\delta_{l+1}-2$ imply
\[
 \delta_{l+1}-2\geq t'+4.
\]
Thus every active correct process in $\mathsf{members}(R_l)$ completes the
transition into $R_l$ by $t'+5\leq\delta_{l+1}-1$.
\end{proof}

\begin{lemma}[Every assigned round-$\delta_{l+1}$ enabled-configs value of a
correct $R_l$-member is $\{R_l\}$]
\label{lemma:well-spaced-base-round}
Every correct process $p_i\in\mathsf{members}(R_l)$ for which
$\mathit{enabled\text{-}configs}_i[\delta_{l+1}]$ is assigned satisfies
$\mathit{enabled\text{-}configs}_i[\delta_{l+1}]=\{R_l\}$.
\end{lemma}

\begin{proof}
By \cref{lemma:well-spaced-predecessor-completion}, $p_i$ has completed the
transition into $R_l$. By \cref{lemma:no-premature-pending}, no block below
$\delta_{l+1}$ is marked by $R_{l+1}$, so no correct process can invoke
$\texttt{InstallTransition}$ with an install transition round below
$\delta_{l+1}$. An invocation with install transition round $a_i$ assigns
$\{R_l,R_{l+1}\}$ only to rounds $a_i+1,\ldots,a_i+4$; it never changes
the value for round $a_i$. Hence the assigned value for round
$\delta_{l+1}$ is $\{R_l\}$.
\end{proof}

\begin{lemma}[From round $\delta_{l+1}-1$ until transition completion, every
block a correct process counts toward an $R_l$ parent quorum has an
$R_l$-member broadcaster]
\label{lemma:predecessor-parent-slots}
Fix a valid transition $R_l\rightarrow R_{l+1}$ in an execution satisfying
the issuance-rule assumption, let
$\delta_{l+1}=\mathsf{earliest}(R_{l+1})$, and assume the
configuration-level induction hypothesis for $R_l$ and the conclusions
already established for preceding transitions.
Let $p_i$ be correct. Before
$p_i$ invokes $\texttt{FinishTransition}_i(R_{l+1})$, let
$s\geq\delta_{l+1}-1$ be any round for which the protocol requires $p_i$
to satisfy
\[
  |\mathsf{ParentBlocks}_i(R_l,s)|\geq q_l.
\]
Then every block
$B=\langle m,s,\mu,C\rangle\in\mathsf{ParentBlocks}_i(R_l,s)$ satisfies
\[
  \mathsf{ID}(\mu)\in\mathsf{members}(R_l).
\]
\end{lemma}

\begin{proof}
The conclusions already established for preceding transitions imply that
configurations earlier than $R_{l-1}$ cannot contribute certified blocks in
these rounds. Hence the only possible
broadcaster identity outside $\mathsf{members}(R_l)$ comes from the
immediately preceding membership $\mathsf{members}(R_{l-1})$. If the
transition into $R_l$ added a member, then
\[
  \mathsf{members}(R_{l-1})\subseteq\mathsf{members}(R_l),
\]
so there is no such identity. Suppose instead that the transition removed
one member, so $|\mathsf{members}(R_{l-1})|=n_l+1$. Among the observed $q_l$
round-$(\delta_{l+1}-2)$ blocks declared under $R_l$, at least
$q_l-f_l=n_l-2f_l$ have correct broadcasters. For each such correct
broadcaster $p_h$, let $a_h$ be its install transition round for
$R_{l-1}\rightarrow R_l$. Because that transition is a removal, $p_h$ is
retained; there is no joining process whose
$\texttt{DAG-Transition-Append}$ can declare $R_l$ before transition
completion. By \cref{alg:reconfiguration}, an ordinary block produced by
$p_h$ declares $R_l$ only after $\texttt{FinishTransition}$ executes at the
beginning of round $a_h+4$. Its observed $R_l$-declared block in round
$\delta_{l+1}-2$ therefore implies
\[
  a_h+4\leq\delta_{l+1}-2.
\]
Thus $\delta_{l+1}-1\geq a_h+5$, whereas
$\texttt{InstallTransition}$ assigns $\{R_{l-1},R_l\}$ only in rounds
$a_h+1,\ldots,a_h+4$. Moreover, by
\cref{lemma:earliest-install-lower-bound}, no correct process installs
$R_l\rightarrow R_{l+1}$ with an install transition round below
$\delta_{l+1}$. Hence every assigned value in round
$s\geq\delta_{l+1}-1$ excludes $R_{l-1}$. For the fixed round $s$, if
$\mathit{enabled\text{-}configs}_h[s]$ is unassigned, $p_h$ emits no share;
otherwise it refuses every round-$s$ BRB instance declared under
$R_{l-1}$.
Because
\[
 n_l-2f_l\geq f_{l-1}+1,
\]
fewer than $q_{l-1}$ vote shares can ever be emitted for such an instance.
Delayed messages cannot increase that number. A block declared under
$R_{l+1}$ cannot belong to $\mathsf{ParentBlocks}_i(R_l,s)$ because it
violates
$\mathsf{number}(R_{l+1})\leq\mathsf{number}(R_l)$. Every remaining counted
block therefore has a broadcaster in $\mathsf{members}(R_l)$, including a
member removed by $R_l\rightarrow R_{l+1}$.
\end{proof}

\begin{lemma}[No correct process invokes $\texttt{InstallTransition}$ with
install transition round below $\delta_{l+1}$]
\label{lemma:earliest-install-lower-bound}
Fix a valid transition $R_l\rightarrow R_{l+1}$ in an execution satisfying
the issuance-rule assumption, and let
$\delta_{l+1}=\mathsf{earliest}(R_{l+1})$. No correct process
$p_i\in\mathsf{members}(R_l)\cup\mathsf{members}(R_{l+1})$ invokes
$\texttt{InstallTransition}_i(R_l,R_{l+1},a_i)$ with
$a_i<\delta_{l+1}$. Consequently, if
\[
  t=\min\bigl\{a_i:
  p_i\in\mathsf{members}(R_l)\text{ is correct and invokes }
  \texttt{InstallTransition}_i(R_l,R_{l+1},a_i)\bigr\}
\]
exists, then $t\geq\delta_{l+1}$.
\end{lemma}

\begin{proof}
An invocation with install transition round $a_i$ is justified by $f_l+1$
round-$a_i$ blocks marked for $R_l\rightarrow R_{l+1}$: at a correct old
member through
$\mathsf{MarkedBlocks}_i(R_l,R_{l+1},a_i)\geq f_l+1$ in
\cref{alg:reconfiguration}, and at a joining process through the discovery
condition of \cref{alg:joining-state-transfer}. A marked block has, in its
reflexive causal history, a certified block whose pending field contains
$R_{l+1}$, and the causal history of a round-$a_i$ block lies in rounds at
most $a_i$. By \cref{lemma:no-premature-pending}, no certified block below
round $\delta_{l+1}$ has pending field $R_{l+1}$. Hence no round-$a_i$
block is marked when $a_i<\delta_{l+1}$, and no such invocation occurs.
\end{proof}

\subsubsection{Enabled-Configs Alignment and BRB Agreement During
$R_l\rightarrow R_{l+1}$}

\paragraph*{Alignment of Assigned Enabled-Configs Values and BRB Agreement.}
We first establish two prerequisites: a following reconfiguration
$R_{l+1}\rightarrow R_{l+2}$ cannot become applicable while correct
processes may still assign enabled-configs values containing $R_l$, and
evidence of $R_l\rightarrow R_{l+1}$ propagates
through marked blocks. We then rule out a first logical round in which
alignment or BRB Agreement fails. This establishes that the enabled-configs
values assigned by correct processes remain aligned for
$R_l\rightarrow R_{l+1}$ and that every valid certificate is declared under
either $R_l$ or $R_{l+1}$ throughout the transition.

\begin{lemma}[Issuing $R_{l+2}$ requires a correct $R_l$-member to invoke
$\texttt{InstallTransition}$ for $R_l\rightarrow R_{l+1}$]
\label{lemma:following-reconfiguration-requires-installation}
Fix a valid transition $R_l\rightarrow R_{l+1}$ in an execution satisfying
the issuance-rule assumption. If the
dealer issues $R_{l+2}$, then there exist a correct process
$p_i\in\mathsf{members}(R_l)$ and a round $a_i$ such that $p_i$ invokes
\[
  \texttt{InstallTransition}_i(R_l,R_{l+1},a_i).
\]
In particular, the set of install transition rounds of correct processes
in $\mathsf{members}(R_l)$ is nonempty.
\end{lemma}

\begin{proof}
Without an invocation of $\texttt{InstallTransition}$ by a correct old
member, blocks declared under $R_{l+1}$ can occupy only the
$f_{l+1}$ Byzantine $R_{l+1}$-member slots and, during an addition, the
joining process's single slot. These are fewer than $q_{l+1}$ distinct
round-broadcaster slots. The dealer's issuance rule
(\cref{sec: config begin}) requires it to observe, before issuing
$R_{l+2}$, at least $q_{l+1}$ certified blocks declared under $R_{l+1}$
in distinct broadcaster slots in one round; this
cannot happen unless a correct old member has invoked
$\texttt{InstallTransition}$ for $R_l\rightarrow R_{l+1}$.
\end{proof}

\begin{lemma}[If $t$ is the earliest correct old-member install transition
round for
$R_l\rightarrow R_{l+1}$, then
$\mathsf{earliest}(R_{l+2})\geq t+6$]
\label{lemma:following-reconfiguration-earliest-round}
Fix a valid transition $R_l\rightarrow R_{l+1}$ in an execution satisfying
the issuance-rule assumption, and suppose that the dealer issues
$R_{l+2}$. Let $t$ be the earliest install transition round among correct
processes in $\mathsf{members}(R_l)$ for $R_l\rightarrow R_{l+1}$.
This round exists by
\cref{lemma:following-reconfiguration-requires-installation}. Then
\[
  \mathsf{earliest}(R_{l+2})\geq t+6.
\]
\end{lemma}

\begin{proof}
A correct process in $\mathsf{members}(R_l)$ declares $R_{l+1}$ in a
block only after invoking $\texttt{FinishTransition}$, which occurs at
the beginning of round $a_i+4$ for its install transition round
$a_i\geq t$. Hence no correct old member produces an $R_{l+1}$-declared
block before round $t+4$. As in the proof of
\cref{lemma:following-reconfiguration-requires-installation}, certified
$R_{l+1}$-declared blocks in any round below $t+4$ therefore occupy at
most the $f_{l+1}$ Byzantine $R_{l+1}$-member slots and, during an
addition, the joining process's single slot; a correct joining process
may indeed produce its $R_{l+1}$-declared transition blocks before round
$t+4$. Since $f_{l+1}+1<q_{l+1}$, no round below $t+4$ contains
$q_{l+1}$ certified $R_{l+1}$-declared blocks in distinct
round-broadcaster slots. The first round in which the dealer can observe
$q_{l+1}$ such blocks is therefore at least $t+4$. The dealer's snapshot
rule sets $\mathsf{earliest}(R_{l+2})$ two rounds above the observed
round, so
$\mathsf{earliest}(R_{l+2})\geq t+6$.
\end{proof}

\begin{lemma}[Every round-$(t+1)$ block produced by a correct $R_l$-member is
marked for $R_l\rightarrow R_{l+1}$ at every correct process that
DAG-delivers it]
\label{lemma:marked-block-propagation}
Fix a valid transition $R_l\rightarrow R_{l+1}$ in an execution satisfying
the issuance-rule assumption and the configuration-level induction
hypothesis for $R_l$, and assume the conclusions already established for
preceding transitions. Let $t$ be the earliest install transition round
among correct processes in $\mathsf{members}(R_l)$.
Suppose BRB Agreement holds through round $t$. Let
\[
  X=\langle m,t+1,\mu_X,C_X\rangle
\]
be a certified block produced by a correct process in
$\mathsf{members}(R_l)$. At every correct process $p_i$ that DAG-delivers
$X$,
\[
  \mathsf{Marked}_i(R_l,R_{l+1},X)
\]
holds. Equivalently, for every such $p_i$, there exist a round $s$ and a
block $X'=\langle m',s,\mu',C'\rangle\in DAG_i[s]$ such that
\[
  X\leadsto_i X'
  \qquad\text{and}\qquad
  \mathsf{pending}(\mu')=R_{l+1}.
\]
\end{lemma}

\begin{proof}
An installation with install transition round $t$ is justified by $f_l+1$ marked
round-$t$ blocks. Each such block either records $R_{l+1}$ in its own
pending field or has an earlier ancestor that does. By
\cref{lemma:no-premature-pending}, no such pending field can occur below
round $\delta_{l+1}$, so $t\geq\delta_{l+1}$, and the Agreement premise
includes round $\delta_{l+1}$. Let $\mathcal M$ be the set of the $f_l+1$ marked
round-$t$ blocks that justify the earliest installation. By
\cref{lemma:predecessor-parent-slots}, these blocks occupy distinct
round-broadcaster slots in
$\mathsf{members}(R_l)$, and Agreement in round $t$ ensures that whenever
another correct process DAG-delivers a block in one of these slots, it
delivers the same certified block and its required ancestry; the block is
therefore marked there as well.

Let $X$ be a round-$(t+1)$ block produced by a correct member of
$\mathsf{members}(R_l)$. If its broadcaster installed in round $t$, the
pending field in $X$ contains $R_{l+1}$. Otherwise, its prescribed
$R_l$-declared $\texttt{DAG-Append}$ block has $q_l$ round-$t$ parents.
Their slots and $\mathcal M$ lie in the
same $n_l$-member configuration, so
\[
 |\mathcal H_X(t)\cap\mathcal M|
 \geq q_l+(f_l+1)-n_l=1.
\]
Thus $X$ has a marked ancestor and is itself marked.
\end{proof}

\begin{lemma}[Every correct $R_l$-member that advances to round $t+2$ has
install transition round $t$ or $t+1$ for
$R_l\rightarrow R_{l+1}$]
\label{lemma:lagging-process-installation}
Under the assumptions and notation of
\cref{lemma:marked-block-propagation}, let
$p_i\in\mathsf{members}(R_l)$ be correct. If
\[
  |\mathsf{ParentBlocks}_i(R_l,t+1)|\geq q_l
\]
and $p_i$ advances to round $t+2$, then it has invoked
\[
  \texttt{InstallTransition}_i(R_l,R_{l+1},a_i)
\]
for some $a_i\in\{t,t+1\}$. Moreover, every correct process in
$\mathsf{members}(R_l)$ that ever invokes
$\texttt{InstallTransition}$ for this transition has an install transition
round in $\{t,t+1\}$.
\end{lemma}

\begin{proof}
Consider a correct process that advances through round $t+1$ without
installing in round $t$. It obtains $q_l$ admissible round-$(t+1)$ blocks
in $\mathsf{ParentBlocks}_i(R_l,t+1)$. By
\cref{lemma:predecessor-parent-slots}, their broadcasters occupy distinct
slots in $\mathsf{members}(R_l)$. At most $f_l$ are Byzantine, so at least
\[
 q_l-f_l=n_l-2f_l\geq f_l+1
\]
were produced by correct broadcasters and are marked by
\cref{lemma:marked-block-propagation}. The process therefore invokes
$\texttt{InstallTransition}$ in round $t+1$. A process that invokes in a
later round necessarily advances through round $t+1$ first, so no correct
install transition round lies outside $\{t,t+1\}$.
\end{proof}

\begin{lemma}[If a correct $R_l$-member's round-$r$ enabled-configs value
contains $R_l$, every valid round-$r$ certificate is declared under $R_l$ or
$R_{l+1}$]
\label{lemma:transition-round-certificate-scope}
Fix a valid transition $R_l\rightarrow R_{l+1}$ in an execution satisfying
the issuance-rule assumption and the configuration-level induction
hypothesis for $R_l$, and assume the conclusions already established for
preceding transitions. Let
$\delta_{l+1}=\mathsf{earliest}(R_{l+1})$. Let
$r\geq\delta_{l+1}$, suppose BRB Agreement holds through round $r-1$,
and suppose there exists a correct process
$p_i\in\mathsf{members}(R_l)$ such that
\[
  \mathit{enabled\text{-}configs}_i[r]\neq\bot
  \qquad\text{and}\qquad
  R_l\in\mathit{enabled\text{-}configs}_i[r].
\]
Then every valid BRB quorum certificate for a complete BRB value
$\langle m,r,\mu\rangle$ satisfies
\[
  \mathsf{local}(\mu)\in\{R_l,R_{l+1}\}.
\]
\end{lemma}

\begin{proof}
There is no configuration earlier than $R_l$ to exclude when $l=1$. Suppose
$l>1$. By
\cref{lemma:well-spaced-predecessor-completion}, every active correct
$R_l$-member has completed the transition into $R_l$ before acting in round
$\delta_{l+1}-1$, and therefore does not enable an earlier configuration in
round $r$. For the case of a correct process removed by the immediately
preceding transition, let $t'$ be that transition's earliest correct install
transition round and let $a_h$ be the removed process's install transition
round. \Cref{corollary:install-transition-round-bound}, applied to that
preceding transition, gives
$a_h\leq t'+1$, while the bound established in the proof of
\cref{lemma:well-spaced-predecessor-completion} gives
$t'+5\leq\delta_{l+1}-1$. Hence
\[
  a_h+4\leq t'+5\leq\delta_{l+1}-1<r.
\]
If the removed process has already stopped, it contributes no share. If it
is still processing BRB messages, its removal procedure nevertheless assigns
no enabled-configs value after round $a_h+4$, so it cannot echo or vote in
round $r$. The conclusions for still earlier transitions exclude their
removed correct processes in the same way. Thus, under every configuration
$R_c$ with $c<l$, only its at most $f_c<q_c$ Byzantine members can contribute
round-$r$ vote shares, which are insufficient to form a quorum certificate.

Suppose $R_{l+2}$ exists. By
\cref{lemma:following-reconfiguration-requires-installation}, let $t$ be the
earliest install transition round for $R_l\rightarrow R_{l+1}$. By
\cref{lemma:following-reconfiguration-earliest-round},
$\mathsf{earliest}(R_{l+2})\geq t+6$. If $r\leq t+1$, then directly
$r<\mathsf{earliest}(R_{l+2})$. If $r\geq t+2$, Agreement through $r-1$
includes Agreement through $t$. Let $p_i$ be the process in the premise. If
$p_i$ invokes $\texttt{InstallTransition}$, then
\cref{lemma:marked-block-propagation,lemma:lagging-process-installation}
place its install transition round in $\{t,t+1\}$. Since
$\texttt{InstallTransition}$ with install transition round $a_i$ assigns the pair
exactly in rounds $a_i+1,\ldots,a_i+4$, its round-$r$ value containing
$R_l$ implies $r\leq t+5$. Otherwise, assigning a value for a
round $r\geq t+2$ requires $p_i$ to advance through round $t+1$; the same
lemmas would then force it to install by that round, a contradiction. Hence
\[
 r\leq t+5<\mathsf{earliest}(R_{l+2}).
\]
In either case, $r<\mathsf{earliest}(R_{l+2})$.
By \cref{lemma:earliest-install-lower-bound,%
lemma:following-reconfiguration-earliest-round}, the earliest-inclusion
rounds strictly increase along the issued configuration sequence.
Thus $r<\mathsf{earliest}(R_c)$ for every issued $R_c$ with $c\geq l+2$.
Correct processes reject a configuration before its authenticated
earliest-inclusion round, so none of these configurations can obtain
enough correct votes for a round-$r$ certificate.
If $R_{l+2}$ does not exist, authentication excludes a later
configuration directly.

Thus every valid round-$r$ certificate is declared under either $R_l$ or
$R_{l+1}$.
\end{proof}

\begin{lemma}[Agreement through round $r$ preserves round-$(r+1)$
enabled-configs alignment while a correct old member still enables $R_l$]
\label{lemma:next-round-configuration-alignment}
Fix a valid issued transition $R_l\rightarrow R_{l+1}$ under the
configuration-level induction hypothesis for $R_l$, let
$r\geq\delta_{l+1}=\mathsf{earliest}(R_{l+1})$, and suppose BRB
Agreement holds through round $r$. If a correct process
$p_i\in\mathsf{members}(R_l)$ has an assigned round-$(r+1)$ enabled-configs
value containing $R_l$, then all round-$(r+1)$ enabled-configs values assigned
by correct processes in $\mathsf{members}(R_l)$ are aligned.
\end{lemma}

\begin{proof}
We first show that no round-$(r+1)$ value assigned by a correct
$R_l$-member involves a configuration other than $R_l$ and $R_{l+1}$. By
\cref{lemma:well-spaced-predecessor-completion}, every correct
$R_l$-member completes the transition into $R_l$ before acting in round
$\delta_{l+1}-1$, so values involving $R_{l-1}$ are assigned only for
rounds below $\delta_{l+1}\leq r$; the conclusions for preceding
transitions exclude earlier configurations. If a value involving
$R_{l+2}$ or a later configuration is ever assigned, then $R_{l+2}$ is
issued, so by
\cref{lemma:following-reconfiguration-requires-installation} the
earliest install transition round $t$ for $R_l\rightarrow R_{l+1}$
exists, and by
\cref{lemma:following-reconfiguration-earliest-round,%
lemma:earliest-install-lower-bound} such a value is assigned only for
rounds at least $\mathsf{earliest}(R_{l+2})+1\geq t+7$. If $t>r$, then
$r+1\leq t<t+7$; if $t\leq r$, the last paragraph below shows
$r+1\leq t+5<t+7$. In every case, every assigned round-$(r+1)$ value is
$\{R_l\}$, $\{R_l,R_{l+1}\}$, or $\{R_{l+1}\}$.

It remains to exclude the coexistence of $\{R_l\}$ and $\{R_{l+1}\}$.
If no correct process has invoked $\texttt{InstallTransition}$ by round
$r$, no pair or new singleton has been assigned for round $r+1$, and
alignment is immediate. If the earliest install transition round is
$t=r$, a correct process has either not yet invoked
$\texttt{InstallTransition}$ and assigns $\{R_l\}$, or has invoked it
and assigns $\{R_l,R_{l+1}\}$ to round $r+1$; only the old singleton
and the pair coexist, and $r+1=t+1\leq t+5$.

Suppose $t\leq r-1$. Agreement through round $r$ includes Agreement through
$t$. By \cref{lemma:marked-block-propagation,lemma:lagging-process-installation},
every correct old member that advances through round $t+1$ invokes
$\texttt{InstallTransition}$ with install transition round $t$ or $t+1$.
The process in the premise has assigned a value for round $r+1\geq t+2$;
hence it has either already invoked $\texttt{InstallTransition}$ or advanced
through round $t+1$ and was forced to invoke it. Since its round-$(r+1)$
value still contains $R_l$, and an invocation with install transition round
$a_i$ assigns $\{R_l,R_{l+1}\}$ only to rounds
$a_i+1,\ldots,a_i+4$, we have $r+1\leq t+5$. Any correct process with an
assigned round-$(r+1)$ value that has advanced through round $t+1$
likewise has install transition round $t$ or $t+1$, so its value is
$\{R_l,R_{l+1}\}$ or $\{R_{l+1}\}$; one that has not advanced through
round $t+1$ either installed in round $t$ and assigned the pair through
round $t+4$, or has assigned no value beyond round $t+1<r+1$. Hence the
old and new singleton values cannot coexist. The assigned values are
aligned in every case.
\end{proof}

\begin{lemma}[Enabled-configs values remain aligned and BRB Agreement holds
while a correct $R_l$-member still enables $R_l$]
\label{lemma:brb-agreement-transition-rounds}
Fix a valid issued transition $R_l\rightarrow R_{l+1}$ under the
configuration-level induction hypothesis for $R_l$, and let
$\delta_{l+1}=\mathsf{earliest}(R_{l+1})$. For every logical round
$r\geq\delta_{l+1}$ such that some correct
$p_i\in\mathsf{members}(R_l)$ has an assigned round-$r$ enabled-configs
value containing $R_l$:
\begin{enumerate}
  \item the assigned round-$r$ enabled-configs values of correct processes
  in $\mathsf{members}(R_l)$ are aligned for
  $R_l\rightarrow R_{l+1}$; and
  \item BRB Agreement holds for every round-broadcaster slot $(r,ID)$.
\end{enumerate}
\end{lemma}

\begin{proof}
First observe that, if the lemma's premise holds in round $r$, it also holds
in every round from $\delta_{l+1}$ through $r$. It is enough to prove this
for the immediately preceding round. Let a correct $R_l$-member have an
assigned round-$r$ enabled-configs value containing $R_l$, where
$r>\delta_{l+1}$. If the process assigned this value while advancing from
round $r-1$ to round $r$, it had already assigned its round-$(r-1)$ value.
For one transition, a correct process's assigned values change only from the
old singleton to the pair and then to the new singleton. Because the
round-$r$ value still contains $R_l$, the round-$(r-1)$ value also contains
$R_l$. If $\texttt{InstallTransition}$ assigned the round-$r$ value in
advance, then $r$ is one of the next four rounds assigned the pair
$\{R_l,R_{l+1}\}$. Its immediately preceding value was either assigned the
same pair by that invocation or was assigned before the invocation and still
contained $R_l$. Repeating this argument reaches $\delta_{l+1}$.

Suppose the lemma is false, and let $r$ be the least round satisfying its
premise in which either conclusion fails. By the Agreement clause of the
configuration-level induction hypothesis for $R_l$, BRB Agreement holds through round
$\delta_{l+1}-1$. If $r=\delta_{l+1}$, alignment follows from
\cref{lemma:well-spaced-base-round}. If $r>\delta_{l+1}$, the property just
proved and the minimality of $r$ give BRB Agreement through round $r-1$.
Applying
\cref{lemma:next-round-configuration-alignment} to round $r-1$ therefore
gives alignment in round $r$.

In either case, BRB Agreement holds through round $r-1$. Hence
\cref{lemma:transition-round-certificate-scope} restricts every valid
round-$r$ certificate to $R_l$ or $R_{l+1}$, and
\cref{lemma:brb-aligned-safety} gives BRB Agreement in round $r$. Both
conclusions hold, contradicting the choice of $r$.
\end{proof}



\begin{lemma}[No correct $R_l$-member's enabled-configs value contains $R_l$
at or after $\mathsf{earliest}(R_{l+2})$]
\label{lemma:old-configuration-round-cutoff}
Fix a valid transition $R_l\rightarrow R_{l+1}$ in an execution satisfying
the issuance-rule assumption and the configuration-level induction hypothesis for $R_l$, and
suppose $R_{l+2}$ exists. For every correct process
$p_i\in\mathsf{members}(R_l)$ and every logical round $s$,
\[
  \mathit{enabled\text{-}configs}_i[s]\neq\bot
  \;\land\;
  R_l\in\mathit{enabled\text{-}configs}_i[s]
  \quad\Longrightarrow\quad
  s<\mathsf{earliest}(R_{l+2}).
\]
\end{lemma}

\begin{proof}
By \cref{lemma:following-reconfiguration-requires-installation}, an earliest
install transition round $t$ exists, and
\cref{lemma:following-reconfiguration-earliest-round} gives
$\mathsf{earliest}(R_{l+2})\geq t+6$. By
\cref{lemma:brb-agreement-transition-rounds,%
lemma:marked-block-propagation,lemma:lagging-process-installation},
every correct process in $\mathsf{members}(R_l)$ that advances from round
$t+1$ to round $t+2$ invokes $\texttt{InstallTransition}$ in round $t$ or
$t+1$. For such a process,
every assigned enabled-configs value containing $R_l$ has round index at
most $t+5$. A process that does not advance through round $t+1$ either
invoked $\texttt{InstallTransition}$ in round $t$, in which case that
invocation assigned values equal to $\{R_l,R_{l+1}\}$ only through round
$t+4$, or did not invoke it, in which case round advancement assigned no
enabled-configs value beyond round $t+1$.
Hence every such $s$ satisfies
$s\leq t+5<\mathsf{earliest}(R_{l+2})$.
\end{proof}

\begin{lemma}[BRB Agreement holds before $\mathsf{earliest}(R_{l+1})$ and in
every later round where a correct $R_l$-member's enabled-configs value
contains $R_l$]
\label{lemma:brb-safety-current-step}
Fix a valid transition $R_l\rightarrow R_{l+1}$ in an execution satisfying
the issuance-rule assumption and the configuration-level induction hypothesis for $R_l$, and let
$\delta_{l+1}=\mathsf{earliest}(R_{l+1})$. BRB Agreement holds for every
round-broadcaster slot whose round index $s$ satisfies either
\[
  s<\delta_{l+1},
\]
or
\[
\begin{aligned}
  &s\geq\delta_{l+1}
  \quad\text{and}\quad
  \exists\,p_i\in\mathsf{members}(R_l):\\
  &\quad p_i\text{ is correct},\quad
  \mathit{enabled\text{-}configs}_i[s]\neq\bot,\quad
  R_l\in\mathit{enabled\text{-}configs}_i[s].
\end{aligned}
\]
\end{lemma}

\begin{proof}
For every $s<\delta_{l+1}$, BRB Agreement is exactly the Agreement clause
of the configuration-level induction hypothesis for $R_l$. For every
$s\geq\delta_{l+1}$ satisfying the second condition in the statement,
\cref{lemma:brb-agreement-transition-rounds} gives BRB Agreement in every
round-broadcaster slot of round $s$.
\end{proof}

\subsubsection{DAG Progress and Round Advancement under BRB Agreement}

Having established BRB Agreement for the rounds considered above, we next
separate certification from advancement. First, if the correct processes
needed by a round-$r$ BRB instance continue those tasks until certificates
form, the resulting certified blocks and their ancestors reach every correct
recipient that remains active while waiting. Second, during a transition, a
correct process in round $r_0+2$ or $r_0+3$ can also advance after receiving
$q_{l+1}$ admissible blocks from the following round. The transition-progress
proofs below show that one of these two ways of completing each affected
$\texttt{AdvanceRound}$ iteration occurs before any correct process on which
that progress depends stops its BRB and DAG tasks.

\begin{lemma}[A correct process's local DAG is ancestor-closed]
\label{lemma:local-dag-ancestor-closure}
At every correct process $p_i$, if $B\in DAG_i[r]$, then every parent and
ancestor named by $B$ is also in $DAG_i$.
\end{lemma}

\begin{proof}
Initially the DAG is empty. When the certified block $B$ is received outside
state-transfer, $\texttt{FetchAndDeliver}_i(B)$ invokes
$\texttt{AddAndDeliver}_i(B)$
only after every named parent has been inserted; its recursive parent tasks
apply the same rule. When a joining process $p_i$ receives a state-reply
payload $D$ for $R_{l+1}$,
$\mathsf{ValidState}_i(D,R_{l+1})$ accepts only an ancestor-closed set
containing at most one block in each round-broadcaster slot, and
$\texttt{MergeState}_i(D)$ invokes $\texttt{AddAndDeliver}_i(B)$ for accepted
blocks $B$ in increasing round order. These are the only insertion cases,
so induction over insertions proves the claim.
\end{proof}

\begin{lemma}[A certified block declared under $R_c$ with $q_c$ parents has
an admissible parent quorum]
\label{lemma:ordinary-block-parent-quorum}
Let $p_i$ be correct and let
$X=\langle m,s+1,\mu,C\rangle\in DAG_i[s+1]$ be admissible. Put
$R_c=\mathsf{local}(\mu)$ and suppose
$|\mathsf{parents}(\mu)|=q_c$. Then
\[
 |\{B\in DAG_i[s]:\mathsf{ref}(B)\in\mathsf{parents}(\mu)\}|=q_c,
\]
and every block in this set belongs to
$\mathsf{ParentBlocks}_i(R_c,s)$.
\end{lemma}

\begin{proof}
The displayed premise gives exactly $q_c$ parent references, and
$\mathsf{CertifiedParentCount}(s+1,\mu)$ ensures that their broadcaster
identities are distinct. Local DAG ancestor closure puts the matching blocks
in $DAG_i[s]$, and
$\mathsf{ValidParentSet}_i(X,DAG_i)$ applies exactly the broadcaster and
configuration-number checks in $\mathsf{ParentBlocks}_i(R_c,s)$.
\end{proof}

\begin{lemma}[Transferred state is validated independently of its sender]
\label{lemma:sender-independent-state-validation}
Let $p_i$ be correct and let $D$ be a state-reply payload for which
$\mathsf{ValidState}_i(D,R_{l+1})$ holds. Every block
$B=\langle m,s,\mu,C\rangle$ that
$\texttt{MergeState}_i(D)$ inserts carries a valid BRB certificate for its
complete payload, satisfies $\mathsf{CertifiedParentCount}(s,\mu)$ and the
referenced-parent checks, and is inserted together with all of its
ancestors. Moreover, $D$ contains at most one block in every
round-broadcaster slot. These conclusions do not depend on whether the
process that sent $D$ is correct.
\end{lemma}

\begin{proof}
All stated properties are checked from $D$ and the certificates bound to
its blocks. The predicate $\mathsf{ValidState}_i(D,R_{l+1})$ checks that $D$
is ancestor-closed and contains at most one block in every round-broadcaster
slot. For each block
$B=\langle m,s,\mu,C\rangle$, it also checks
$\mathsf{ValidCertifiedMetadata}_i(B)$, including
$\mathsf{CertifiedParentCount}(s,\mu)$, and every parent's hash and identity.
It rejects a block if its metadata names a configuration numbered later
than $l+1$. The reply-construction rule applies the same bound.
$\texttt{MergeState}_i(D)$ processes accepted blocks in
increasing round order, so every named parent is inserted before its child.
The sender's identity is not an input to any of these checks.
\end{proof}

\begin{lemma}[The known-identity set contains the members of every processed
configuration]
\label{lemma:known-coverage}
If $R_h$ is the highest-numbered authenticated reconfiguration message
processed by $\texttt{UpdateKnown}_i(R_h)$ at a correct process $p_i$, then
\[
  \mathit{known}_i=\bigcup_{c=1}^{h}\mathsf{members}(R_c),
\]
and this set never shrinks.  In particular, upon receiving a valid
$R_{l+1}$, every correct recipient in
$\mathsf{members}(R_l)\cup\mathsf{members}(R_{l+1})$ knows every identity
in that union before the retained-process handler executes its wait for
$\mathit{local}_i=R_l$ and
$r_i\geq\mathsf{earliest}(R_{l+1})$, or the joining process begins state
transfer.
\end{lemma}

\begin{proof}
Each authenticated $R_h$ carries the identities that appeared in
$R_1,\ldots,R_{h-1}$. The $\texttt{UpdateKnown}_i(R_h)$ procedure
authenticates that carried set and unions it and $\mathsf{members}(R_h)$ into
$\mathit{known}_i$. Processing $R_h$ through
$\texttt{UpdateKnown}_i(R_h)$ therefore
gives the displayed equality,
and subsequent invocations only enlarge the set. Both handlers invoke the
procedure before those actions.
\end{proof}

\begin{lemma}[Conditional DAG delivery of certified causal history to a
correct process]
\label{lemma:certified-history-dissemination}
Let $p_h$ and $p_j$ be correct, and let $\mathcal D$ be a finite
ancestor-closed subset of $DAG_h$. Because every block has finitely many
parents, all in the preceding round, the causal history of a finite set of
blocks is finite; by \cref{lemma:local-dag-ancestor-closure}, the causal
history of blocks stored by $p_h$ is contained in $DAG_h$. Suppose that
$p_j\in\mathit{known}_h$ while $p_h$ still performs its prescribed
communication and DAG tasks. Then $p_h$ sends every block in
$\mathcal D$ to $p_j$; if $p_h$ is removed, it does so before stopping its
remaining tasks.

Suppose additionally that $p_j$ is active or eventually becomes active and,
while any block in $\mathcal D$ is missing, continues processing
certified-block messages and running $\texttt{FetchAndDeliver}$; that every
$B\in\mathcal D$ satisfies
$\mathsf{ValidCertifiedMetadata}_j(B)$ and the referenced-parent checks once
its parents are present; and that no different block already in $DAG_j$
occupies the same round-broadcaster slot as a block in $\mathcal D$. Then
$p_j$ eventually DAG-delivers every block in $\mathcal D$.

For a correct removed process $p_j$, the second conclusion applies only
before it leaves. An application to such a process must show that the
condition for leaving keeps these receive and $\texttt{FetchAndDeliver}$
actions active
until the blocks used to satisfy that condition and their ancestors have
been DAG-delivered. No delivery is claimed after $p_j$ leaves.
\end{lemma}

\begin{proof}
Fix $X\in\mathcal D$. If $p_j$ is already in $\mathit{known}_h$ when $X$
is inserted, $\texttt{AddAndDeliver}_h(X)$ sends the complete certified block
to $p_j$. Otherwise, let $R_c$ be the authenticated reconfiguration message
processed by the first $\texttt{UpdateKnown}_h(R_c)$ invocation that adds
$p_j$ to $\mathit{known}_h$. That invocation executes a finite
$\texttt{FloodDAG}_h$ over the blocks then stored at $p_h$, and therefore
sends $X$ if $X$ was inserted before that invocation. The invocation and
the insertion of $X$ are ordered local events: if the insertion occurs
first, $\texttt{FloodDAG}_h$ sends $X$; if it occurs second,
$\texttt{AddAndDeliver}_h(X)$ sends $X$. Repeated invocations of
$\texttt{FloodDAG}_h$ provide additional copies while $p_h$ remains active,
and, if $p_h$ is removed, its final invocation sends every block then stored
before it stops. Thus every block in $\mathcal D$ is sent to $p_j$ before a
removed $p_h$ stops its remaining tasks. Reliable delivery is unaffected by
the sender's later departure. A copy sent before $p_j$ activates remains
queued and is processed after $p_j$ activates.

Assume now the additional premises for delivery at $p_j$. We prove the
following statement by induction on $s\geq1$: every block in $\mathcal D$
whose round is at most $s$ is eventually inserted into $DAG_j$. For $s=1$,
let $X\in\mathcal D$ be a round-$1$ block. It has no parents. Upon receipt,
$p_j$ either already has a live $\texttt{FetchAndDeliver}_j(X)$ task or
schedules one. For the fixed block $X$, membership in
$\mathit{scheduled}_j$ cannot refer only to a task that returned after a
failed check: $\mathsf{ValidCertifiedMetadata}_j(X)$ is deterministic and true
by assumption; a task with a missing parent waits instead of returning; and,
once all parents are present, the assumed referenced-parent and slot
conditions permit insertion unless $X$ is already present. The assumed
checks imply
$\mathsf{Admissible}_j(X)$, so the procedure invokes
$\texttt{AddAndDeliver}_j(X)$, which inserts and DAG-delivers the block unless
it is already present.

For the induction step, let $X\in\mathcal D$ be a round-$s$ block with
$s>1$, and assume the statement through round $s-1$. Because $\mathcal D$
is ancestor-closed, every parent referenced by $X$ belongs to $\mathcal D$
and lies in round $s-1$. The induction hypothesis therefore inserts all of
these parents into $DAG_j$. Once the last missing parent is inserted, each
parent condition tested by $\texttt{FetchAndDeliver}_j(X)$ is true. Because
$p_j$ continues that task, the procedure resumes, verifies the now-present
parents, and obtains $\mathsf{Admissible}_j(X)$. The round-broadcaster-slot
premise then permits $\texttt{AddAndDeliver}_j(X)$ to insert and DAG-deliver
$X$, unless it was already present. This proves the induction statement.
Since $\mathcal D$ is finite, it has a greatest round, so the statement for
that round covers every block in $\mathcal D$. Blocks may arrive in any
order because each $\texttt{FetchAndDeliver}_j(X)$ task waits until the
parents of $X$ are present.
\end{proof}

\begin{lemma}[A correct broadcaster's round-$r$ block is certified and
delivered while the required processes remain active]
\label{lemma:correct-block-dissemination}
Fix a configuration $R_l$ and a logical round $r$ for which BRB Agreement
holds for every round-broadcaster slot in rounds at most $r$. Let
$\mathcal G\subseteq\mathsf{members}(R_l)$ contain at least $q_l$ active
correct processes. Suppose every $p_i\in\mathcal G$ enters round $r$, invokes
its prescribed round-$r$ $\texttt{DAG-Append}$ operation with frozen value
$\langle m_i,r,\mu_i\rangle$, and reaches that invocation's call to
$\texttt{BRB-Broadcast}$. At that call, suppose it satisfies
\[
  \mathsf{local}(\mu_i)=R_l
  \qquad\text{and}\qquad
  R_l\in\mathit{enabled\text{-}configs}_i[r],
\]
and suppose $\mathsf{MaySendBRB}_i(R_l,r)$ holds, so the call passes its
guard.
Fix $p_k\in\mathcal G$, and let $\langle m_k,r,\mu_k\rangle$ be the
complete BRB value proposed by its $\texttt{DAG-Append}$ invocation. If,
while no correct process in $\mathcal G$ has formed a certificate for this
value, every process in $\mathcal G$ remains active and continues its
prescribed BRB tasks, then a valid certificate $C_k$ declared under $R_l$
is eventually formed. Moreover, the
resulting block
\[
B_k=\langle m_k,r,\mu_k,C_k\rangle
\]
is eventually DAG-delivered by every correct process
$p_j\in\mathsf{members}(R_l)$ such that, while $B_k$ is missing, $p_j$ remains active
and continues processing certified-block messages and running
$\texttt{FetchAndDeliver}$.
\end{lemma}

\begin{proof}
Let the fixed correct broadcaster propose
$\langle m_k,r,\mu_k\rangle$. Its declared configuration is the frozen
$R_l$, even if the broadcaster installs another configuration before this
delayed call. Since
\[
  \mathcal G\subseteq\mathsf{members}(R_l)
  \subseteq\mathsf{Recipients}(\mathsf{enabled}(\mu_k)),
\]
every process in $\mathcal G$ receives the proposal, echo, and vote messages
for this BRB instance. Every $p_i\in\mathcal G$ passes the stated
$\mathsf{MaySendBRB}_i(R_l,r)$ guard and continues the prescribed BRB
handlers. The required echoes and votes are therefore sent to every process
in $\mathsf{Recipients}(\mathsf{enabled}(\mu_k))$. Eventual message delivery
causes a correct process $p_h\in\mathcal G$ to receive $q_l$ valid votes and
form $C_k$. The Agreement premise makes this the unique BRB value for the
slot.

In the vote-message receive handler, $p_h$ constructs $B_k$ and sends
$\langle\mathsf{certified\text{-}block},B_k\rangle$ to every identity in
$\mathit{known}_h$ before triggering
$\texttt{BRB-Deliver}_h(m_k,r,\mu_k,C_k)$. Since $p_h$ has processed
$R_l$, \cref{lemma:known-coverage} gives
$\mathsf{members}(R_l)\subseteq\mathit{known}_h$.
Thus the complete certified block is sent to every $R_l$-member before
$p_h$ can leave, and its later departure does not prevent delivery.

Before invoking $\texttt{BRB-Broadcast}$, $p_k$ stored every block named by
$\mathsf{parents}(\mu_k)$ and, by
\cref{lemma:local-dag-ancestor-closure}, every ancestor of those blocks.
For each correct $R_l$-member that remains active while any of these blocks
is missing, \cref{lemma:certified-history-dissemination} delivers this finite
ancestor-closed set from $p_k$; if $p_k$ is removed, its final
$\texttt{FloodDAG}$ sends the set before it stops. The Agreement premise
excludes a different block in any relevant round-broadcaster slot.
Consequently, if a correct $R_l$-member $p_x$ receives $B_k$, it obtains all
of its named parents, satisfies $\mathsf{Admissible}_x(B_k)$, and DAG-delivers
$B_k$. Therefore, the $\mathsf{certified\text{-}block}$ message sent by
$p_h$ when it forms $C_k$ makes every correct $R_l$-member that remains
active while $B_k$ is missing DAG-deliver $B_k$.
\end{proof}

\begin{lemma}[Delivery of $q_{l+1}$ admissible round-$(r+1)$ blocks advances
a transitioning process from round $r$]
\label{lemma:transition-lookahead-delivery}
Let $p_j$ be correct with
$\mathit{transition}_j=\langle R_l,R_{l+1},r_0\rangle$, and let
$r\in\{r_0+2,r_0+3\}$. Suppose BRB Agreement holds through round $r+1$.
Let $p_h$ be correct and suppose that $DAG_h[r+1]$ contains a set
$\mathcal Q\subseteq\mathsf{ParentBlocks}_h(R_{l+1},r+1)$ of at least
$q_{l+1}$ blocks in distinct round-broadcaster slots. Suppose
$p_j\in\mathit{known}_h$ while $p_h$ still performs its prescribed
communication and DAG tasks. If, while its round-$r$
$\texttt{AdvanceRound}$ iteration is pending, $p_j$ remains active and
continues processing certified-block messages and running
$\texttt{FetchAndDeliver}$, then $p_j$ eventually advances to round $r+1$.
\end{lemma}

\begin{proof}
Let $\mathcal D$ contain $\mathcal Q$ and all of their ancestors. It is a
finite ancestor-closed subset of $DAG_h$. The blocks have valid certified
metadata and referenced parents, and Agreement through round $r+1$ excludes
a different block in any of their slots. Thus
\cref{lemma:certified-history-dissemination} makes $p_j$ DAG-deliver
$\mathcal D$ while the iteration remains pending. Because every block in
$\mathcal Q$ is then present in $DAG_j[r+1]$, and the remaining conditions
in $\mathsf{ParentBlocks}_j(R_{l+1},r+1)$ depend only on the certified
metadata, $\mathcal Q\subseteq\mathsf{ParentBlocks}_j(R_{l+1},r+1)$. Hence
\[
 |\mathsf{ParentBlocks}_j(R_{l+1},r+1)|\geq q_{l+1},
\]
which is the permitted additional transition-round condition in round $r$.
The iteration
therefore terminates and increments $r_j$.
\end{proof}

\begin{lemma}[Continued activity delivers a round-$r$ parent quorum to every
correct $R_l$-member]
\label{lemma:dag-quorum-dissemination}
Fix a configuration $R_l$ and a logical round $r$ for which BRB Agreement
holds for every round-broadcaster slot in rounds at most $r$. Let
$\mathcal G\subseteq\mathsf{members}(R_l)$ contain at least $q_l$ correct
processes. Suppose every $p_k\in\mathcal G$ enters round $r$, invokes its
prescribed round-$r$ $\texttt{DAG-Append}$ operation, and reaches that
invocation's call to $\texttt{BRB-Broadcast}$ with proposed metadata $\mu_k$
satisfying $\mathsf{local}(\mu_k)=R_l$,
$R_l\in\mathit{enabled\text{-}configs}_k[r]$, and
$\mathsf{MaySendBRB}_k(R_l,r)$, so the call passes its guard. Suppose
moreover that, while
a certificate for any of these selected round-$r$ values is missing, every
process in $\mathcal G$ remains active and continues its prescribed BRB
tasks.

Let $p_j\in\mathsf{members}(R_l)$ be correct. If, while fewer than $q_l$ of
the resulting blocks are present in $DAG_j[r]$, $p_j$ remains active and
continues its prescribed BRB receive task, processing certified-block
messages, and running $\texttt{FetchAndDeliver}$, then $p_j$ eventually
DAG-delivers at least $q_l$ admissible round-$r$ blocks in distinct
round-broadcaster slots. Hence
\[
  |\mathsf{ParentBlocks}_j(R_l,r)|\geq q_l.
\]
\end{lemma}

\begin{proof}
Choose any $q_l$ processes in $\mathcal G$. For each chosen process, the
certification premises of \cref{lemma:correct-block-dissemination} hold:
the selected processes are correct eligible $R_l$-members, their prescribed
appends reach $\texttt{BRB-Broadcast}$, and they continue the BRB tasks
until every selected certificate is formed. The condition above that $p_j$
remains active and continues its receive and fetch tasks is exactly the
corresponding condition in
\cref{lemma:correct-block-dissemination}. Apply that lemma once to each
selected broadcaster. Process $p_j$ eventually DAG-delivers all $q_l$
blocks. They occupy distinct round-broadcaster slots, and their metadata
make them members of $\mathsf{ParentBlocks}_j(R_l,r)$, proving the displayed
condition.
\end{proof}

\begin{corollary}[An ordinary parent quorum advances a correct process when
the additional transition-round condition is unavailable]
\label{corollary:round-advancement}
Fix $R_l$, $r$, and a set
$\mathcal G\subseteq\mathsf{members}(R_l)$ of at least $q_l$ correct
processes. Suppose BRB Agreement holds through round $r$ and every process
in $\mathcal G$ reaches and passes the $\texttt{BRB-Broadcast}$ call of its
prescribed $R_l$-declared round-$r$ append, with
$\mathsf{MaySendBRB}_k(R_l,r)$ holding for every $p_k\in\mathcal G$. Let
$p_j\in\mathsf{members}(R_l)$ be
correct and executing a round-$r$ $\texttt{AdvanceRound}$ iteration whose
ordinary condition is
\[
  |\mathsf{ParentBlocks}_j(R_l,r)|\geq q_l
\]
and suppose that, throughout these round-$r$ iterations, every
$p_x\in\mathcal G\cup\{p_j\}$ satisfies
\[
 \nexists R_a,R_b,r_x:\quad
 \mathit{transition}_x=\langle R_a,R_b,r_x\rangle
 \ \land\ r\in\{r_x+2,r_x+3\}.
\]
Then $p_j$ eventually advances to round $r+1$.
\end{corollary}

\begin{proof}
Suppose, for contradiction, that $p_j$ never advances. The additional
condition is unavailable, so its ordinary condition remains false. The
process therefore cannot reach a later departure point and remains active,
processing certified blocks and running $\texttt{FetchAndDeliver}$.

If some $p_h\in\mathcal G$ advances from round $r$, the additional
condition is unavailable there as well, so $p_h$ stores $q_l$ blocks counted
by $\mathsf{ParentBlocks}_h(R_l,r)$. Let $\mathcal D$ contain those blocks
and their ancestors. The set is finite and ancestor-closed. Both processes
have authenticated $R_l$, hence $p_j\in\mathit{known}_h$; Agreement through
round $r$ excludes a conflicting block in any selected slot. The sends made
upon insertion, subsequent $\texttt{FloodDAG}$ invocations, and the final
invocation if $p_h$ is later removed satisfy the sending premise of
\cref{lemma:certified-history-dissemination}. That lemma makes $p_j$
DAG-deliver the selected quorum, a contradiction.

Otherwise no process in $\mathcal G$ advances from round $r$. None can
reach a later departure point, so all remain active and continue their BRB
tasks while their selected certificates are missing. The condition in
\cref{lemma:dag-quorum-dissemination} that these processes continue the
required tasks therefore holds, and that lemma
makes $p_j$ DAG-deliver $q_l$ admissible round-$r$ blocks, again a
contradiction.
\end{proof}

\subsubsection{Install-Transition-Round Bounds and BRB Agreement for a
Correct Joining Process}

\begin{corollary}[If $t$ is the minimum install transition round among
correct old members, every correct old member has install transition round
$t$ or $t+1$]
\label{corollary:install-transition-round-bound}
Fix a valid issued transition $R_l\rightarrow R_{l+1}$ and assume the
configuration-level induction hypothesis for $R_l$ defined above. If $t$ is
the earliest install transition round among correct
processes in $\mathsf{members}(R_l)$, then every such process eventually
invokes $\texttt{InstallTransition}$ in round $t$ or $t+1$.
\end{corollary}

\begin{proof}
By the configuration-level induction hypothesis for $R_l$, there is a round $s_0$
that every active correct $R_l$-member eventually enters and in which at
least $q_l$ active correct members invoke the prescribed append and reach
its $\texttt{BRB-Broadcast}$ call. If
$s_0>t+1$, sequential round execution means that every correct old member
has already advanced through round $t+1$ before entering $s_0$. If
$s_0\leq t+1$, induct from $s_0$ through $t+1$. In each such
round, at least $q_l$ active correct old members invoke the prescribed
append and reach its broadcast call. Indeed, the first transition is
installed only while round $t$ is completed, and its additional advancement
condition is unavailable through round $t+1$; ordinary entry into the next
round supplies the same predecessor quorum on which the next append waits.
For rounds below $\delta_{l+1}$, BRB Agreement is supplied by the
Agreement clause of the induction hypothesis; at and above
$\delta_{l+1}$, it is supplied by
\cref{lemma:brb-agreement-transition-rounds}. Therefore
\cref{corollary:round-advancement} makes every correct old member enter the
next round.

Thus every correct old member reaches round $t+2$ after collecting the
required round-$(t+1)$ quorum. By
\cref{lemma:lagging-process-installation}, its install transition round
belongs to $\{t,t+1\}$.
\end{proof}

For the next results, suppose the transition from $R_l$ to $R_{l+1}$ adds a
correct process $p_j$. Let $u$ be the round returned by
$\texttt{DiscoverBegin}_j(R_l,R_{l+1})$, and let $t$ be the earliest install
transition round among correct processes in $\mathsf{members}(R_l)$.

\begin{lemma}[The minimum install transition round among correct old members
is at most $u+1$]
\label{lemma:joining-install-upper-bound}
Fix a valid transition $R_l\rightarrow R_{l+1}$ in an execution satisfying
the issuance-rule assumption and the configuration-level induction hypothesis for $R_l$, and
suppose it adds a correct process $p_j$. Let $u$ be the round returned by
$\texttt{DiscoverBegin}_j(R_l,R_{l+1})$, and assume that $t$ is the
earliest round in which a correct process in $\mathsf{members}(R_l)$
invokes $\texttt{InstallTransition}$ for this transition. Then
\[
  t\leq u+1.
\]
\end{lemma}

\begin{proof}
Let $\mathcal M$ be any $f_l+1$ marked round-$u$ blocks present in
$DAG_j[u]$ when $\texttt{DiscoverBegin}$ returns. Suppose $t>u+1$.
A round-$(u+1)$ block $X$
produced by a correct process in $\mathsf{members}(R_l)$ is then a block
from the prescribed $R_l$-declared $\texttt{DAG-Append}$ operation with
$q_l$ round-$u$ parents. By
\cref{lemma:predecessor-parent-slots}, their slots and $\mathcal M$ lie in
the same $n_l$-member configuration, so
\[
 |\mathcal H_X(u)\cap\mathcal M|
 \geq q_l+(f_l+1)-n_l=1.
\]
BRB Agreement in round $u$, from
\cref{lemma:brb-safety-current-step}, identifies the intersecting slot with
the same marked block. Every such $X$ is therefore marked.

Let $s_0$ be the logical round whose existence is required by clause~(iv)
of the configuration-level induction hypothesis for $R_l$. If $s_0>u+1$,
every correct old member advances
through round $u+1$ before entering $s_0$. If $s_0\leq u+1$, the
supposition $t>u+1$ lets us induct from $s_0$ through $u+1$: in every such
round at least $q_l$ correct old members invoke their prescribed
$\texttt{DAG-Append}$ operation and reach its broadcast call, BRB Agreement
holds by
\cref{lemma:brb-safety-current-step}, and
\cref{corollary:round-advancement} advances every correct old member. Any
such process obtains $q_l$ round-$(u+1)$ blocks,
at least $q_l-f_l\geq f_l+1$ of which have correct broadcasters and are
marked. It invokes
$\texttt{InstallTransition}$, contradicting $t>u+1$.
\end{proof}

\begin{lemma}[The minimum install transition round among correct old members
is at least $u-1$]
\label{lemma:joining-install-lower-bound}
Fix a valid transition $R_l\rightarrow R_{l+1}$ in an execution satisfying
the issuance-rule assumption and the configuration-level induction hypothesis for $R_l$, and
suppose it adds a correct process $p_j$. Let $u$ be the round returned by
$\texttt{DiscoverBegin}_j(R_l,R_{l+1})$, and assume that $t$ is the
earliest round in which a correct process in $\mathsf{members}(R_l)$
invokes $\texttt{InstallTransition}$ for this transition. Then
\[
  t\geq u-1.
\]
\end{lemma}

\begin{proof}
Suppose $t\leq u-2$. Agreement through round $t+1$ follows from the
Agreement clause of the configuration-level induction hypothesis for $R_l$ and
\cref{lemma:brb-agreement-transition-rounds}. By
\cref{lemma:earliest-install-lower-bound},
$t\geq\mathsf{earliest}(R_{l+1})$. A correct
old member that installs in round $t$ has $R_l$ in its assigned
enabled-configs values for rounds $t+1$ and $t+2$, so the premises of
\cref{lemma:transition-round-certificate-scope} hold in both rounds. By
\cref{lemma:marked-block-propagation}, every
round-$(t+1)$ block produced by a correct process in
$\mathsf{members}(R_l)$ is marked.

When $\texttt{DiscoverBegin}$ returns, round $u$ satisfies its two local
DAG conditions. Since $u\geq t+2$, choose any block counted in
$\mathsf{LocalBlocks}_j(R_l,u)$ and repeatedly follow a parent in the
immediately preceding round until reaching a block $Z\in DAG_j[t+2]$.
Such a parent exists at every step because every imported round-$s'$ block
with metadata $\mu'$ passes
$\mathsf{CertifiedParentCount}(s',\mu')$, and
\cref{lemma:local-dag-ancestor-closure} places every chosen parent in the
joining process's DAG.

The joining process has produced no block before discovery returns. By
\cref{lemma:transition-round-certificate-scope}, $Z$ declares $R_l$ or
$R_{l+1}$ and has at least $q_l$ distinct round-$(t+1)$ parents. Let
$B=\langle m_B,t+1,\mu_B,C_B\rangle$ be one of these parents. If
$\mathsf{local}(\mu_B)=R_l$, then
$\mathsf{ValidCertifiedMetadata}_j(B)$ gives
$\mathsf{ID}(\mu_B)\in\mathsf{members}(R_l)$. Otherwise,
\cref{lemma:transition-round-certificate-scope} in round $t+1$ gives
$\mathsf{local}(\mu_B)=R_{l+1}$, and
$\mathsf{ValidCertifiedMetadata}_j(B)$ gives
$\mathsf{ID}(\mu_B)\in\mathsf{members}(R_{l+1})$. This identity also belongs
to $\mathsf{members}(R_l)$ because the sole added process has not yet
produced a block. At least
$q_l-f_l\geq f_l+1$ of these parent broadcasters are correct, and their
blocks are marked by \cref{lemma:marked-block-propagation}. Ancestor closure
places all these parents in $DAG_j[t+1]$. Thus round $t+1$ already satisfies
both conditions checked by $\texttt{DiscoverBegin}$. The minimality of $u$
gives $u\leq t+1$, contradicting $t\leq u-2$.
\end{proof}

\begin{lemma}[The install transition round of the joining process differs by at most one
from the minimum install transition round among correct old members]
\label{lemma:joining-install-round-cases}
Fix a valid transition $R_l\rightarrow R_{l+1}$ in an execution satisfying
the issuance-rule assumption and the configuration-level induction hypothesis for $R_l$, and
suppose it adds a correct process $p_j$. Let $u$ be the round returned by
$\texttt{DiscoverBegin}_j(R_l,R_{l+1})$, and assume that $t$ is the
earliest round in which a correct process in $\mathsf{members}(R_l)$
invokes $\texttt{InstallTransition}$ for this transition. Then
\[
  t\in\{u-1,u,u+1\}.
\]
\end{lemma}

\begin{proof}
The claim follows from
\cref{lemma:joining-install-upper-bound,lemma:joining-install-lower-bound}.
\end{proof}

\begin{lemma}[Every correct old member's install transition round differs by
at most one from the install transition round of the joining process]
\label{lemma:joining-old-member-install-rounds}
Fix a valid transition $R_l\rightarrow R_{l+1}$ in an execution satisfying
the issuance-rule assumption and the configuration-level induction hypothesis for $R_l$, and
suppose it adds a correct process $p_j$. Let $u$ be the round returned by
$\texttt{DiscoverBegin}_j(R_l,R_{l+1})$, and let $t$ be the earliest
install transition round among correct processes in
$\mathsf{members}(R_l)$. Then every correct
$p_i\in\mathsf{members}(R_l)$ eventually invokes
$\texttt{InstallTransition}$ for this transition. Let $a_i$ denote the
round of that invocation. Then
\[
\begin{array}{rcl}
t=u-1 &\Longrightarrow& a_i\in\{u-1,u\},\\
t=u   &\Longrightarrow& a_i\in\{u,u+1\},\\
t=u+1 &\Longrightarrow& a_i=u+1.
\end{array}
\]
\end{lemma}

\begin{proof}
The first two cases follow from
\cref{corollary:install-transition-round-bound}. Suppose $t=u+1$. Every
correct round-$(u+1)$ block has $q_l$ round-$u$ parents. Their slots
intersect the $f_l+1$ marked round-$u$ slots used by
$\texttt{DiscoverBegin}$, and round-$u$ BRB Agreement identifies an
intersecting slot with the same marked block. Hence every such
round-$(u+1)$ block is marked.
Every correct old member eventually reaches that round. Among the $q_l$
blocks it obtains before advancing, at least
$q_l-f_l\geq f_l+1$ have correct broadcasters and are therefore marked.
It consequently installs while processing round $u+1$, so $a_i=u+1$.
\end{proof}

\begin{lemma}[Every correct old member has enabled-configs value
$\{R_l,R_{l+1}\}$ in rounds $u+2$ and $u+3$]
\label{lemma:joining-transition-round-bound}
Fix a valid transition $R_l\rightarrow R_{l+1}$ in an execution satisfying
the issuance-rule assumption and the configuration-level induction hypothesis for $R_l$, and
suppose it adds a correct process $p_j$. Let $u$ be the round returned by
$\texttt{DiscoverBegin}_j(R_l,R_{l+1})$, and let $t$ be the earliest
install transition round among correct processes in
$\mathsf{members}(R_l)$. Then every correct
$p_i\in\mathsf{members}(R_l)$ satisfies
\[
  \mathit{enabled\text{-}configs}_i[u+2]
  =\mathit{enabled\text{-}configs}_i[u+3]
  =\{R_l,R_{l+1}\}.
\]
The joining protocol prescribes
$\texttt{DAG-Transition-Append}_j$ in precisely these two rounds.
\end{lemma}

\begin{proof}
By \cref{lemma:joining-install-round-cases,%
lemma:joining-old-member-install-rounds}, a correct process's install
transition round $a_i$ belongs to $\{u-1,u,u+1\}$. The four rounds assigned
the pair are:
\[
\begin{array}{c|c}
a_i & \text{rounds }s\text{ with }
\mathit{enabled\text{-}configs}_i[s]=\{R_l,R_{l+1}\}\\
\hline
u-1 & u,u+1,u+2,u+3\\
u   & u+1,u+2,u+3,u+4\\
u+1 & u+2,u+3,u+4,u+5
\end{array}
\]
In every case, the value in rounds $u+2$ and $u+3$ is
$\{R_l,R_{l+1}\}$.
\end{proof}

\begin{lemma}[If the joining process enables both configurations after
correct old members stop voting under $R_l$, BRB Agreement holds in that
round]
\label{lemma:brb-safety-joining-round}
Fix a valid transition $R_l\rightarrow R_{l+1}$ in an execution satisfying
the issuance-rule assumption and the configuration-level induction hypothesis for $R_l$, and
suppose it adds a correct process $p_j$. Let $u$ be the round returned by
$\texttt{DiscoverBegin}_j(R_l,R_{l+1})$; the joining protocol uses $u$ as
$p_j$'s install transition round. Let $s$ be a logical round such that
\[
  \mathit{enabled\text{-}configs}_j[s]=\{R_l,R_{l+1}\},
\]
and suppose every correct $p_i\in\mathsf{members}(R_l)$ satisfies one of
the following conditions: either it never assigns
$\mathit{enabled\text{-}configs}_i[s]$ and therefore never sends a
round-$s$ BRB echo or vote, or its assigned value satisfies
\[
  R_l\notin\mathit{enabled\text{-}configs}_i[s].
\]
Then the configuration-scoped BRB algorithm satisfies Agreement for every
round-broadcaster slot $(s,ID)$.
\end{lemma}

\begin{proof}
By \cref{lemma:joining-install-round-cases,%
lemma:joining-old-member-install-rounds}, this case can occur only when
the earliest install transition round is $t=u-1$, every correct process in
$\mathsf{members}(R_l)$ installs in round $u-1$, and $s=u+4=t+5$.
If $R_{l+2}$ exists, then
\cref{lemma:following-reconfiguration-earliest-round} gives
\[
  s=t+5<t+6\leq\mathsf{earliest}(R_{l+2}).
\]
By
\cref{lemma:following-reconfiguration-earliest-round,%
lemma:earliest-install-lower-bound}, earliest-inclusion rounds increase
strictly along the issued sequence. Also,
$s=t+5\geq\mathsf{earliest}(R_{l+1})$.
For any $c<l$, a correct round-$s$ voter under $R_c$ would be a member
of $\mathsf{members}(R_c)$ whose round-$s$ enabled-configs value
contains $R_c$. Applying
\cref{lemma:old-configuration-round-cutoff} to that preceding
transition would give
\[
  s<\mathsf{earliest}(R_{c+2})
  \leq\mathsf{earliest}(R_{l+1}),
\]
a contradiction.

If $R_{l+2}$ exists, then
$s<\mathsf{earliest}(R_{l+2})\leq\mathsf{earliest}(R_c)$
for every $c\geq l+2$. A correct member of $R_c$ enables $R_c$
only in rounds greater than its install transition round for
$R_{c-1}\rightarrow R_c$. By
\cref{lemma:earliest-install-lower-bound}, that install transition round
is at least $\mathsf{earliest}(R_c)$. Hence no such member can vote
under $R_c$ in round $s$.
If $R_{l+2}$ does not exist, authentication excludes every later
configuration directly. Since Byzantine members alone cannot form
a certificate, every round-$s$ certificate is declared under $R_l$
or $R_{l+1}$.

By the first alternative in the lemma statement, a correct old member that
never assigns its round-$s$ value sends no round-$s$ BRB echo or vote. By
the second alternative, a correct old member that does assign the value
excludes $R_l$ and is therefore ineligible to echo or vote for an
instance declared under $R_l$. The joining process is not a member of
$R_l$ and cannot vote in an instance declared under $R_l$. At most the
$f_l$ Byzantine members can vote in such an instance, fewer than the $q_l$
shares required for a certificate. Consequently, every round-$s$
certificate is declared under $R_{l+1}$. The quorum-intersection argument
for two certificates declared under $R_{l+1}$ gives Agreement.
\end{proof}

\subsubsection{Parent-Quorum Progress During Reconfiguration}

The next six results prove conditional progress: they assume that some
correct old member has installed the transition and, when the added process
is correct, that its state-transfer discovery has returned. These
assumptions are discharged later by
\cref{lemma:eventual-first-installation,lemma:eventual-joining-discovery};
thus the conditional arguments are not used to prove their own premises.

\begin{lemma}[The additional transition-round condition advances a correct
process to its transition-completion round]
\label{lemma:transition-next-layer-advancement}
Suppose a correct process $p_i$ has
$\mathit{transition}_i=\langle R_l,R_{l+1},a_i\rangle$. If, while
$r_i=a_i+2$,
\[
  |\mathsf{ParentBlocks}_i(R_{l+1},a_i+3)|\geq q_{l+1},
\]
then $p_i$ eventually reaches the beginning of round $a_i+4$ and invokes
$\texttt{FinishTransition}_i(R_{l+1})$. The same conclusion holds if,
while $r_i=a_i+3$,
\[
  |\mathsf{ParentBlocks}_i(R_{l+1},a_i+4)|\geq q_{l+1}.
\]
\end{lemma}

\begin{proof}
In the first case, the additional disjunct of $\texttt{AdvanceRound}$ is
true with the frozen value $r=a_i+2$, so the procedure assigns
$r_i\gets a_i+3$. Its next serialized iteration selects $R_{l+1}$ for the
ordinary final-transition-round condition. The same round-$(a_i+3)$ blocks
satisfy that condition, so it assigns $r_i\gets a_i+4$. In the second case,
the additional disjunct is true with $r=a_i+3$ and directly causes the
assignment $r_i\gets a_i+4$. In either case the beginning-of-round rule
then invokes $\texttt{FinishTransition}_i(R_{l+1})$. Each iteration writes
only $r+1$, and the serialized $\texttt{AdvanceRound}$ loop is the only
post-activation writer of $r_i$.
\end{proof}

\begin{lemma}[Before stopping, a correct process stores blocks sufficient
for another correct process's transition-round advancement]
\label{lemma:stopping-process-transition-state}
Fix a valid transition $R_l\rightarrow R_{l+1}$ in an execution satisfying
the issuance-rule assumption and the configuration-level induction
hypothesis for $R_l$. Assume that the earliest correct old-member install transition
round $t$ exists. If the transition adds a correct process, assume that its
$\texttt{DiscoverBegin}$ call returns a round $u$. Let $p_i,p_h$ be correct
retained, joining, or removed processes with install transition rounds
$a_i,a_h$. Let
$r\in\{a_i+1,a_i+2,a_i+3\}$. If $p_h$ stops all protocol tasks, then,
immediately before its final $\texttt{FloodDAG}$, $DAG_h$ has the following
state:
\begin{enumerate}
  \item if $r=a_i+1$, it contains the blocks required by $p_i$'s ordinary
  round-$r$ condition; and
  \item if $r\in\{a_i+2,a_i+3\}$, it contains either the blocks required by
  $p_i$'s ordinary round-$r$ condition or at least $q_{l+1}$ blocks in
  $\mathsf{ParentBlocks}_h(R_{l+1},r+1)$.
\end{enumerate}
These conclusions do not depend on whether $p_i$ begins its round-$r$
iteration before or after $p_h$ sends that final flood.
\end{lemma}

\begin{proof}
By \cref{corollary:install-transition-round-bound}, every correct old
member has install transition round in $\{t,t+1\}$. If the transition does
not add a correct process, these are all correct retained or removed
processes. Otherwise,
let $u$ be the correct joining process's install transition round. By
\cref{lemma:joining-install-round-cases,%
lemma:joining-old-member-install-rounds}, if $t=u-1$, all correct retained,
joining, or removed processes install in $\{u-1,u\}$; if $t=u$, they install in
$\{u,u+1\}$; and if $t=u+1$, they install in $\{u,u+1\}$. Thus the install
transition rounds of every two correct retained, joining, or removed
processes differ by at most one. Consequently, if $a_h=a_i-1$, every such
correct process installs in $\{a_i-1,a_i\}$; if $a_h=a_i+1$, every such
correct process installs in $\{a_i,a_i+1\}$.

The DAG never removes a block, so it is enough to inspect how $p_h$ crosses
the relevant logical rounds before it can stop. If $a_h=a_i$, the two
processes have the same applicable conditions in round $r$. In round
$a_i+1$ the additional condition is unavailable, so crossing the round
stores the blocks required by the ordinary condition. In either of the next
two rounds, crossing it stores those blocks or at least $q_{l+1}$ blocks in
$\mathsf{ParentBlocks}_h(R_{l+1},r+1)$.

If $r=a_i+1$, the ordinary condition at $p_i$ uses $q_l$. For
$a_h\in\{a_i,a_i+1\}$, process $p_h$ cannot cross logical round $r$
without storing the blocks required by that ordinary round-$r$ condition.
If $a_h=a_i-1$, the round is $p_h$'s penultimate transition round. It either
stores the $q_l$ blocks required by that ordinary condition or advances using
$q_{l+1}$ blocks in round $r+1$.
That latter set contains a correct old-member block. For an addition,
$n_{l+1}\geq5$ and $q_{l+1}>f_{l+1}+1$, so excluding at most
$f_{l+1}$ Byzantine slots and the single joining-process slot leaves such a
block. A removal leaves $n_{l+1}\geq4$, hence $f_{l+1}\geq1$; together with
$f_l\leq f_{l+1}+1$ and $q_{l+1}\geq2f_{l+1}+1$, this gives
$q_{l+1}>f_l$. With unchanged membership,
$q_{l+1}=q_l>f_l$. The correct old-member block is declared under
$R_l$: if $a_x$ is its broadcaster's install transition round, the bound
above gives $a_x\in\{a_i-1,a_i\}$, and the block is produced in round
$r+1=a_i+2<a_x+4$, before that broadcaster completes. It therefore has
$q_l$ round-$r$ parents. Ancestor closure and
\cref{lemma:ordinary-block-parent-quorum} therefore place the ordinary
round-$r$ set in $DAG_h$ before it stops.

Suppose first that $r=a_i+2$. If $a_h=a_i+1$, then round $r$ is only
$a_h+1$ at $p_h$, where the additional condition is unavailable; crossing
it stores the ordinary $q_l$ round-$r$ set required by $p_i$. If
$a_h=a_i-1$, round $r$ is $p_h$'s final transition round. Crossing it
stores a $q_{l+1}$ set in round $r$ or $r+1$. In the latter case the second
set in the statement is already present. In the former case, $p_h$ cannot
stop all tasks immediately upon completing the transition. If it is removed
by this transition, its departure rule first obtains $f_{l+1}+1$
admissible round-$(r+2)$ blocks from distinct $R_{l+1}$-member
broadcasters. At least one has a correct broadcaster; choose that block.
If $a_x$ is the chosen correct broadcaster's install transition round, the
bound above gives $a_x\leq a_i$, so
$a_x+4\leq a_i+4=r+2$. Completion occurs at the beginning of that round,
before block production, so the chosen block is an ordinary
$R_{l+1}$-declared block with exactly $q_{l+1}$ eligible round-$(r+1)$
parents. They form the required set. If it belongs to $\mathsf{members}(R_{l+1})$,
any later removal occurs only after
it has continued advancing, and its first advancement from round $r+1$
with $R_{l+1}$ installed has no additional disjunct and stores a
$q_{l+1}$ round-$(r+1)$ set. Thus that set is present before the final
flood.

Now suppose that $r=a_i+3$. If $a_h=a_i-1$, then $p_h$ has already
completed this transition at the beginning of round $r$. If it is removed
by this transition, its DAG contains the $f_{l+1}+1$ admissible
round-$(r+1)$ blocks from distinct $R_{l+1}$-member broadcasters. Choose
one with a correct broadcaster. If $a_x$ is the chosen correct
broadcaster's install transition round, the bound above gives
$a_x\leq a_i$, so $a_x+4\leq a_i+4=r+1$. Completion occurs before block
production in that round, so this is an ordinary
$R_{l+1}$-declared block with exactly $q_{l+1}$ eligible round-$r$
parents. If it belongs to $\mathsf{members}(R_{l+1})$,
it cannot be removed later
without first taking its ordinary $R_{l+1}$ advancement from round $r$.
Either way its DAG stores the $q_{l+1}$ round-$r$ set required by $p_i$
before it stops. If $a_h=a_i+1$, then $r$ is
$p_h$'s penultimate transition round. Crossing it stores either a $q_l$
round-$r$ set or a $q_{l+1}$ round-$(r+1)$ set. For a removal or unchanged
membership, $q_l\geq q_{l+1}$, and every block eligible under $R_l$ is also
eligible under $R_{l+1}$, so the former contains the blocks required by
$p_i$'s ordinary condition. For an addition, $q_l$ may equal $q_{l+1}-1$.
Process $p_h$ must
then cross its final transition round $r+1$ before it can stop. It stores a
$q_{l+1}$ set in round $r+1$ or $r+2$. In the first case the second set in
the statement is present. In the second, the set in round $r+2$ contains a
block with a correct broadcaster because $q_{l+1}>f_{l+1}$. Every correct process
that produces a round-$(r+2)$ block has installed $R_{l+1}$ before block
production: if $a_x$ is its install transition round, the bound above gives
$a_x\leq a_i+1$, so $a_x+4\leq a_i+5=r+2$. Thus this is an ordinary
$R_{l+1}$-declared block with exactly
$q_{l+1}$ round-$(r+1)$ parents. By
\cref{lemma:ordinary-block-parent-quorum}, those parents form the required
round-$(r+1)$ set. These cases exhaust the possible relative install
transition rounds of $p_i$ and $p_h$.
\end{proof}

\begin{lemma}[Every $\texttt{AdvanceRound}$ iteration in rounds
$a_i+1,a_i+2,a_i+3$ eventually terminates]
\label{lemma:transition-advancement-alternatives}
Fix a valid transition $R_l\rightarrow R_{l+1}$ in an execution satisfying
the issuance-rule assumption and the configuration-level induction
hypothesis for $R_l$. Assume that the earliest correct old-member install transition
round $t$ exists. If the transition adds a correct process, assume that its
$\texttt{DiscoverBegin}$ call returns a round $u$. Let $p_i$ be a correct
retained, joining, or removed process with install transition round $a_i$.
Then every $\texttt{AdvanceRound}_i$ iteration before the beginning of
round $a_i+4$ eventually terminates. If the transition adds a correct
process $p_j$, then the wait $|\mathsf{LocalBlocks}_j(R_l,u+1)|\geq q_l$
that precedes $r_j\gets u+1$ and the start of $\mathsf{CatchUpTasks}_j$ in
\cref{alg:joining-process} also eventually terminates. In particular:
\begin{enumerate}
  \item an iteration whose ordinary condition uses $R_l$ eventually
  satisfies that condition, or, if its current round is $a_i+2$, eventually
  satisfies
  \[
    |\mathsf{ParentBlocks}_i(R_{l+1},a_i+3)|\geq q_{l+1};
  \]
  \item in round $a_i+3$, process $p_i$ eventually satisfies
  \[
    |\mathsf{ParentBlocks}_i(R_{l+1},a_i+3)|\geq q_{l+1}
  \]
  or
  \[
    |\mathsf{ParentBlocks}_i(R_{l+1},a_i+4)|\geq q_{l+1}.
  \]
\end{enumerate}
Moreover, if two correct $R_{l+1}$-members complete at the beginnings of
consecutive rounds $c$ and $c+1$, then the earlier process eventually
satisfies its ordinary round-$c$ parent-quorum condition.
\end{lemma}

\begin{proof}
By \cref{corollary:install-transition-round-bound,%
lemma:joining-install-round-cases,lemma:joining-old-member-install-rounds},
the install transition rounds of any two correct retained, joining, or
removed processes differ by at
most one. The threshold counts used below are the following. There are at
least $q_l$ correct old members. After a removal there are at least
$q_{l+1}$ correct retained members. After an addition, if the joining
process is correct, the correct old members together with it contain at
least $q_{l+1}\leq q_l+1$ correct members; if it is Byzantine, the fault
bound for $R_{l+1}$ leaves at least $q_{l+1}$ correct old members. With
unchanged membership the two thresholds are equal. Moreover,
\cref{lemma:joining-transition-round-bound} makes every correct old member
eligible for the correct joining process's round-$(u+2)$ and
round-$(u+3)$ BRB instances.

For every round-broadcaster slot used below, BRB Agreement is already
available without invoking the later global BRB-safety theorem.
While a correct old member still enables $R_l$, it follows from
\cref{lemma:brb-safety-current-step}; for the only round in which the joining
process still enables both configurations and no correct old member still
enables $R_l$, it follows from
\cref{lemma:brb-safety-joining-round}; and once only $R_{l+1}$ can certify a
slot, it follows from the same-configuration quorum-intersection case of
\cref{lemma:brb-aligned-safety}.

If the transition adds a correct process $p_j$, we treat its wait
$|\mathsf{LocalBlocks}_j(R_l,u+1)|\geq q_l$ as a wait of round index
$u+1=a_j+1$ whose ordinary condition is
$|\mathsf{ParentBlocks}_j(R_l,u+1)|\geq q_l$, the ordinary condition of its
round-$(u+1)$ $\texttt{AdvanceRound}$ iteration, with no additional
condition. By \cref{lemma:predecessor-parent-slots}, applied to $p_j$ in
round $u+1\geq\delta_{l+1}$, every block in
$\mathsf{ParentBlocks}_j(R_l,u+1)$ has a broadcaster in
$\mathsf{members}(R_l)$, so that condition implies the wait's condition.
While the wait is pending, $\mathsf{HistoryTasks}_j$ runs, so $p_j$
processes certified-block messages and runs $\texttt{FetchAndDeliver}$.
Suppose that some such wait, or some $\texttt{AdvanceRound}$ iteration
before its process's transition-completion round, never terminates. There
are finitely many correct retained, joining, or removed processes and at
most four such waits at each. Choose a nonterminating wait with the least
logical-round index $s$, and let $p_i$ execute it.
Let $R_c$ and $q_c$ be the configuration and threshold in its ordinary
condition. Thus $R_c=R_l$ and $q_c=q_l$ unless $s=a_i+3$, in which case
$R_c=R_{l+1}$ and $q_c=q_{l+1}$.

Choose correct processes whose prescribed round-$s$ blocks can fill that
ordinary condition. If $R_c=R_l$, choose $q_l$ correct old members. If
$R_c=R_{l+1}$, choose $q_{l+1}$ correct retained members after a removal,
$q_{l+1}$ correct retained members when membership is unchanged, and, after
an addition, choose $q_{l+1}$ correct processes from the correct old
members together with the correct joining process when it is correct, or
from the correct old members when it is Byzantine. A correct joining
process is selected only in a round in which the protocol prescribes one of
its two transition appends or ordinary post-join production; its passive
round $u+1$ is never needed for a new-configuration parent set. For every
selected block, also choose correct members to perform the prescribed BRB
tasks for its declared configuration: $q_l$ correct old members for an
$R_l$-declared block and $q_{l+1}$ correct $R_{l+1}$-members for an
$R_{l+1}$-declared block. None of these chosen members is in a passive round
when the selected instance is prescribed.

First suppose that a selected producer or a member selected to perform the
BRB tasks has already stopped all protocol tasks. It stopped only after a
removal rule and sent a final finite
$\texttt{FloodDAG}$. By
\cref{lemma:stopping-process-transition-state}, its final DAG contained the
blocks required by $p_i$'s ordinary round-$s$ condition, or, when the additional condition
is available, the permitted $q_{l+1}$ round-$(s+1)$ set. Both processes had
authenticated $R_{l+1}$, so $p_i\in\mathit{known}_h$ for the stopped process
$p_h$. The final flood was therefore sent to $p_i$, even if $p_i$ had not
yet begun this iteration. Reliable queued delivery, the certificate-level
form of BRB Agreement established above,
and \cref{lemma:certified-history-dissemination} make $p_i$ insert the set
and its ancestors and satisfy the applicable condition, a contradiction.
Hence no selected process has already stopped.

Every selected producer $p_h$ eventually enters round $s$. A correct old
member enters round $t+2$ by the proof of
\cref{corollary:install-transition-round-bound}, and every round from
$t+2$ to $s-1$ is one of its transition rounds; the correct joining process
is selected only for $s\geq u+2$, and its waits begin at round $u+1$ once
$\texttt{DiscoverBegin}$ has returned. Hence a producer that never enters
round $s$ is permanently blocked in a wait of smaller round index,
contradicting the choice of $s$. If a selected producer entered $s$ through its
ordinary condition, the predecessor set on which its prescribed round-$s$
append waits is already in its DAG. The producer is not in its passive
round, so $\mathsf{MaySendBRB}_h(R_c,s)$ holds, and the lifecycle rule keeps the
already invoked append live until it reaches and passes
$\texttt{BRB-Broadcast}$.

If instead a selected producer $p_h$ entered $s$ through the additional
condition, it already stores $q_{l+1}$ certified round-$s$ blocks. When
$R_c=R_{l+1}$, dissemination of that set and its ancestors directly
satisfies $p_i$'s ordinary condition, contradicting the choice of $p_i$.
It remains to consider this event when
$R_c=R_l$. It can occur only for a producer whose install transition round
is one less than $a_i$. Its certified round-$s$ set contains a block
produced by a correct old member and declared under $R_l$: after an addition, exclude the
at most $f_{l+1}$ Byzantine slots and the single joining-process slot; after
a removal or unchanged membership, the applicable threshold similarly
leaves a correct old-member slot. No correct old member has completed this
transition before producing a block in this case, where the producer's
install transition round is one less than $a_i$. The chosen ordinary block has exactly
$q_l$ round-$(s-1)$ parents. By
\cref{lemma:ordinary-block-parent-quorum}, those parents belong to
$\mathsf{ParentBlocks}_h(R_l,s-1)$. Dissemination of its causal history
therefore discharges the predecessor waits of the selected old-member
round-$s$ appends. Thus the selected appends reach and pass their
$\texttt{BRB-Broadcast}$ calls unless the correct process whose DAG contains
the certified round-$s$ set has already satisfied $p_i$'s condition.

The immutable round-$s$ enabled-configs values make every member selected to
perform the BRB tasks eligible for the selected instance. For every selected
producer $p_x$, advancing a local round does not change
$\mathsf{MaySendBRB}_x(R_c,s)$ because it consults the immutable historical
round-$s$ value. The BRB handlers likewise apply
$\mathsf{ValidInstance}_x(s,\mu)$ to the immutable round-$s$ value. If a
selected producer or a member selected
to perform the BRB tasks is the first one that stops all protocol tasks while
a selected certificate is missing, apply the same final-$\texttt{FloodDAG}$
argument above. The strengthened
\cref{lemma:stopping-process-transition-state} applies independently of
whether the stop precedes or follows the start of $p_i$'s wait, and again
contradicts nontermination.

In the remaining case, every chosen member continues its prescribed BRB
tasks while the selected certificates are missing. For each selected
proposal, the chosen members send the required echoes and votes, so its
certificate is formed. The $\mathsf{certified\text{-}block}$ messages sent
when certificates are assembled and blocks are inserted, together with
$\texttt{FloodDAG}$, deliver the selected blocks and their ancestors to
$p_i$. When $p_i$ is the joining process, each such sender eventually
processes $R_{l+1}$, after which $p_i$ belongs to its known-identity set by
\cref{lemma:known-coverage}, and the $\texttt{FloodDAG}$ invoked by
$\texttt{UpdateKnown}$ resends every block it stored earlier.
BRB Agreement excludes a conflicting block in every
selected slot. Hence \cref{lemma:certified-history-dissemination} makes
$p_i$ obtain the blocks required by its ordinary parent-quorum condition,
again a contradiction. Thus, if the transition adds a correct process
$p_j$, its wait $|\mathsf{LocalBlocks}_j(R_l,u+1)|\geq q_l$ terminates, and
every $\texttt{AdvanceRound}$ iteration before its process's
transition-completion round terminates.

Finally, suppose correct $R_{l+1}$-members $p_x$ and $p_y$ complete at the
beginnings of consecutive rounds $c$ and $c+1$, respectively. Process $p_y$
executes its final transition-round iteration in round $c$. Apply the result
just proved to that iteration. If it satisfies its ordinary parent-quorum
condition, dissemination
directly supplies $p_x$'s ordinary round-$c$ condition. If it uses the
$q_{l+1}$ admissible round-$(c+1)$ blocks permitted by the additional
transition-round condition and $p_x$ has already advanced to
$c+1$, then $p_x$'s round-$c$ condition has already terminated. Otherwise,
that round-$(c+1)$ set contains a block with a correct broadcaster: after
an addition, exclude at most $f_{l+1}$ Byzantine slots and the joining slot;
after a removal or unchanged membership, $q_{l+1}$ exceeds the applicable
old-configuration fault bound as shown above. Every correct old member has
completed before producing in round $c+1$, so this is an ordinary
$R_{l+1}$-declared block. By
\cref{lemma:ordinary-block-parent-quorum}, its $q_{l+1}$ parents belong to
$\mathsf{ParentBlocks}_x(R_{l+1},c)$ after their causal history is
disseminated. Thus $p_x$'s ordinary round-$c$ condition eventually holds.
Correct $R_{l+1}$-members are not removed by this transition, and
\cref{lemma:following-reconfiguration-earliest-round} prevents a following
transition from authorizing a departure during these rounds.
This proves every conclusion.
\end{proof}

\begin{lemma}[Every transition-round $\texttt{AdvanceRound}$ iteration whose
parent-quorum threshold is $q_l$ eventually terminates]
\label{lemma:transition-predecessor-quorum-progress}
Fix a valid transition $R_l\rightarrow R_{l+1}$ in an execution satisfying
the issuance-rule assumption and the configuration-level induction hypothesis for $R_l$. Assume
that an earliest install transition round $t$ among correct processes in
$\mathsf{members}(R_l)$ exists. If the transition adds a correct process
$p_j$, also assume that
$\texttt{DiscoverBegin}_j(R_l,R_{l+1})$ returns a round $u$.
Then every correct retained, joining, or removed process $p_i$ with install
transition round $a_i$ whose current
$\texttt{AdvanceRound}$ iteration has an ordinary $q_l$ condition eventually
advances. It either satisfies that condition or, in round $a_i+2$, satisfies
the additional transition-round condition with threshold $q_{l+1}$. For a correct joining
process with install transition round $u$, both $q_l$ waits in its prescribed
round-$(u+2)$ and round-$(u+3)$
$\texttt{DAG-Transition-Append}$ invocations eventually become true, and
both invocations reach $\texttt{BRB-Broadcast}$.
\end{lemma}

\begin{proof}
Apply conclusion~1 of
\cref{lemma:transition-advancement-alternatives}. If the additional
condition is used, \cref{lemma:transition-next-layer-advancement} also
discharges the final transition-round iteration.

For a correct joining process, \cref{lemma:transition-advancement-alternatives}
shows that the wait preceding $\mathsf{CatchUpTasks}_j$ terminates and that
its round-$(u+1)$ iteration terminates through
$|\mathsf{ParentBlocks}_j(R_l,u+1)|\geq q_l$, the additional condition being
unavailable in round $u+1$. Since
$\mathsf{ParentBlocks}_j(R_l,u+1)\subseteq\mathsf{ParentBlocks}_j(R_{l+1},u+1)$,
the $q_l$ wait of the round-$(u+2)$ transition append is already satisfied
when that append starts. For the round-$(u+3)$ append, either its $q_l$
admissible round-$(u+2)$ blocks are delivered directly or at least
$q_{l+1}$ admissible round-$(u+3)$ blocks are delivered. The latter set
contains a block produced by a correct old member or by the correct joining
process, with at least $q_l$ round-$(u+2)$ parents. Ancestor
closure and the certified parent-count check therefore supply the pending
append's $q_l$ set. In both rounds the enabled-configs premise for the BRB
call follows from \cref{lemma:joining-transition-round-bound}, so each
invocation reaches $\texttt{BRB-Broadcast}$.
\end{proof}

\begin{lemma}[Every final transition-round $\texttt{AdvanceRound}$ iteration
eventually terminates]
\label{lemma:transition-new-quorum-progress}
Fix a valid transition $R_l\rightarrow R_{l+1}$ in an execution satisfying
the issuance-rule assumption and the configuration-level induction hypothesis for $R_l$. Assume
that an earliest install transition round $t$ among correct processes in
$\mathsf{members}(R_l)$ exists. If the transition adds a correct process
$p_j$, also assume that
$\texttt{DiscoverBegin}_j(R_l,R_{l+1})$ returns a round $u$.
Then every correct retained, joining, or removed process $p_i$ with install
transition round $a_i$ eventually advances from its final transition round
$a_i+3$. It either DAG-delivers
$q_{l+1}$ admissible blocks in that round or DAG-delivers $q_{l+1}$
admissible blocks in round $a_i+4$ and uses the additional condition.
Moreover, if two correct $R_{l+1}$-members complete at the beginnings of
consecutive rounds $c$ and $c+1$, then the earlier process eventually
satisfies its ordinary round-$c$ parent-quorum condition.
\end{lemma}

\begin{proof}
Apply conclusion~2 and the final conclusion of
\cref{lemma:transition-advancement-alternatives}. In the additional-condition
case, \cref{lemma:transition-next-layer-advancement} reaches the completion
boundary.
\end{proof}

\begin{corollary}[Every correct retained, joining, or removed process reaches
its transition-completion round]
\label{corollary:transition-round-progress}
Under the premises of
\cref{lemma:transition-predecessor-quorum-progress}, every correct retained,
joining, or removed process with install transition round $a_i$ eventually
reaches the beginning of round $a_i+4$. A correct process removed by the
transition then continues the communication, BRB, fetch, and DAG tasks
prescribed by its departure procedure.
\end{corollary}

\begin{proof}
Apply \cref{lemma:transition-predecessor-quorum-progress,%
lemma:transition-new-quorum-progress}.
The serialized algorithm increments one round at a time, so the
beginning-of-round rule is executed at $a_i+4$.
\end{proof}

\subsubsection{Transition Completion}

\begin{lemma}[Once a correct old member installs the transition and any
correct joining process completes discovery, every correct
$R_{l+1}$-member completes the transition]
\label{lemma:retained-joining-completion}
Fix a valid transition $R_l\rightarrow R_{l+1}$ in an execution satisfying
the issuance-rule assumption and the configuration-level induction hypothesis for $R_l$. Assume
that the earliest install transition round $t$ among correct processes in
$\mathsf{members}(R_l)$ exists. If the transition adds a correct process
$p_j$, also assume that
$\texttt{DiscoverBegin}_j(R_l,R_{l+1})$ returns a round $u$.
Then every correct process $p_i\in\mathsf{members}(R_{l+1})$ eventually
installs $R_{l+1}$, completes this transition, and continues its prescribed
protocol tasks with $\mathit{local}_i=R_{l+1}$.
\end{lemma}

\begin{proof}
Every required round advancement occurs by
\cref{corollary:transition-round-progress}. By
\cref{corollary:install-transition-round-bound}, every correct retained
member invokes $\texttt{InstallTransition}$ with an install transition
round $a_i\in\{t,t+1\}$. If the transition adds a correct process $p_j$,
the joining procedure invokes $\texttt{InstallTransition}$ with
$a_j=u$, where $u$ is returned by $\texttt{DiscoverBegin}$. The progress
corollary makes each process advance through round $a_i+3$. At the
beginning of round $a_i+4$, the applicable protocol rule invokes
$\texttt{FinishTransition}_i(R_{l+1})$, which installs $R_{l+1}$ and
completes the transition.
\end{proof}

\begin{lemma}[A removed correct process eventually reaches the point at which
the protocol permits it to leave]
\label{lemma:removed-process-departure}
Fix a valid transition $R_l\rightarrow R_{l+1}$ in an execution satisfying
the issuance-rule assumption and the configuration-level induction hypothesis for $R_l$, and
assume that its earliest install transition round exists.
Let
\[
  p_i\in
  \mathsf{members}(R_l)\setminus\mathsf{members}(R_{l+1})
\]
be correct. Let $t$ be the earliest correct old-member install transition
round. By \cref{corollary:install-transition-round-bound}, $p_i$ invokes
$\texttt{InstallTransition}$ in some round $a_i\in\{t,t+1\}$. Eventually,
\[
\left|
\left\{
  \mathsf{ID}(\mu)\;\middle|\;
  \begin{array}{l}
  B=\langle m,a_i+5,\mu,C\rangle\in DAG_i[a_i+5],\\
  \mathsf{ID}(\mu)\in\mathsf{members}(R_{l+1}),\quad
  \mathsf{Admissible}_i(B)
  \end{array}
\right\}
\right|\geq f_{l+1}+1.
\]
Once this condition holds, $p_i$ stops its remaining communication and DAG
tasks and may leave.
\end{lemma}

\begin{proof}
Process $p_i$ installs $R_{l+1}$ and completes its transition at the
beginning of round $a_i+4$. It stops starting new block-production
operations and stops round advancement, leader-election, and finalization
tasks. While its departure wait is pending, every block-production operation
invoked earlier remains live until it invokes $\texttt{BRB-Broadcast}$.

It also continues processing certified-block messages, running
$\texttt{FetchAndDeliver}$, and performing the remaining communication and
DAG tasks prescribed during the departure wait. Put $v=a_i+5$. Because this
transition removes a process, every $R_{l+1}$-member is retained, and at
least $q_{l+1}$ of them are correct. By
\cref{lemma:retained-joining-completion}, all correct retained members have
completed by the beginning of round $t+5$.


We first show that some correct retained process stores at least
$q_{l+1}$ admissible round-$v$ blocks in distinct $R_{l+1}$-member slots.
If $v=t+5$, let $c_+$ be the latest completion round of a correct retained
process. By \cref{corollary:install-transition-round-bound},
$c_+\in\{t+4,t+5\}$. Suppose first that $c_+=t+5$. The processes that
complete in round $t+5$ enter that round by
\cref{corollary:transition-round-progress}, while the last conclusion of
\cref{lemma:transition-new-quorum-progress} makes every
process that completed in round $t+4$ satisfy its ordinary round-$(t+4)$
condition and enter round $t+5$. If any correct retained process uses the
additional condition to enter round $t+5$, it already stores the required
round-$v$ set. Otherwise every correct retained process enters round $v$
with its ordinary round-$(v-1)$ set, so at least $q_{l+1}$ correct
prescribed round-$v$ appends reach $\texttt{BRB-Broadcast}$.

Suppose instead that $c_+=t+4$. The removed correct process $p_i$ has an
assigned round-$(t+4)$ enabled-configs value containing $R_l$, so
\cref{lemma:brb-safety-current-step} gives Agreement through that round.
Moreover, every correct retained process has completed this transition, and
\cref{lemma:following-reconfiguration-earliest-round,lemma:no-premature-pending}
keeps a following transition from being installed before round $t+6$.
Thus the additional condition is unavailable in the selected
round-$(t+4)$ iterations. If a correct retained process used the additional
condition to enter round $t+4$, it stores $q_{l+1}$ admissible
round-$(t+4)$ blocks. Dissemination of that set and its ancestors makes
every correct retained process satisfy its ordinary round-$(t+4)$
condition, enter round $v$, and make its prescribed round-$v$ append reach
$\texttt{BRB-Broadcast}$. Otherwise every correct retained process entered
round $t+4$ with its ordinary round-$(t+3)$ set, and, because completion
precedes block production, at least $q_{l+1}$ correct prescribed
$R_{l+1}$-declared round-$(t+4)$ appends reach
$\texttt{BRB-Broadcast}$. Applying \cref{corollary:round-advancement} to
each correct retained process makes all of them enter round $v$ with the
ordinary round-$(v-1)$ set, so at least $q_{l+1}$ correct prescribed
round-$v$ appends reach $\texttt{BRB-Broadcast}$.

If $v=t+6$, first apply the same dichotomy to round $t+5$. A process that
used the additional condition already stores $q_{l+1}$ certified
round-$(t+5)$ blocks; otherwise the prescribed round-$(t+5)$
$\texttt{DAG-Append}$ operations reach $\texttt{BRB-Broadcast}$. Here
$a_i=t+1$, so the removed correct process's immutable
round-$(t+5)$ enabled-configs value still contains $R_l$;
\cref{lemma:brb-safety-current-step}, together with the BRB Agreement
clause of the configuration-level induction hypothesis for $R_l$,
gives Agreement through round $t+5$.
Each correct retained process $p_x$ whose prescribed round-$(t+5)$ append
reaches $\texttt{BRB-Broadcast}$ belongs to $R_{l+1}$, has $R_{l+1}$ in its
immutable enabled-configs value, and is not passive, so its call passes
$\mathsf{MaySendBRB}_x(R_{l+1},t+5)$. Each such process has completed the
transition into $R_{l+1}$. If $R_{l+2}$ exists,
\cref{lemma:following-reconfiguration-earliest-round,lemma:no-premature-pending}
prevents installing the following transition before round $t+6$; if it
does not exist, there is no following transition. Hence the additional
condition is unavailable in every selected round-$(t+5)$ iteration. Thus
\cref{corollary:round-advancement} supplies a correct process whose DAG
contains such a set.
Dissemination makes every correct retained member $p_x$ obtain that set and
advance ordinarily to $v$. Its round-$v$ append then uses those
round-$(v-1)$ parents, passes $\mathsf{MaySendBRB}_x(R_{l+1},v)$, and reaches
$\texttt{BRB-Broadcast}$.

In either case, if a correct retained process already stores the selected
round-$v$ set, fix it as $p_h$. Otherwise apply
\cref{corollary:round-advancement} to the $q_{l+1}$ correct processes whose
prescribed round-$v$ appends reach $\texttt{BRB-Broadcast}$. Each such
process $p_x$ is a non-passive $R_{l+1}$-member with $R_{l+1}$ in its
immutable round-$v$ enabled-configs value, so its call passes
$\mathsf{MaySendBRB}_x(R_{l+1},v)$. We next verify the
corollary's Agreement premise. The removed correct process $p_i$ has
$R_l$ in each assigned enabled-configs value before and during its install
round; $\texttt{InstallTransition}$ then assigns $\{R_l,R_{l+1}\}$ in
rounds $a_i+1,\ldots,a_i+4=v-1$. Thus
\cref{lemma:brb-safety-current-step}, together with the BRB Agreement clause
of the configuration-level induction hypothesis for $R_l$, gives Agreement
through round $v-1$.

It remains to prove Agreement in round $v$. If $v=t+5$ and some correct
$R_l$-member still has an assigned
round-$v$ enabled-configs value containing $R_l$, then
\cref{lemma:brb-safety-current-step} gives Agreement. If no such process
exists, a certificate declared under $R_l$ cannot be formed: its Byzantine
members alone are fewer than $q_l$; the conclusions for preceding
transitions similarly exclude every earlier configuration. If $R_{l+2}$
exists, \cref{lemma:following-reconfiguration-earliest-round} gives
$\mathsf{earliest}(R_{l+2})\geq t+6>v$, and authenticated consecutive
processing prevents $R_{l+3}$ or any later configuration from being enabled
before $R_{l+2}$. If $R_{l+2}$ does not exist, authentication excludes every
later configuration directly. Thus only $R_{l+1}$ can certify in round $v$,
and the
same-configuration case of \cref{lemma:brb-aligned-safety} gives Agreement.

If $v=t+6$, every correct old member's four rounds assigned
$\{R_l,R_{l+1}\}$ end by $t+5$. Hence $R_l$, every earlier configuration,
and their Byzantine
members cannot supply a quorum in round $v$. If a following $R_{l+2}$
exists, it cannot be enabled in round $v$: its earliest-inclusion round is
at least $t+6$; therefore no certified block below $v$ can mark it by
\cref{lemma:no-premature-pending} applied to the following transition, and
an old member can invoke
$\texttt{InstallTransition}$ from round-$v$ marked blocks only after its
pending ordinary round-$v$ wait terminates.
That invocation assigns the new enabled pair beginning in round $v+1$.
Likewise, a process joining in $R_{l+2}$ can discover no qualifying marked
round below $v$, and its invocation assigns the pair only after its returned
round. Authenticated consecutive processing prevents $R_{l+3}$ or any later
configuration from being enabled before $R_{l+2}$. If $R_{l+2}$ does not
exist, authentication excludes every later configuration directly.
Consequently only $R_{l+1}$ can certify in round $v$, and
same-configuration quorum intersection again gives Agreement.

By round $v$, every correct retained process has completed the transition
into $R_{l+1}$, so the additional condition for that transition is
unavailable in round $v$. A following transition cannot make that condition
applicable during the same execution of the round-$v$
$\texttt{AdvanceRound}$ loop, because $\texttt{UpdateConfigs}$ runs only
after the ordinary round-$v$ condition.
The corollary therefore makes some correct retained process $p_h$ advance
from $v$, and its ordinary condition places the required set in $DAG_h$.

Every block in the selected round-$v$ set occupies an
$R_{l+1}$-member broadcaster slot. Indeed, this transition removes only
$p_i$. Process $p_i$ is correct and starts no block-production operation for
round $v$ after completing in round $a_i+4$, so BRB Integrity excludes a
valid round-$v$ certificate in its slot. The conclusions for preceding
transitions exclude identities belonging only to a configuration earlier
than $R_l$ from producing a certified block in round $v$. Every other member
of $\mathsf{members}(R_l)$ is retained in $R_{l+1}$.
If $R_{l+2}$ exists,
\cref{lemma:following-reconfiguration-earliest-round} places its
earliest-inclusion round at or after $t+6$, and
\cref{lemma:no-premature-pending} applied to that transition shows that no
earlier certified block can mark it. Thus it cannot be installed before
round $t+6$; even if it is installed when $v=t+6$, that invocation assigns
its enabled pair beginning only in round $v+1$. Authenticated consecutive
processing also prevents a later configuration from being enabled first.
If $R_{l+2}$ does not exist, authentication excludes every later
configuration. Hence no configuration later than $R_{l+1}$ is enabled in
round $v$.
Thus the selected distinct slots belong to
$\mathsf{members}(R_{l+1})$.

Suppose, for contradiction, that the displayed condition never becomes
true. Then $p_i$ does not reach the instruction that stops its remaining
tasks, so it continues processing certified-block messages and running
$\texttt{FetchAndDeliver}$.

The selected round-$v$ set in $DAG_h$ contains at least
$q_{l+1}-f_{l+1}\geq f_{l+1}+1$ blocks with correct broadcasters; choose
that many and include their ancestors. The removed identity belongs to
$\mathit{known}_h$, so the $\mathsf{certified\text{-}block}$ messages sent
when these blocks are inserted and during $\texttt{FloodDAG}$ invocations
are addressed to $p_i$. BRB Agreement through round $v-1=a_i+4$ excludes a
different block in every ancestor slot, and BRB Integrity excludes a
different certified value in each selected correct broadcaster's round-$v$
slot. Hence \cref{lemma:certified-history-dissemination} makes $p_i$
DAG-deliver the selected blocks. Their distinct broadcaster identities all
belong to $R_{l+1}$, so the displayed departure condition holds, a
contradiction.

Hence the condition eventually holds. Process $p_i$ performs the prescribed
final finite $\texttt{FloodDAG}$ and only then stops and leaves; messages it
already sent remain deliverable.
\end{proof}

\begin{lemma}[Every correct $R_{l+1}$-member reaches a round in which
$q_{l+1}$ correct members' prescribed appends reach
$\texttt{BRB-Broadcast}$]
\label{lemma:next-outer-step-round}
Fix a valid transition $R_l\rightarrow R_{l+1}$ in an execution satisfying
the issuance-rule assumption and the configuration-level induction hypothesis for $R_l$. Assume
that the earliest install transition round $t$ among correct processes in
$\mathsf{members}(R_l)$ exists. If the transition adds a correct process
$p_j$, also assume that
$\texttt{DiscoverBegin}_j(R_l,R_{l+1})$ returns a round $u$.
Then there exists a logical-round index $s$ such that:
\begin{enumerate}
  \item every correct process in $\mathsf{members}(R_{l+1})$ eventually
  enters round $s$ with installed configuration $R_{l+1}$; and
  \item at least $q_{l+1}$ distinct active correct processes in
  $\mathsf{members}(R_{l+1})$ invoke their prescribed
  $\texttt{DAG-Append}$ operation in round $s$ and reach its
  $\texttt{BRB-Broadcast}$ call.
\end{enumerate}

\end{lemma}

\begin{proof}
For each correct
$p_i\in\mathsf{members}(R_{l+1})$, let $c_i$ be the round at whose
beginning it completes the transition. By
\cref{lemma:retained-joining-completion}, every $c_i$ is finite. Since
$\mathsf{members}(R_{l+1})$ is finite, let
\[
  c_-=\min\{c_i:
  p_i\in\mathsf{members}(R_{l+1})\text{ is correct}\}
  \qquad\text{and}\qquad
  c_+=\max\{c_i:
  p_i\in\mathsf{members}(R_{l+1})\text{ is correct}\}.
\]
The bounds on the completion rounds give $c_+\leq c_-+1$.

If $c_-=c_+$, every correct member enters round $c_+$ when it
completes the transition. Otherwise $c_+=c_-+1$. Members that complete
in round $c_+$ enter that round by
\cref{lemma:retained-joining-completion}. For each member that completes
in round $c_-$, the final conclusion of
\cref{lemma:transition-new-quorum-progress} supplies its ordinary
round-$c_-$ parent quorum, so it also advances to $c_+$.

Thus every correct member enters round $c_+$ with $R_{l+1}$ installed.
Completion occurs before round-$c_+$ block production.

If every correct member entered $c_+$ through its ordinary condition, each
already has the $q_{l+1}$ round-$(c_+-1)$ blocks on which its ordinary
round-$c_+$ append waits. Thus at least
\[
  n_{l+1}-f_{l+1}=q_{l+1}
\]
correct members invoke that append and reach $\texttt{BRB-Broadcast}$.
Take $s=c_+$.

Otherwise, some correct member entered $c_+$ through the additional
condition and therefore already stores $q_{l+1}$ admissible round-$c_+$
blocks. The $\mathsf{certified\text{-}block}$ messages sent upon insertion
and $\texttt{FloodDAG}$ invocations disseminate that set and its ancestors
to every correct $R_{l+1}$-member. After all have reached
$c_+$, their transitions are complete, so the set satisfies their ordinary
round-$c_+$ advancement condition and they enter $c_++1$. Their
round-$(c_++1)$ appends use that same round-$c_+$ set as parents and reach
$\texttt{BRB-Broadcast}$. Take $s=c_++1$. In both cases the two conclusions
hold.
\end{proof}

\subsubsection{Establishing the Prerequisites and Applying the
Configuration-Level Induction}

The preceding conditional progress results apply once their prerequisites
hold. We now prove that some correct old member invokes
$\texttt{InstallTransition}$ and, when the transition adds a correct
process, that $\texttt{DiscoverBegin}$ returns at the joining process. We
then prove the configuration-level induction hypothesis for $R_{l+1}$.

\begin{lemma}[Some correct $R_l$-member eventually invokes
$\texttt{InstallTransition}$ for $R_l\rightarrow R_{l+1}$]
\label{lemma:eventual-first-installation}
Fix a valid transition $R_l\rightarrow R_{l+1}$ in an execution satisfying
the issuance-rule assumption and the configuration-level induction hypothesis for $R_l$. Then
there exist a correct process $p_i\in\mathsf{members}(R_l)$ and a logical
round $a_i$ such that $p_i$ eventually invokes
$\texttt{InstallTransition}_i(R_l,R_{l+1},a_i)$. Consequently, the set of
install transition rounds of correct processes in
$\mathsf{members}(R_l)$ is nonempty, and its minimum
\[
  t=
  \min\left\{
  a_i\;\middle|\;
  \begin{array}{l}
  p_i\in\mathsf{members}(R_l)\text{ is correct and}\\
  p_i\text{ invokes }
  \texttt{InstallTransition}_i(R_l,R_{l+1},a_i)
  \end{array}
  \right\}
\]
exists.
\end{lemma}

\begin{proof}
Every correct process that receives $R_{l+1}$ authenticates and buffers it.
Authentication and consecutive numbering exclude a conflicting message.

Suppose, for contradiction, that no correct process in
$\mathsf{members}(R_l)$ ever invokes $\texttt{InstallTransition}$. Every enabled-configs
value subsequently assigned by such a process then equals $\{R_l\}$.
Starting from the logical round whose existence is required by clause~(iv)
of the configuration-level induction hypothesis for $R_l$, BRB Agreement for
certificates declared under $R_l$ and induction on the logical round using
\cref{corollary:round-advancement} make every active correct old member enter
every later logical round. The induction invariant includes at least $q_l$
correct appends that reach $\texttt{BRB-Broadcast}$: clause~(iv) gives the
base, the contradiction hypothesis makes the additional transition
condition unavailable in every round, and ordinary entry supplies the
parent set needed by the next append.

Consequently, after installing $R_l$, every correct old member that has
received $R_{l+1}$ eventually reaches
$\mathsf{earliest}(R_{l+1})$. Its reconfiguration handler then completes
the wait and assigns $\mathit{pending}_i\gets R_{l+1}$.

Choose a set $\mathcal S$ of $f_l+1$ correct old members. Eventual message
delivery ensures that every member of the finite set $\mathcal S$ receives
$R_{l+1}$ and eventually makes this pending assignment. For each member,
let $s_i$ be the first logical round whose append begins after the handler
assigns $\mathit{pending}_i\gets R_{l+1}$, and let
$s=\max_{p_i\in\mathcal S}s_i$.
Because these processes continue advancing, every member of $\mathcal S$
eventually enters round $s$ and produces a certified block whose
pending field contains $R_{l+1}$. Let $\mathcal M$ be these $f_l+1$ marked
round-$s$ blocks.

If a correct old member invokes $\texttt{InstallTransition}$ while
processing round $s$, the
contradiction is immediate. Otherwise, let $X$ be any round-$(s+1)$ block
produced by a correct old member. The round-broadcaster slots occupied by
its $q_l$ round-$s$ parents intersect the $f_l+1$ slots occupied by
$\mathcal M$. BRB Agreement in round $s$
ensures that an intersecting slot denotes the same marked block, so $X$ is
marked. Thus every round-$(s+1)$ block from a correct old broadcaster is
marked. Any correct old member that advances from that round obtains $q_l$
blocks in distinct round-broadcaster slots occupied by members of $R_l$, at
least
$q_l-f_l\geq f_l+1$ of which have correct broadcasters and are marked. It
therefore invokes $\texttt{InstallTransition}$, contradicting the
assumption.
\end{proof}

\begin{lemma}[$\texttt{DiscoverBegin}$ at a correct joining process
eventually terminates]
\label{lemma:eventual-joining-discovery}
Fix a valid transition $R_l\rightarrow R_{l+1}$ in an execution satisfying
the issuance-rule assumption and the configuration-level induction hypothesis for $R_l$, and
suppose it adds a correct process $p_j$. Let $t$ be the earliest install
transition round among correct processes in $\mathsf{members}(R_l)$,
whose existence follows from
\cref{lemma:eventual-first-installation}. Then
$\texttt{DiscoverBegin}_j(R_l,R_{l+1})$ eventually terminates and returns
the least qualifying round $u$ present in $DAG_j$ when its wait is
discharged: that round contains blocks in at least $q_l$ distinct
round-broadcaster slots occupied by members of $R_l$, including at least
$f_l+1$ blocks marked for $R_l\rightarrow R_{l+1}$.
\end{lemma}

\begin{proof}
The distribution mechanism sends the authenticated pair
$\langle R_l,R_{l+1}\rangle$ to $p_j$, so the joining handler can validate
the transition and obtain $\Pi_l$, $q_l$, and $f_l$.
BRB Agreement through round $t$ follows from
\cref{lemma:brb-safety-current-step}. By
\cref{lemma:marked-block-propagation}, every correct
round-$(t+1)$ block in $\mathsf{members}(R_l)$ is marked. At least $q_l$
correct old-member round-$(t+1)$ appends reach $\texttt{BRB-Broadcast}$:
entry into that round used the ordinary round-$t$ condition, and the
additional condition is unavailable before round $t+2$.
BRB Agreement and \cref{corollary:round-advancement} therefore make a
correct old member store at least $q_l$ round-$(t+1)$ blocks in distinct
old-member broadcaster slots. Any such set contains at least
\[
  q_l-f_l=n_l-2f_l\geq f_l+1
\]
blocks with correct broadcasters, and all of those blocks are marked. Fix one
correct old member $p_h$, let $\mathcal Q$ be such a set in $DAG_h$, and let
$\mathcal D$ contain $\mathcal Q$ and all of its ancestors. The set is finite
because each block has finitely many parents in the preceding round, and
\cref{lemma:local-dag-ancestor-closure} gives $\mathcal D\subseteq DAG_h$.

Process $p_h$ eventually receives $R_{l+1}$. Upon that receipt,
$\texttt{UpdateKnown}_h(R_{l+1})$ adds $p_j$ to $\mathit{known}_h$ and invokes
$\texttt{FloodDAG}_h$. That invocation, certified-block messages sent upon
insertion, later invocations of $\texttt{FloodDAG}_h$, and a final invocation
if $p_h$ is later removed send every block in $\mathcal D$ to $p_j$.

Suppose, for contradiction, that
$\texttt{DiscoverBegin}_j(R_l,R_{l+1})$ never returns.
The joining protocol starts $\mathsf{HistoryTasks}_j$ before invoking
$\texttt{DiscoverBegin}_j(R_l,R_{l+1})$; those tasks continue processing certified-block
messages and running $\texttt{FetchAndDeliver}$. The blocks in $\mathcal D$
pass the required metadata and referenced-parent checks, and BRB Agreement
through round $t+1$ excludes a different block in any of their
round-broadcaster slots. Therefore,
\cref{lemma:certified-history-dissemination} makes $p_j$ DAG-deliver
$\mathcal D$. Consequently, $DAG_j$ contains the $q_l$ distinct
blocks in $\mathcal Q$, whose broadcaster identities belong to $R_l$,
including at least $f_l+1$ marked blocks, even if $p_h$ has already
departed.

Thus $DAG_j$ eventually contains a round satisfying the two discovery
conditions, contradicting the assumption that
$\texttt{DiscoverBegin}_j(R_l,R_{l+1})$ never returns. Therefore its wait is
eventually discharged, and it returns
the least qualifying round then present in $DAG_j$; it makes no claim about
blocks merged only after the return. The repeated request
$\langle\mathsf{state\text{-}request},R_l,R_{l+1}\rangle$ to every identity
in $\mathit{known}_j$
provides an additional retrieval path, and one sufficient state reply would
establish the same condition, but the proof does not require a process that
stores a block to remain active until it receives a request. Whenever a reply is used, its safety is
sender-independent by
\cref{lemma:sender-independent-state-validation}.
\end{proof}

\begin{lemma}[A valid transition establishes completion, departure, progress,
and the BRB Agreement premise for the next induction step]
\label{lemma:configuration-step-closure}
Fix a valid transition $R_l\rightarrow R_{l+1}$ in an execution satisfying
the system and reconfiguration assumptions of
\cref{sec:problem-definition} and the issuance-rule assumption, and assume
the configuration-level induction hypothesis for $R_l$. Then the following conclusions hold:
\begin{enumerate}
  \item Every correct process
  $p_i\in\mathsf{members}(R_{l+1})$ eventually invokes
  $\texttt{InstallTransition}_i(R_l,R_{l+1},a_i)$ for some finite round
  $a_i$, reaches the beginning of round $a_i+4$, and invokes
  $\texttt{FinishTransition}_i(R_{l+1})$. Immediately afterward,
  \[
    \mathit{local}_i=R_{l+1},\qquad
    \mathit{pending}_i=\bot,\qquad
    \mathit{transition}_i=\bot.
  \]

  \item Every correct process in
  $\mathsf{members}(R_l)\setminus\mathsf{members}(R_{l+1})$ eventually
  reaches the point at which the protocol permits it to leave.

  \item There is a logical round $s$ that every correct process in
  $\mathsf{members}(R_{l+1})$ eventually enters with $R_{l+1}$ installed
  and in which at least $q_{l+1}$ active correct members invoke the
  prescribed $\texttt{DAG-Append}$ operation and reach its
  $\texttt{BRB-Broadcast}$ call.

  \item If a valid transition $R_{l+1}\rightarrow R_{l+2}$ is issued, then
  BRB Agreement holds for every round-broadcaster slot whose round index is
  less than
  $\delta_{l+2}=\mathsf{earliest}(R_{l+2})$. If no following transition is
  issued, BRB Agreement holds for every round-broadcaster slot whose round
  index is at least $\delta_{l+1}$.
\end{enumerate}
Consequently, if $R_{l+1}\rightarrow R_{l+2}$ is issued, the
configuration-level induction hypothesis holds for $R_{l+1}$.
\end{lemma}

\begin{proof}
By \cref{lemma:eventual-first-installation}, the earliest install transition
round $t$ among correct processes in $\mathsf{members}(R_l)$ exists. If the
transition adds a correct process $p_j$,
\cref{lemma:eventual-joining-discovery} supplies the round $u$ returned by
$\texttt{DiscoverBegin}_j(R_l,R_{l+1})$.

Every correct retained process invokes $\texttt{InstallTransition}$ in
round $t$ or $t+1$ by
\cref{corollary:install-transition-round-bound}. A correct joining process
invokes $\texttt{InstallTransition}$ with install transition round $u$ by
\cref{alg:joining-process}. Hence the premises of
\cref{corollary:transition-round-progress} hold. That corollary supplies each
required round advancement, and
\cref{lemma:retained-joining-completion} shows that every correct
$R_{l+1}$-member reaches the beginning of its round $a_i+4$. At that point,
the retained-process and joining-process rules in
\cref{alg:reconfiguration,alg:joining-process} invoke
$\texttt{FinishTransition}_i(R_{l+1})$; the definition of that procedure in
\cref{alg:reconfiguration} gives the stated post-transition values. This
proves the first conclusion. If a correct process is removed,
\cref{lemma:removed-process-departure} proves the second conclusion;
otherwise it is vacuous. Finally, \cref{lemma:next-outer-step-round} proves
the third conclusion.

For the fourth conclusion, if $R_{l+1}\rightarrow R_{l+2}$ is issued, fix
a round $r<\delta_{l+2}$; otherwise, fix any round
$r\geq\delta_{l+1}$. If $r<\delta_{l+1}$, BRB Agreement follows directly
from the configuration-level induction hypothesis for $R_l$. Suppose
instead that $r\geq\delta_{l+1}$. There are three cases.
\begin{enumerate}
  \item If some correct old member has an assigned round-$r$
  enabled-configs value containing $R_l$, apply
  \cref{lemma:brb-safety-current-step}.

  \item Suppose no correct old member has such a value, but the transition
  adds a correct process whose round-$r$ value is $\{R_l,R_{l+1}\}$. Every
  correct old member then either has no assigned value and sends no
  round-$r$ BRB echo or vote, or has an assigned value excluding $R_l$.
  \Cref{lemma:brb-safety-joining-round} gives Agreement.

  \item Otherwise, no correct process eligible to vote enables $R_l$.
  Earliest-inclusion rounds increase strictly along the issued sequence
  by
  \cref{lemma:following-reconfiguration-earliest-round,%
  lemma:earliest-install-lower-bound}.
  For any $c<l$, a correct round-$r$ voter under $R_c$ would be a member
  of $\mathsf{members}(R_c)$ whose round-$r$ enabled-configs value
  contains $R_c$. Applying
  \cref{lemma:old-configuration-round-cutoff} to that preceding
  transition would give
  \[
    r<\mathsf{earliest}(R_{c+2})
    \leq\mathsf{earliest}(R_{l+1})
    =\delta_{l+1},
  \]
  contradicting $r\geq\delta_{l+1}$.

  If $R_{l+2}$ is issued, then
  $r<\mathsf{earliest}(R_{l+2})\leq\mathsf{earliest}(R_c)$
  for every $c\geq l+2$. A correct member of $R_c$ enables $R_c$
  only in rounds greater than its install transition round for
  $R_{c-1}\rightarrow R_c$. By
  \cref{lemma:earliest-install-lower-bound}, that install transition
  round is at least $\mathsf{earliest}(R_c)$.
  Hence no such member can vote under $R_c$ in round $r$.
  If $R_{l+2}$ is not issued, authentication excludes every later
  configuration directly. Since Byzantine members alone cannot form
  a certificate, every valid round-$r$ certificate is declared under
  $R_{l+1}$. Two such certificate quorums intersect in a correct process,
  which votes for at most one value in a round-broadcaster slot, so
  BRB Agreement holds.
\end{enumerate}
These cases cover every required round and prove the fourth conclusion.

When a successor transition is issued, conclusions~1, 2, 4, and~3 establish
clauses~(i), (ii), (iii), and~(iv), respectively, of the
configuration-level induction hypothesis for $R_{l+1}$.
\end{proof}

\begin{lemma}[The configuration-level induction hypothesis holds for every
configuration with an issued successor]
\label{lemma:configuration-level-induction}
For every issued valid transition $R_l\rightarrow R_{l+1}$ in the
authenticated configuration sequence, the configuration-level induction
hypothesis for $R_l$ defined above holds.
\end{lemma}

\begin{proof}
Proceed by induction on $l$. For $l=1$, every correct process in
$\mathsf{members}(R_1)$ starts in round $1$ with $R_1$ installed by
\cref{alg:brb-variables}, while the
clauses concerning a preceding transition and removed processes are
vacuous. By the rule in \cref{sec:our-solution} that invokes
$\texttt{DAG-Append}$ at the beginning of a round, every active correct
process in $\mathsf{members}(R_1)$ invokes the prescribed operation in
round~$1$. Since at most $f_1$ of the $n_1$ members of $R_1$ are Byzantine,
at least $n_1-f_1=q_1$ active correct processes in $\mathsf{members}(R_1)$
make that invocation. A round-$1$ append has no parent wait, so each reaches
its $\texttt{BRB-Broadcast}$ call. Before
$\delta_2=\mathsf{earliest}(R_2)$, only an instance declared under $R_1$
can receive correct votes. Quorum intersection for two certificates declared
under $R_1$ therefore
gives BRB Agreement in every round below $\delta_2$. Thus the hypothesis
holds for $R_1$.

For the induction step, suppose the hypothesis holds for $R_l$ and that
$R_{l+1}\rightarrow R_{l+2}$ is issued. Authenticated consecutive
processing implies that $R_l\rightarrow R_{l+1}$ was issued. Apply
\cref{lemma:configuration-step-closure} to that transition. Its four
conclusions establish the configuration-level induction hypothesis for
$R_{l+1}$.
\end{proof}

\begin{lemma}[Every valid transition completes at all correct retained or
joining processes and eventually permits all correct removed processes to
leave]
\label{lemma:reconfiguration-completion}
In every execution satisfying the system and reconfiguration assumptions
of \cref{sec:problem-definition} and the issuance-rule assumption, every
valid transition
$R_l\rightarrow R_{l+1}$ issued in the authenticated configuration sequence
has the following properties. Each correct process
$p_i\in\mathsf{members}(R_{l+1})$ eventually invokes
$\texttt{InstallTransition}_i(R_l,R_{l+1},a_i)$
for some finite round $a_i$, reaches the beginning of round $a_i+4$, and
invokes $\texttt{FinishTransition}_i(R_{l+1})$. Immediately afterward,
\[
  \mathit{local}_i=R_{l+1},\qquad
  \mathit{pending}_i=\bot,\qquad
  \mathit{transition}_i=\bot.
\]
Moreover, every correct process in
$\mathsf{members}(R_l)\setminus\mathsf{members}(R_{l+1})$ eventually reaches
the point at which the protocol permits it to leave.
\end{lemma}

\begin{proof}
By \cref{lemma:configuration-level-induction}, the configuration-level
induction hypothesis holds for $R_l$. The first two conclusions of
\cref{lemma:configuration-step-closure} respectively establish completion
at every correct retained or joining process and permitted departure for
every correct removed process.
\end{proof}

\subsubsection{BRB Agreement for Every Round-Broadcaster Slot and Dynamic
DAG-Structured Byzantine Broadcast Correctness}

\begin{lemma}[The configuration-scoped BRB algorithm satisfies Agreement for
every round-broadcaster slot]
\label{lemma:brb-safety}
The configuration-scoped BRB algorithm in \cref{alg:brb} satisfies
BRB Agreement for every round-broadcaster slot.
\end{lemma}


\begin{proof}
Let $R_1,\ldots,R_L$ be the finite authenticated configuration sequence.
If $L=1$, authentication and configuration authorization imply that every
valid certificate is declared under $R_1$. Any two certificates for the
same round-broadcaster slot then have voter quorums that intersect in a
correct process. Since a correct process votes for at most one value in a
slot, the certificates cannot certify different values.

Suppose now that $L\geq2$. By
\cref{lemma:configuration-level-induction}, the configuration-level
induction hypothesis holds before every issued transition. In particular,
its Agreement clause covers every round before
$\delta_2=\mathsf{earliest}(R_2)$. Apply
\cref{lemma:configuration-step-closure} successively to the transitions
$R_1\rightarrow R_2,\ldots,R_{L-1}\rightarrow R_L$. For every non-final
transition, the fourth conclusion of that lemma extends Agreement to every
round before the following transition. For
$R_{L-1}\rightarrow R_L$, there is no following transition, so the same
conclusion establishes Agreement in every remaining logical round. Thus BRB
Agreement holds for every round-broadcaster slot in the execution.
\end{proof}

\begin{lemma}[The protocol implements Dynamic DAG-structured Byzantine
Broadcast]
\label{lemma:dynamic-dag-structured-broadcast}
Under the system and reconfiguration assumptions of
\cref{sec:problem-definition}, the issuance-rule assumption, and the
protocol's $\mathit{known}_i$ and admissibility checks, the
protocol specified in
\cref{alg:brb,alg:dag-append-deliver,alg:reconfiguration,%
alg:joining-transition-validation,alg:joining-state-transfer,%
alg:joining-process}
implements a Dynamic DAG-structured Byzantine Broadcast satisfying
\cref{def:dynamic-dag-structured-broadcast}.
\end{lemma}

\begin{proof}
\emph{Configuration-scoped BRB Agreement and Integrity.}
\Cref{lemma:brb-integrity} proves Integrity for every
configuration-scoped BRB instance, and \cref{lemma:brb-safety} proves
Agreement for every round-broadcaster slot across the complete
configuration sequence.

\emph{DAG-Integrity.}
Fix a correct process $p_j$. A block
$B=\langle m,r,\mu,C\rangle$ is inserted into $DAG_j[r]$
only after the corresponding BRB instance produces the valid quorum
certificate $C$. BRB Integrity and Agreement
prevent a correct process from accepting two certified values for one
round-broadcaster slot. If the broadcaster $p_i$ is correct, the
initial BRB message for that slot originates from its prescribed
$\texttt{DAG-Append}_i$ or $\texttt{DAG-Transition-Append}_i$
invocation. Thus every block inserted under a correct broadcaster's identity
originates from the prescribed append invocation, and each correct process
DAG-delivers at most one block for each round-broadcaster slot.

\emph{DAG-Consistency.}
Suppose two correct processes DAG-deliver blocks for the same
round-broadcaster slot. BRB Agreement from \cref{lemma:brb-safety} makes
their $(m,r,\mu)$ tuples identical, and
threshold-signature uniqueness makes their certificates identical. Hence
they DAG-deliver the same block
$B=\langle m,r,\mu,C\rangle$.

\emph{Causal Closure.}
Fix a block $B=\langle m,r,\mu,C\rangle$ DAG-delivered by a correct
process, and let $R_c=\mathsf{local}(\mu)$ be the configuration under which
$C$ is formed. By \cref{lemma:local-dag-ancestor-closure}, the certified-block
receive and fetch procedures and the validated state-reply merge procedure
insert every parent and ancestor before its descendant.

The remaining checks establish that every referenced parent is admissible.
If $\mu'$ is the metadata of a parent, then
$\mathsf{number}(\mathsf{local}(\mu'))
\leq\mathsf{number}(\mathsf{local}(\mu))$. The insertion check also verifies
that the parent broadcaster belongs to
$\mathsf{members}(R_d)$ for some
$R_d\in\mathsf{enabled}(\mu')$: the unique configuration when
$\mathsf{enabled}(\mu')$ is a singleton, or either configuration when it is
the authenticated pair of a valid transition. A
correct joining process need not reconstruct the earlier
process-local values $\mathit{enabled\text{-}configs}_k[s]$.

Because $q_c>f_c$, the block's BRB certificate includes a share from a
correct voter $p_v$ that checked the complete $(m,r,\mu)$ tuple and applied
the applicable parent-count rule. Suppose
$\mathsf{JoiningMetadata}(\mu,R_l,R_{l+1})$ holds and the block uses
$q_l$ parents while declaring $R_{l+1}$. If $q_l<q_{l+1}$, that check
required
\[
 \mathit{enabled\text{-}configs}_v[r]
  =\mathsf{enabled}(\mu)
  =\{R_l,R_{l+1}\}.
\]
Here $\{R_l,R_{l+1}\}$ is the authenticated pair satisfying
$\mathsf{ValidReconfiguration}(R_l,R_{l+1})$.
If $q_l=q_{l+1}$, the parent count satisfies the ordinary rule and does
not imply this equality.
Finally, a correct joining process inserts only validated,
ancestor-closed state into the local DAG. Hence every DAG-delivered
block is stored together with every parent and ancestor required by
Causal Closure.


\emph{Per-Round Quorum Progress.}
BRB Agreement for every round-broadcaster slot follows from
\cref{lemma:brb-safety}. When the prescribed correct appends reach
$\texttt{BRB-Broadcast}$ and the correct processes needed to form the
certificates continue their BRB tasks,
\cref{lemma:dag-quorum-dissemination} delivers the certified
current-round blocks and their ancestors. When the additional condition is
unavailable, \cref{corollary:round-advancement} removes the separate lifetime
premise for the processes needed to form the certificates: either a correct
process first stores the current-round quorum and sends it, or all of those
processes remain active long enough to certify it.

During a transition, a process that advances without satisfying its ordinary
parent-quorum condition does so only by satisfying the additional condition,
which requires $q_{l+1}$ admissible blocks from the next round. The
ancestor-closure and final-$\texttt{FloodDAG}$ arguments in
\cref{lemma:transition-advancement-alternatives} then give every affected
correct process either the blocks required by its ordinary parent-quorum
condition or the $q_{l+1}$ admissible next-round blocks allowed by the
additional condition in its last two transition rounds.
\Cref{lemma:transition-predecessor-quorum-progress,%
lemma:transition-new-quorum-progress} apply these alternatives to the
parent-quorum conditions with thresholds $q_l$ and $q_{l+1}$, including the
ordinary condition of a process that completes the transition one round
before another correct member.
Thus every active correct process whose prescribed round advancement
depends on the applicable quorum eventually DAG-delivers sufficient
admissible blocks from distinct broadcaster slots to satisfy an
applicable round-advancement condition and enter the next round.
This establishes Per-Round Quorum Progress as stated in
\cref{def:dynamic-dag-structured-broadcast}.

Finally, \cref{lemma:reconfiguration-completion} proves the completion
clause of \cref{def:dynamic-dag-structured-broadcast}: every retained or
joining correct process completes a valid transition, and every removed
correct process reaches the point at which it may leave.
\end{proof}

\subsection{Leader Integrity and Total Order}
\label{sec:leader-total-order}

\subsubsection{Alignment of Enabled-Configs Values Across Correct Processes
in the Same Logical Round}

\cref{lemma:brb-agreement-transition-rounds} establishes alignment for
correct members of the old configuration throughout one transition. We
first include a correct joining process for that transition and then compose
the result across the finite configuration sequence. For every issued
transition used below, \cref{lemma:configuration-level-induction} supplies
the configuration-level induction hypothesis required by the
lemmas that fix one transition.

\begin{lemma}[Transition windows of distinct valid transitions are
disjoint]
\label{lemma:transition-window-disjointness}
In every execution satisfying the assumptions of
\cref{sec:problem-definition} and the issuance-rule assumption, the
following holds. For a valid issued transition $R_l\rightarrow R_{l+1}$,
let
\[
  W_l=\bigl\{s\;\bigm|\;\exists\text{ correct }p_i:\;
  \mathit{enabled\text{-}configs}_i[s]=\{R_l,R_{l+1}\}\bigr\}
\]
be its transition window. If $W_l\neq\varnothing$, then the earliest
install transition round $t_l$ among correct old members exists and
$W_l\subseteq[t_l,\,t_l+5]$. For any two valid issued transitions with
$l<m$ and $W_l,W_m\neq\varnothing$, every round in $W_m$ exceeds every
round in $W_l$. In particular, $W_l\cap W_m=\varnothing$: for every
logical round $s$ there is at most one index $l$ such that some correct
process $p_i$ satisfies
$\mathit{enabled\text{-}configs}_i[s]=\{R_l,R_{l+1}\}$.
\end{lemma}

\begin{proof}
A correct process $p_i$ assigns the pair $\{R_l,R_{l+1}\}$ only through an
invocation \\ $\texttt{InstallTransition}_i(R_l,R_{l+1},a)$, which assigns
it exactly to rounds $a+1,\ldots,a+4$, and only correct members of
$\mathsf{members}(R_l)$ and a correct process joining in the transition
invoke it. Hence, if $W_l\neq\varnothing$, some correct process invokes
$\texttt{InstallTransition}$ for the transition and $t_l$ exists.
Because the transition is issued,
\cref{lemma:configuration-level-induction} supplies the
configuration-level induction hypothesis for $R_l$. By
\cref{corollary:install-transition-round-bound}, every correct member of
$\mathsf{members}(R_l)$ that invokes $\texttt{InstallTransition}$ does
so in round $t_l$ or $t_l+1$, and by
\cref{lemma:joining-install-round-cases,%
lemma:joining-old-member-install-rounds}, a correct joining process does
so in a round $u_l\in\{t_l-1,t_l,t_l+1\}$. Every assignment of this
pair therefore covers rounds in $[a+1,a+4]$ with
$t_l-1\leq a\leq t_l+1$, so $W_l\subseteq[t_l,\,t_l+5]$.

Now fix $l<m$ with $W_l,W_m\neq\varnothing$. For every $k$ with
$l\leq k<m$, the transition $R_{k+1}\rightarrow R_{k+2}$ is issued, so
\cref{lemma:following-reconfiguration-requires-installation} yields a
correct invocation of $\texttt{InstallTransition}$ for
$R_k\rightarrow R_{k+1}$, and $t_k$ exists.
\Cref{lemma:following-reconfiguration-earliest-round} gives
$\mathsf{earliest}(R_{k+2})\geq t_k+6$, and
\cref{lemma:earliest-install-lower-bound} gives
$t_{k+1}\geq\mathsf{earliest}(R_{k+2})$. Hence $t_{k+1}\geq t_k+6$ for
every such $k$, and
\[
  \min W_m\;\geq\;t_m\;\geq\;t_l+6(m-l)\;>\;t_l+5\;\geq\;\max W_l.
\]
Every round of $W_m$ therefore exceeds every round of $W_l$, and the
windows are disjoint. 
\end{proof}

\begin{lemma}[Correct members of $R_l$ and $R_{l+1}$ cannot assign
$\{R_l\}$ and $\{R_{l+1}\}$ in the same round]
\label{lemma:same-transition-enabled-configs}
In an execution satisfying the assumptions of
\cref{sec:problem-definition}, fix a valid issued transition
$R_l\rightarrow R_{l+1}$ and let
\[
  \delta_{l+1}=\mathsf{earliest}(R_{l+1}).
\]
Fix a logical round $r\geq\delta_{l+1}$, and suppose some correct
$p_i\in\mathsf{members}(R_l)\cup\mathsf{members}(R_{l+1})$ satisfies
\[
  \mathit{enabled\text{-}configs}_i[r]\neq\bot
  \qquad\text{and}\qquad
  R_l\in\mathit{enabled\text{-}configs}_i[r].
\]
Then every correct process
$p_j\in\mathsf{members}(R_l)\cup\mathsf{members}(R_{l+1})$
with $\mathit{enabled\text{-}configs}_j[r]\neq\bot$ satisfies
\[
  \mathit{enabled\text{-}configs}_j[r]
  \in
  \bigl\{
    \{R_l\},
    \{R_l,R_{l+1}\},
    \{R_{l+1}\}
  \bigr\}.
\]
Moreover, there do not exist two such correct processes $p_j$ and $p_k$
for which
\[
  \mathit{enabled\text{-}configs}_j[r]=\{R_l\}
  \qquad\text{and}\qquad
  \mathit{enabled\text{-}configs}_k[r]=\{R_{l+1}\}.
\]
\end{lemma}

\begin{proof}
By \cref{lemma:configuration-level-induction}, the configuration-level
induction hypothesis holds for $R_l$.
By \cref{lemma:eventual-first-installation}, let $t$ be the earliest install
transition round. By
\cref{corollary:install-transition-round-bound}, every correct process in
$\mathsf{members}(R_l)$ installs in round $t$ or $t+1$. Its successive
immutable enabled-configs values are $\{R_l\}$, then
$\{R_l,R_{l+1}\}$, and then $\{R_{l+1}\}$, and two such processes cannot have
the two different singleton values in the same round. If a correct
process joins, \cref{lemma:joining-install-round-cases,%
lemma:joining-old-member-install-rounds} gives the same alternatives between
its enabled-configs values and those of the old members.

To exclude earlier configurations, fix $c<l$. By
\cref{lemma:transition-window-disjointness,%
lemma:following-reconfiguration-earliest-round}, every round assigned
the pair $\{R_c,R_{c+1}\}$ is below
$\mathsf{earliest}(R_{c+2})$. Earliest-inclusion rounds increase strictly
along the issued sequence by
\cref{lemma:following-reconfiguration-earliest-round,%
lemma:earliest-install-lower-bound}. Therefore,
\[
  \mathsf{earliest}(R_{c+2})
  \leq\mathsf{earliest}(R_{l+1})
  =\delta_{l+1}\leq r,
\]
so no such pair is assigned in round $r$.
Since $r>1$, a correct process assigning the singleton $\{R_c\}$
in round $r$ is a member of $\mathsf{members}(R_c)$.
\Cref{lemma:old-configuration-round-cutoff} would then give
$r<\mathsf{earliest}(R_{c+2})$, which is also impossible.

The premise and the install transition round bounds above now give
$r\leq t+5$. If $R_{l+2}$ exists,
\cref{lemma:following-reconfiguration-earliest-round} gives
\[
  r\leq t+5<\mathsf{earliest}(R_{l+2})
  \leq\mathsf{earliest}(R_c)
  \qquad\text{for every }c\geq l+2.
\]
A correct process assigns a value containing such an $R_c$ only after
an invocation of $\texttt{InstallTransition}$ for
$R_{c-1}\rightarrow R_c$ or $R_c\rightarrow R_{c+1}$.
By \cref{lemma:earliest-install-lower-bound} and the increasing
earliest-inclusion rounds, either invocation assigns such values only
to rounds greater than $\mathsf{earliest}(R_c)$.
Thus no configuration after $R_{l+1}$ occurs in an assigned round-$r$
value. If $R_{l+2}$ does not exist, authentication excludes every later
configuration directly. The stated alternatives therefore cover all
assigned values in round $r$.
\end{proof}

\begin{lemma}[Enabled-configs values assigned to correct processes in one
logical round involve at most one valid transition]
\label{lemma:same-round-enabled-configs}
Fix a logical round $r$ and correct processes $p_i$ and $p_j$ such that
\[
  E_i:=\mathit{enabled\text{-}configs}_i[r]\neq\bot,
  \qquad
  E_j:=\mathit{enabled\text{-}configs}_j[r]\neq\bot.
\]
Then either
\[
  E_i=E_j=\{R_c\}
\]
for some configuration $R_c$, or there is a valid transition
$R_l\rightarrow R_{l+1}$ such that
\[
  E_i,E_j
  \in
  \bigl\{
    \{R_l\},
    \{R_l,R_{l+1}\},
    \{R_{l+1}\}
  \bigr\}
\]
and at least one of $E_i$ and $E_j$ equals
$\{R_l,R_{l+1}\}$. In particular, if $E_i$ and $E_j$ are both
singletons, then $E_i=E_j$.
\end{lemma}

\begin{proof}
It suffices to prove the stronger claim for the collection of all assigned
round-$r$ values at correct processes. For a generic correct process $p_x$,
inspect $\mathit{enabled\text{-}configs}_x[r]$. Each
process moves through the authenticated configuration sequence in consecutive
order, and for each transition its successive enabled-configs values are the
old singleton, then the pair comprising the two configurations of that
valid transition, then the new singleton.

Suppose one such value is $\{R_l,R_{l+1}\}$, and let $t$ be the earliest
install transition round for this transition.
\Cref{lemma:transition-window-disjointness} excludes every other
transition pair in round $r$ and gives $r\leq t+5$.
The pair is assigned only to rounds after its install transition round,
so \cref{lemma:earliest-install-lower-bound} also gives
\[
  r>\mathsf{earliest}(R_{l+1}).
\]

Consider any assigned singleton $\{R_c\}$ with $c<l$.
Since $r>1$, the correct process assigning it is a member of
$\mathsf{members}(R_c)$. By
\cref{lemma:old-configuration-round-cutoff},
\[
  r<\mathsf{earliest}(R_{c+2})
  \leq\mathsf{earliest}(R_{l+1}),
\]
a contradiction. Here earliest-inclusion rounds increase strictly along
the issued sequence by
\cref{lemma:following-reconfiguration-earliest-round,%
lemma:earliest-install-lower-bound}.

If an assigned singleton $\{R_c\}$ has $c\geq l+2$, then $R_{l+2}$
exists, and \cref{lemma:following-reconfiguration-earliest-round} gives
\[
  r\leq t+5<\mathsf{earliest}(R_{l+2})
  \leq\mathsf{earliest}(R_c).
\]
A correct process assigns $\{R_c\}$ only after completing the transition
into $R_c$. By \cref{lemma:earliest-install-lower-bound}, its install
transition round is at least $\mathsf{earliest}(R_c)$, so such a
singleton can be assigned only to a round greater than
$\mathsf{earliest}(R_c)$, again a contradiction.

Thus every assigned value is $\{R_l\}$, $\{R_l,R_{l+1}\}$, or
$\{R_{l+1}\}$. \Cref{lemma:same-transition-enabled-configs} forbids the
two singleton enabled-configs values from coexisting.

Suppose instead that all enabled-configs values in the round are singletons.
If $r=1$, all assigned values equal $\{R_1\}$, so assume $r\geq2$. Consider
two of them, $\{R_a\}$ and $\{R_b\}$, with $a<b$. If $b=a+1$,
\cref{lemma:same-transition-enabled-configs} excludes their coexistence. If
$b\geq a+2$, instance authorization implies that the process assigning
$\{R_a\}$ is a correct member of $\mathsf{members}(R_a)$. Let $p_a$ be
that process; then $R_a\in\mathit{enabled\text{-}configs}_a[r]$, so
\cref{lemma:old-configuration-round-cutoff} gives
\[
  r<\mathsf{earliest}(R_{a+2}).
\]
Consider the correct process assigning $\{R_b\}$. It assigned this
singleton only after completing the transition $R_{b-1}\rightarrow R_b$,
either as a retained member of $\mathsf{members}(R_{b-1})$ or as the
process joining in that transition; a process removed by that transition
stops round advancement and never assigns $\{R_b\}$. In both cases its
install transition round for $R_{b-1}\rightarrow R_b$ is at least
$\mathsf{earliest}(R_b)$ by \cref{lemma:earliest-install-lower-bound},
and it assigns $\{R_b\}$ only for rounds greater than that install
transition round. Hence $r>\mathsf{earliest}(R_b)$. Since $b\geq a+2$ and
earliest-inclusion rounds increase along the issued sequence, as noted
above, $\mathsf{earliest}(R_b)\geq\mathsf{earliest}(R_{a+2})$, so
$r>\mathsf{earliest}(R_{a+2})$ --- the opposite inequality. This
contradiction excludes nonadjacent singleton values.
Thus all singleton enabled-configs values are identical. The two cases prove
both claims.
\end{proof}

By \cref{lemma:dynamic-dag-structured-broadcast}, DAG-Consistency and
Causal Closure hold for every block inserted by a correct process.
Consequently, the ancestor relation between two stored blocks is
independent of the correct process storing them; we omit the process
subscript when reasoning about fixed block objects. By
\cref{alg:brb}, a correct process $p_i$ votes for a configuration-scoped BRB
instance declared under $R_l$ in round $r$ only when
$R_l\in\mathit{enabled\text{-}configs}_i[r]$. Invocations of
$\texttt{Elect}_i$ are serialized by the
requirement in \cref{alg:leader-election_2}.
For a correct process $p_i$, write $\texttt{WaveLeader}_i(R_l,w)$ for the
result of evaluating $\texttt{WaveLeader}(R_l,w)$ at $p_i$.

\subsubsection{Leader Uniqueness and Direct-Anchor Reachability}

\paragraph*{Unique Wave Leaders and Their Reachability from Later Blocks.}
We first show that, in one wave, at most one configuration can have
leader-voting blocks in more distinct broadcaster slots than its fault bound.
It follows that a wave has at most one accepted leader. We then show that
every later-wave block reaches at least $f_l+1$ leader-voting blocks that
reach an earlier direct anchor. This makes backward scans deterministic and
their leader sequences nested.

\begin{lemma}[In one wave, two configurations cannot each have leader-voting
blocks in more slots than their respective fault bounds]
\label{lemma:single-supported-vote-type}
Let $p_i$ and $p_j$ be correct processes and let $w\geq1$. Suppose
$DAG_i$ contains $R_a$-typed leader-voting blocks for wave $w$ in at least
$f_a+1$ distinct round-broadcaster slots, and $DAG_j$ contains $R_b$-typed
leader-voting blocks for that wave in at least $f_b+1$ distinct
round-broadcaster slots.
Then $R_a=R_b$.
\end{lemma}

\begin{proof}
Each collection contains a block whose broadcaster is correct. Let
$X=\langle m_X,\mathsf{last}(w),\mu_X,C_X\rangle$ and
$Y=\langle m_Y,\mathsf{last}(w),\mu_Y,C_Y\rangle$ be such blocks.
By the definitions of $R_a$-typed and $R_b$-typed leader-voting blocks,
\[
\begin{aligned}
\mathsf{local}(\mu_X)&=R_a,
&\mathsf{enabled}(\mu_X)&=\{R_a\},\\
\mathsf{local}(\mu_Y)&=R_b,
&\mathsf{enabled}(\mu_Y)&=\{R_b\}.
\end{aligned}
\]
Let $r=\mathsf{last}(w)$, and let $p_x$ and $p_y$ be the broadcasters
of $X$ and $Y$. Correct broadcasters set
$\mathsf{enabled}(\mu_X)=\mathit{enabled\text{-}configs}_x[r]$ and
$\mathsf{enabled}(\mu_Y)=\mathit{enabled\text{-}configs}_y[r]$.
By \cref{lemma:same-round-enabled-configs}, the same logical round cannot
be singleton-enabled for two different configurations at correct
broadcasters. Applying this fact to $X$ and $Y$ gives
$R_a=R_b$.
\end{proof}

\begin{lemma}[All blocks identified as direct anchors or accepted during
backward scans for one wave are identical]
\label{lemma:unique-accepted-leader}
For every wave $w$, all blocks that correct processes either identify as a
direct anchor for $w$ or accept for $w$ during a backward scan are identical.
\end{lemma}

\begin{proof}
Let $p_i$ and $p_j$ be correct processes. Let $B$ be selected as the coin
candidate and then identified as a direct anchor at $p_i$ or accepted there
during a backward scan, and define $\widehat B$ analogously at $p_j$.
Suppose $B$ and $\widehat B$ are obtained using $R_l$-typed and $R_s$-typed
leader-voting blocks, respectively. A direct
anchor identified using $R_l$-typed blocks has at least $q_l\geq f_l+1$
distinct $R_l$-typed leader-voting blocks that reach it. If a backward scan
with current anchor $Y$ accepts $B$, at least $f_l+1$ distinct
$R_l$-typed leader-voting blocks $X$ satisfy
$Y\leadsto X\leadsto B$. Thus either way of obtaining $B$
provides at least $f_l+1$ $R_l$-typed leader-voting blocks. The same
argument provides at least $f_s+1$ $R_s$-typed leader-voting blocks for
$\widehat B$. All blocks counted for $B$ lie in round
$\mathsf{last}(w)$ of $DAG_i$, and all blocks counted for $\widehat B$ lie
in that same round of $DAG_j$. By DAG-Integrity, distinct blocks within
each collection occupy distinct round-broadcaster slots. Hence the two
collections satisfy
the premises of \cref{lemma:single-supported-vote-type}, which gives
$R_l=R_s$.

The two candidates are therefore
$\texttt{WaveLeader}_i(R_l,w)$ and $\texttt{WaveLeader}_j(R_l,w)$.
Threshold-signature uniqueness implies
that any two correct processes that reconstruct the coin for
$(R_l,w)$ obtain the same signature and selected member, hence the same
round-broadcaster slot in round $\mathsf{first}(w)$.
DAG-Consistency makes the block occupying that slot unique, so both
invocations select the same block and $B=\widehat B$. Hence all
direct-anchor identifications and backward-scan acceptances for wave $w$
select the same block.
\end{proof}

By \cref{lemma:unique-accepted-leader}, a wave has at most one direct
anchor. If wave $k$ has one, denote it by $D_k$.
Thus every backward scan begun from wave $k$ starts with scan anchor $D_k$.

The next lemma is the dynamic counterpart of DAG-Rider's leader
connectivity lemma~\cite[Lemma 2]{dagrider}.

\begin{lemma}[If a direct anchor is identified using $R_l$-typed
leader-voting blocks, every certified last-round block is declared under
$R_l$ by an $R_l$-member]
\label{lemma:direct-anchor-round-slots}
Let $w$ be a wave, and suppose a correct process $p_d$ identifies a block
$D_w$ as a direct anchor for $w$ using $R_l$-typed leader-voting blocks.
Let $r=\mathsf{last}(w)$. For every block
$B=\langle m,r,\mu,C\rangle$ such that $C$ is a valid BRB quorum
certificate for $\langle m,r,\mu\rangle$ under the configuration declared
by $\mu$,
\[
  \mathsf{local}(\mu)=R_l
  \qquad\text{and}\qquad
  \mathsf{ID}(\mu)\in\mathsf{members}(R_l).
\]
In particular, its round-broadcaster slot is
$(r,\mathsf{ID}(\mu))$, where
$\mathsf{ID}(\mu)\in\mathsf{members}(R_l)$.
\end{lemma}

\begin{proof}
Put $r=\mathsf{last}(w)$. Identifying $D_w$ as a direct anchor gives a set
\[
\mathcal S=
\left\{
X\;\middle|\;
X\text{ is an $R_l$-typed leader-voting block for $w$ in }DAG_d,
\ X\leadsto D_w
\right\}
\qquad\text{with}\qquad |\mathcal S|\geq q_l.
\]
By the definition of a direct anchor and DAG-Integrity, the blocks counted in
$\mathcal S$ occupy distinct round-broadcaster slots. At most $f_l$ of their
broadcasters are corrupted, so the set $\mathcal G$ of correct broadcaster
processes represented in $\mathcal S$ satisfies
\[
  |\mathcal G|\geq q_l-f_l=n_l-2f_l.
\]
For each $p_i\in\mathcal G$, write
$X_i=\langle m_i,r,\mu_i,C_i\rangle$ for its block in $\mathcal S$.
Because $X_i$ is $R_l$-typed, it records
$\mathsf{enabled}(\mu_i)=\{R_l\}$. Correctness of the broadcaster gives
\[
  \mathit{enabled\text{-}configs}_i[r]=\{R_l\},
\]
so $p_i$ issues no round-$r$ BRB vote share for an instance declared under
a configuration other than $R_l$.

Suppose, for contradiction, that a valid round-$r$ certificate is declared
under $R_c\neq R_l$. Since $q_c>f_c$, that certificate contains a share
from a correct voter $p_v$. The voter applied
$\mathsf{ValidInstance}_v(r,\mu)$, and hence
$R_c\in\mathit{enabled\text{-}configs}_v[r]$. Choose any
$p_i\in\mathcal G$. Since $p_i$ and $p_v$ are correct,
\cref{lemma:same-round-enabled-configs} applies to their round-$r$
enabled-configs values. The value at $p_i$ is $\{R_l\}$, whereas the value
at $p_v$ contains $R_c\neq R_l$. Hence the latter value must comprise the
two configurations of one valid transition:
\[
  \mathit{enabled\text{-}configs}_v[r]=\{R_l,R_c\}.
\]

For some process $p$, the memberships of $R_l$ and $R_c$ are either
identical or satisfy
\[
  \mathsf{members}(R_c)=\mathsf{members}(R_l)\cup\{p\}
  \qquad\text{or}\qquad
  \mathsf{members}(R_c)=\mathsf{members}(R_l)\setminus\{p\}.
\]
Consequently,
\[
|\mathcal G\cap\mathsf{members}(R_c)|
\geq
\begin{cases}
 n_l-2f_l,
   &\mathsf{members}(R_c)=\mathsf{members}(R_l)
     \text{ or }
     \mathsf{members}(R_c)=\mathsf{members}(R_l)\cup\{p\},\\
 n_l-2f_l-1,
   &\mathsf{members}(R_c)=\mathsf{members}(R_l)\setminus\{p\}.
\end{cases}
\]
From $f_c=\lfloor(n_c-1)/3\rfloor$, the first lower bound is at least
$f_c+1$ in both the unchanged-membership and addition cases, and the second
is at least $f_c+1$ in the removal case. Thus at least $f_c+1$ correct
members of $R_c$ refuse every $R_c$-declared instance in round $r$. At
most
\[
  n_c-(f_c+1)=q_c-1
\]
shares remain available, contradicting the assumed valid $R_c$
certificate. Hence every valid round-$r$ certificate is declared under
$R_l$.

Finally, let $p_v$ be a correct voter whose share occurs in such a
certificate declared under $R_l$; one exists because $q_l>f_l$. Since
$p_v$ applied $\mathsf{ValidInstance}_v(r,\mu)$, the
broadcaster-membership clause gives
\[
  \mathsf{ID}(\mu)\in
  \mathsf{members}(\mathsf{local}(\mu))
  =\mathsf{members}(R_l).
\]
This proves both claims.
\end{proof}

\begin{lemma}[Every later-wave block DAG-delivered by a correct process
reaches $2f_l+1$ distinct blocks broadcast by $R_l$-members in the direct
anchor's last round]
\label{lemma:later-block-last-round-ancestry}
Let $w$ be a wave, and suppose a correct process $p_d$ identifies a block
$D_w$ as a direct anchor for $w$ using $R_l$-typed leader-voting blocks.
Let $r=\mathsf{last}(w)$. Let $p_i$ be correct and let
$Y\in DAG_i[s]$ be DAG-delivered by $p_i$, where $s$ belongs to a wave
$w'>w$. Use $\mathcal H_Y(r)$ as the abbreviation for
$\mathcal H_{Y,i}(r)$ fixed in the Proof Conventions.
Then
\[
  |\mathcal H_Y(r)|\geq2f_l+1.
\]
For every
$X=\langle m_X,r,\mu_X,C_X\rangle\in\mathcal H_Y(r)$,
\[
  \mathsf{local}(\mu_X)=R_l
  \qquad\text{and}\qquad
  \mathsf{ID}(\mu_X)\in\mathsf{members}(R_l).
\]
Moreover, the blocks in $\mathcal H_Y(r)$ have distinct broadcaster
identifiers: for any $X,X'\in\mathcal H_Y(r)$,
\[
  X\neq X'
  \quad\Longrightarrow\quad
  \mathsf{ID}(\mu_X)\neq\mathsf{ID}(\mu_{X'}).
\]
\end{lemma}

\begin{proof}
We first show that a correct broadcaster remains active through round
$r+1$ and has an enabled-configs value containing $R_l$ in that round. Let
$\mathcal S$ be the set of $R_l$-typed leader-voting
blocks used to identify the direct anchor. Since
$|\mathcal S|\geq q_l>f_l$, some
$U\in\mathcal S$ has a correct broadcaster $p_u$. Because $U$ is
$R_l$-typed and $p_u$ is correct,
\[
  \mathit{enabled\text{-}configs}_u[r]=\{R_l\}.
\]
The process is active when it produces $U$. If a correct process invokes
$\texttt{InstallTransition}$ with install transition round $a$ for a transition
that removes it, \cref{alg:reconfiguration} assigns the authenticated pair
in rounds $a+1$ through $a+4$. The process may produce blocks only through
round $a+3$ and stops before producing a round-$(a+4)$ block. Since $U$
records the singleton $\{R_l\}$, no earlier invocation can make $p_u$
stop before it advances from round $r$ to round $r+1$. Hence $p_u$ remains
active through that advancement. If a transition from $R_l$ is underway,
this advancement precedes its completion or departure, so
\cref{corollary:transition-round-progress} advances $p_u$. Otherwise,
initialization when $l=1$, or conclusion~3 of
\cref{lemma:configuration-step-closure} for the transition into $R_l$ when
$l>1$, supplies a round reached by every correct $R_l$-member. Sequential
round execution and repeated application of
\cref{corollary:round-advancement} while $p_u$ has not invoked
$\texttt{InstallTransition}$ for a transition from $R_l$ then supply at least
$q_l$ correct processes in each round whose prescribed appends reach
$\texttt{BRB-Broadcast}$ and advance $p_u$; if $p_u$ invokes
$\texttt{InstallTransition}$ first,
\cref{corollary:transition-round-progress} supplies the remaining
advancement.
Thus $p_u$ eventually advances from round $r$ to round $r+1$. Its next
enabled-configs value is either $\{R_l\}$ or $\{R_l,R_s\}$ for a valid
transition $R_l\rightarrow R_s$.
In either case,
\[
  R_l\in\mathit{enabled\text{-}configs}_u[r+1].
\]

Now fix a correct process $p_z$ and let
$Z=\langle m_Z,r+1,\mu_Z,C_Z\rangle\in DAG_z[r+1]$ be arbitrary. Put
$R_c=\mathsf{local}(\mu_Z)$. The valid
certificate $C_Z$ declared under $R_c$ contains a share from a correct voter
$p_v$.
The voter applied $\mathsf{ValidInstance}_v(r+1,\mu_Z)$, so
\[
  R_c\in\mathit{enabled\text{-}configs}_v[r+1].
\]
Together with the value at $p_u$,
\cref{lemma:same-round-enabled-configs} implies that either $R_c=R_l$ or
$R_c=R_s$, where either $R_l\rightarrow R_s$ or
$R_s\rightarrow R_l$ is a valid transition.

If the special parent-count rule for a block produced by
$\texttt{DAG-Transition-Append}$ does not apply, then
$|\mathsf{parents}(\mu_Z)|=q_c$: this is $q_l$ when $R_c=R_l$ and $q_s$
when $R_c=R_s$. If that special rule applies, the block declares the
successor configuration but uses the predecessor's quorum size. The correct
voter then requires its round-$(r+1)$ enabled-configs value to equal the
authenticated pair recorded by $\mu_Z$. Same-round alignment with the value
at $p_u$ makes that pair $\{R_l,R_s\}$, where either
$R_l\rightarrow R_s$ or $R_s\rightarrow R_l$ is a valid transition. The
parent count is therefore $q_l$, or it is $q_s$ for a configuration connected
to $R_l$ by one valid transition.

In all cases the count is at least $2f_l+1$. We have
$q_l\geq2f_l+1$. Whenever the $q_s$ case occurs, equal memberships give
$q_s=q_l$. If
$\mathsf{members}(R_s)=\mathsf{members}(R_l)\cup\{p\}$ for some process
$p$, then $q_s\geq q_l$. If
$\mathsf{members}(R_s)=\mathsf{members}(R_l)\setminus\{p\}$ for some process
$p$, then
$q_s\in\{q_l,q_l-1\}$, where the latter is possible only when
$q_l\geq2f_l+2$. Causal
Closure places all distinct parents of $Z$ in the local
round-$r$ layer, and $Z$ reaches each of them. Therefore
\[
  |\mathcal H_Z(r)|\geq2f_l+1.
\]

The preceding argument proves the base case $v=r+1$ for every block
DAG-delivered by a correct process. We now prove by induction on $v>r+1$
that, for every correct process $p_a$, every block $T\in DAG_a[v]$ satisfies
\[
  |\mathcal H_T(r)|\geq2f_l+1.
\]
Choose any round-$(v-1)$ parent $P$ of $T$. Causal Closure gives
$P\in DAG_a[v-1]$, so the induction hypothesis yields
$|\mathcal H_P(r)|\geq2f_l+1$. Since $T\leadsto_a P$, causal transitivity
gives
\[
  \mathcal H_P(r)\subseteq\mathcal H_T(r),
\]
and hence $|\mathcal H_T(r)|\geq2f_l+1$. If $s=r+1$, the base case gives
the claim for $Y$. If $s>r+1$, take $p_a=p_i$ and $v=s$ in the induction.
Thus the bound holds for $Y$ in all cases.

Every $X\in\mathcal H_Y(r)$ is DAG-delivered at $p_i$ and therefore carries
a valid certificate. Applying
\cref{lemma:direct-anchor-round-slots} gives
\[
  \mathsf{local}(\mu_X)=R_l
  \qquad\text{and}\qquad
  \mathsf{ID}(\mu_X)\in\mathsf{members}(R_l).
\]
Finally, if two blocks $X,X'\in\mathcal H_Y(r)$ had the same broadcaster
identity, they would occupy the same round-broadcaster slot. DAG-Integrity
at $p_i$ then gives $X=X'$. Thus
distinct blocks in $\mathcal H_Y(r)$ have distinct broadcaster identities.
\end{proof}

\begin{lemma}[The causal history of every later-wave block contains $f_l+1$
leader-voting blocks that reach the earlier direct anchor]
\label{lemma:direct-leader-carry-over}
Let $w$ be a wave, and suppose a correct process $p_d$ identifies a block
$D_w$ as a direct anchor for $w$ using $R_l$-typed leader-voting blocks.
Let $p_i$ be any correct process, and let $Y$ be a block DAG-delivered by
$p_i$ whose round belongs to a wave $w'>w$. Then $DAG_i$ contains at least
$f_l+1$ $R_l$-typed leader-voting blocks $X$ for wave $w$, with pairwise
distinct broadcaster identities, such that
\[
  Y\leadsto_i X\leadsto_i D_w.
\]
\end{lemma}

\begin{proof}
Put $r=\mathsf{last}(w)$, and let $\mathcal S$ be the set of at least $q_l$
$R_l$-typed leader-voting blocks that reach $D_w$ in the definition of a
direct anchor. Define
\[
\begin{aligned}
I_{\mathcal S}
  &=
  \left\{
    \mathsf{ID}(\mu_X)
    \;\middle|\;
    X=\langle m_X,r,\mu_X,C_X\rangle\in\mathcal S
  \right\},\\
I_Y
  &=
  \left\{
    \mathsf{ID}(\mu_X)
    \;\middle|\;
    X=\langle m_X,r,\mu_X,C_X\rangle\in\mathcal H_Y(r)
  \right\}.
\end{aligned}
\]
The direct-anchor test counts distinct round-broadcaster slots, so
$|I_{\mathcal S}|\geq q_l$. By
\cref{lemma:later-block-last-round-ancestry},
$|I_Y|\geq2f_l+1$. Both sets are subsets of
$\mathsf{members}(R_l)$, which has size $n_l$. Therefore
\[
\begin{aligned}
|I_{\mathcal S}\cap I_Y|
&\geq |I_{\mathcal S}|+|I_Y|-n_l\\
&\geq (n_l-f_l)+(2f_l+1)-n_l=f_l+1.
\end{aligned}
\]
For each identity in this intersection, $\mathcal S$ and
$\mathcal H_Y(r)$ contain a block in the same round-broadcaster slot.
DAG-Consistency makes those two blocks identical. Hence at least $f_l+1$
blocks lie in $\mathcal S\cap\mathcal H_Y(r)$, and every such block $X$
satisfies
\[
  Y\leadsto_i X\leadsto_i D_w.
\]
Their broadcaster identities are pairwise distinct by construction.
\end{proof}

\begin{corollary}[A later backward scan accepts an earlier direct anchor]
\label{corollary:direct-anchor-scan-acceptance}
Suppose $D_w$ is a direct anchor identified using $R_l$-typed
leader-voting blocks. Let a backward scan at a correct process $p_i$ have
current anchor $Y$, where $Y$ belongs to a wave later than $w$. Then the scan
accepts $D_w$ when it reaches wave $w$.
\end{corollary}

\begin{proof}
Causal Closure gives the scanning process $D_w$, the $f_l+1$ distinct
$R_l$-typed leader-voting blocks from
\cref{lemma:direct-leader-carry-over}, and their coin shares. These blocks
satisfy the backward-scan acceptance condition for $D_w$, so the process
reconstructs the unique coin for $(R_l,w)$ and obtains
$\texttt{WaveLeader}_i(R_l,w)=D_w$.
By \cref{lemma:unique-accepted-leader}, it cannot accept a different block
for wave $w$.
\end{proof}

\subsubsection{Backward Scans, Finalized Leader Sequences, and Total Order}

\paragraph*{Deterministic Backward Scans and Prefix-Comparable Leader
Sequences.}
For a wave $k$ with direct anchor $D_k$ and a correct process $p_i$ with
$D_k\in DAG_i$, define $\Phi_{i,k}$ by a conceptual, read-only execution of
the backward-scan rule in $\texttt{Elect}_i$. Initialize the temporary
sequence to $\langle D_k\rangle$ and the scan anchor to $D_k$, and apply the
same acceptance test to waves $k-1,k-2,\ldots,1$. Unlike an actual
invocation, this full scan ignores
$\mathit{processed\text{-}wave}_i$ and does not modify any persistent state.

\begin{lemma}[Full backward scans rooted at the same direct anchor produce the
same leader sequence]
\label{lemma:deterministic-leader-scans}
Let $k$ be a wave with direct anchor $D_k$. For any correct processes
$p_i$ and $p_j$, evaluate $\Phi_{i,k}$ and $\Phi_{j,k}$ at arbitrary local
states in which $D_k\in DAG_i$ and $D_k\in DAG_j$, respectively. Then
\[
  \Phi_{i,k}=\Phi_{j,k}.
\]
Thus the full backward scan rooted at $D_k$ is independent of both the
correct process executing it and any blocks outside $D_k$'s causal history
that are delivered later.
\end{lemma}

\begin{proof}
Fix a wave $w<k$ and a scan anchor $Y$ stored by correct processes $p_i$
and $p_j$. By DAG-Consistency and Causal Closure, both store the same
causal history of $Y$, including its blocks, metadata, certificates,
coin shares, and parent references.

Although $\mathsf{WaveConfigs}_x(w)$, $\mathsf{Shares}_x(R_c,w)$, and
$\texttt{WaveLeader}_x(R_c,w)$ inspect all of $DAG_x$ for $x\in\{i,j\}$,
the backward-scan candidate set depends only on $Y$'s causal history.
To see this, suppose $R_c$ contributes a candidate at $p_i$:
\[
  B'=\texttt{WaveLeader}_i(R_c,w)\neq\bot,
  \qquad
  \mathsf{HistorySupport}_i(Y,B',R_c,w)\geq f_c+1.
\]
The counted leader-voting blocks have distinct broadcaster identities
and satisfy $Y\leadsto X\leadsto B'$. Hence these blocks, their valid
coin shares, and $B'$ itself all lie in $Y$'s causal history and are
stored by $p_j$. The shares suffice to reconstruct the coin for
$(R_c,w)$. Threshold-signature uniqueness gives the same selected
round-broadcaster slot, and DAG-Consistency gives the same block $B'$
in that slot. Thus $\texttt{WaveLeader}_j(R_c,w)=B'$.
The values $\mathsf{HistorySupport}_i(Y,B',R_c,w)$ and
$\mathsf{HistorySupport}_j(Y,B',R_c,w)$ are equal because both count
exactly the same identities in the shared causal history of $Y$.
Therefore $R_c$ contributes the same candidate at $p_j$.
Applying this argument in both directions shows that the candidate
sets are identical.

The argument applies to arbitrary local states containing $Y$,
including states before and after later deliveries at one process.
In particular, a block outside $Y$'s causal history has history support
zero, while any candidate meeting the threshold already has enough
coin shares in that history. Later deliveries therefore cannot
change the candidate set.

Finally, a contributing configuration $R_c$ supplies at least $f_c+1$
$R_c$-typed leader-voting blocks in distinct round-broadcaster slots.
By \cref{lemma:single-supported-vote-type}, at most one configuration
can contribute a candidate for wave $w$. The candidate set is therefore
either empty or the same singleton at every correct process storing
$Y$. A fixed scan anchor consequently gives the same accept-or-skip
decision and, upon acceptance, the same next anchor. Since both full
backward scans start at $D_k$, induction over descending wave numbers
gives $\Phi_{i,k}=\Phi_{j,k}$.
\end{proof}

By \cref{lemma:deterministic-leader-scans}, denote the common sequence
$\Phi_{i,k}$ by $\Phi_k$.

\begin{lemma}[The leader sequence from an earlier direct anchor is a prefix of
the sequence from a later direct anchor]
\label{lemma:nested-leader-scans}
Write $\preceq$ for the prefix relation. If $h<k$ are waves with direct
anchors, then
\[
  \Phi_h\preceq\Phi_k.
\]
\end{lemma}

\begin{proof}
Now let $h<k$ be waves with direct anchors. When the full backward scan rooted at
$D_k$ reaches wave $h$, its scan anchor belongs to a wave whose index
is greater than $h$. By
\cref{corollary:direct-anchor-scan-acceptance}, the scan accepts $D_h$ and
sets its anchor to $D_h$. Every subsequent
decision below $h$ is then identical to the full backward scan rooted at $D_h$.
Since earlier leaders are prepended, $\Phi_h$ is a prefix of $\Phi_k$.
\end{proof}

Set $\Phi_0=\langle\rangle$. By
\cref{lemma:nested-leader-scans}, $\Phi_0$ and all sequences $\Phi_k$ for
waves $k$ with direct anchors are nested under the prefix relation. Hence
there is a unique finite or infinite sequence $\Phi$ for which every
$\Phi_k$ is a prefix and whose elements are exactly those appearing in the
$\Phi_k$. We call $\Phi$ the \emph{canonical leader sequence}. Every finite
prefix of $\Phi$ is a prefix of some $\Phi_k$, or is empty if no direct
anchor exists.

\begin{lemma}[Every correct process's finalized leader sequence is always a
canonical leader prefix]
\label{lemma:canonical-local-prefixes}
After initialization, at every local state of a correct process $p_i$,
$\mathit{finalized\text{-}values}_i$ is $\Phi_0$ or $\Phi_k$ for some
wave $k$ with a direct anchor.
Immediately after a serialized invocation of $\texttt{Elect}_i$ returns,
let $v=0$ if no invocation $\texttt{Elect}_i(D_x,x)$ has completed at
$p_i$ for any wave $x$; otherwise, let $v$ be the greatest wave index $x$
for which such an invocation has completed. Then
$\mathit{finalized\text{-}values}_i=\Phi_v$, where either $v=0$ or $v$ is
a wave with a direct anchor.
\end{lemma}

\begin{proof}
Induct on the serialized $\texttt{Elect}_i$ invocations. Initially no
$\texttt{Elect}_i$ invocation has completed and
$\mathit{finalized\text{-}values}_i=\Phi_0$. After initialization,
$\texttt{Elect}_i$ is the only procedure that modifies this sequence. Assume
the claim holds immediately before a serialized invocation
$\texttt{Elect}_i(D_k,k)$, and let $v$ be the greatest wave for which an
$\texttt{Elect}_i$ invocation has completed, or $0$ if there is none. Wave
$k$ has a direct anchor because the procedure is invoked only after $p_i$
identifies $D_k$ as one.

If $k\leq v$, the invocation returns without changing the sequence. The
greatest completed wave remains $v$, and the sequence remains $\Phi_v$.
Suppose instead that $k>v$. If $v=0$, the invocation executes exactly
the full backward scan defining $\Phi_k$ and changes the finalized leader sequence to
$\Phi_k$. If $v>0$, the induction hypothesis gives
$\mathit{finalized\text{-}values}_i=\Phi_v$, and wave $v$ has a direct
anchor. The invocation performs exactly the iterations of the full backward
scan for $\Phi_k$ above wave $v$. When that full backward scan reaches $v$,
\cref{corollary:direct-anchor-scan-acceptance} makes it accept $D_v$, after which
its decisions below $v$ produce $\Phi_v$. Since
$\Phi_v\preceq\Phi_k$, there is a unique suffix $\Psi$ such that
$\Phi_k=\Phi_v\cdot\Psi$. The invocation appends exactly $\Psi$, changing
the finalized leader sequence from $\Phi_v$ to $\Phi_k$.
In the case $k>v$, wave $k$ is the greatest argument of a completed
$\texttt{Elect}_i$ invocation when the call returns.

During an invocation, the finalized leader sequence is unchanged before its
single append; immediately afterward it is already the next canonical
leader prefix. Thus the invariant holds at every local state.
\end{proof}

\begin{corollary}[Every active correct process finalizes at most one block
per wave, never finalizes a block twice, and finalizes only blocks already
in its DAG]
\label{corollary:leader-integrity}
For every wave $w$ and active correct process $p_i$,
$\mathit{finalized\text{-}values}_i$ contains at most one block from $w$ and contains no
block more than once. Every block in the sequence has already been
inserted into the process's local DAG.
\end{corollary}

\begin{proof}
Each backward scan visits a wave at most once, and an invocation for a later
direct anchor scans only waves whose indices exceed the greatest wave already
processed.
Together with \cref{lemma:unique-accepted-leader}, this lets a correct
process $p_i$ append at most one leader for each wave. Serialization in
\cref{alg:leader-election_2} prevents concurrent invocations from
duplicating that output. Every appended block was returned by
$\texttt{WaveLeader}_i(R_l,w)$ for some configuration $R_l$ and wave $w$. By its
definition, every non-$\bot$ return is a locally available first-round
DAG block and has therefore already been inserted into $DAG_i$. Thus
every candidate for wave $w$ belongs to round $\mathsf{first}(w)$. Distinct
waves have distinct first rounds, so the same block cannot be selected for
two different waves. Thus all requirements of Integrity in
\cref{def:dynamic-dag-bab} hold.
\end{proof}

\begin{theorem}[Finalized leader sequences at active correct processes are
always prefix-comparable]
\label{theorem:leader-total-order}
At every time, the finalized leader sequences of any two active correct
processes are prefix-comparable. Consequently, if the sequence of an active
correct process $p_i$ contains a block $B$ before a distinct block $B'$,
then the sequence of every active correct process $p_j$ that contains $B'$
also contains $B$ before $B'$.
\end{theorem}

\begin{proof}
By \cref{lemma:canonical-local-prefixes}, at every point each correct
process's finalized leader sequence is a canonical leader prefix. By
\cref{lemma:nested-leader-scans}, all canonical leader prefixes are
pairwise comparable under $\preceq$. Both local sequences are prefixes of
$\Phi$. Since $B$ precedes $B'$ in $\Phi$, every prefix containing $B'$
contains all earlier elements of $\Phi$, including $B$. This is exactly the
stated total-order property.
\end{proof}

\subsection{Probabilistic Progress}
\label{sec:probabilistic-progress}

\subsubsection{Progress After Stabilization at $R_L$}

\paragraph*{Eventual Stability.}
\begin{corollary}[After some round, every correct $R_L$-member that remains
active enables only $R_L$, and every certificate is declared under $R_L$]
\label{corollary:eventual-stable-suffix}
There exist a final configuration $R_L$ and a round $r_\star$ such
that every correct process $p_i\in\mathsf{members}(R_L)$ that remains active
has installed $R_L$ and, for every round $r\geq r_\star$ that it executes,
\[
  \mathit{enabled\text{-}configs}_i[r]=\{R_L\}.
\]
Every certified block in a round $r\geq r_\star$ is declared under $R_L$
and has a broadcaster identity in $\mathsf{members}(R_L)$.
\end{corollary}

\begin{proof}
Finiteness of the reconfiguration sequence gives a final configuration
$R_L$. If $L=1$, then $R_L=R_1$ is installed at initialization. If $L>1$,
the transition into $R_L$ completes by
\cref{lemma:reconfiguration-completion}. In either case, no later
reconfiguration message exists.

Assume first that $L>1$. For every earlier configuration $R_c$, only members of
$\mathsf{members}(R_c)$ can vote for a certificate declared under $R_c$.
Completion of the transition out of each $R_c$ gives a finite logical-round
cutoff $\rho_c$ beyond which no correct member of $R_c$ assigns an
enabled-configs value containing $R_c$. For $c\leq L-2$, this follows from
\cref{lemma:old-configuration-round-cutoff}. For $c=L-1$, every correct
retained old member has a finite install transition round $a_i$ by
\cref{lemma:reconfiguration-completion}, and every correct removed old
member has one by \cref{lemma:removed-process-departure}. An invocation with
install transition round $a_i$ assigns values containing $R_{L-1}$ only
through round $a_i+4$; afterward a retained member assigns only $\{R_L\}$,
while a removed member executes no later protocol rounds after its permitted
departure. Taking the maximum of these finitely many bounds gives
$\rho_{L-1}$. Because every configuration is finite, one cutoff suffices
for all of its correct members.

Completion of the transition into $R_L$ also gives a finite cutoff beyond
which every correct $R_L$-member that remains active has installed $R_L$
and assigns only $\{R_L\}$. When $L>1$, choose $r_\star$ larger than these
finitely many logical-round cutoffs. When $L=1$, take $r_\star=1$;
initialization and the absence of a later configuration imply that every
enabled-configs value assigned by a correct process is $\{R_1\}$.
In either case, every correct process
$p_i\in\mathsf{members}(R_L)$ that remains active then assigns
\[
  \mathit{enabled\text{-}configs}_i[r]=\{R_L\}
  \qquad\text{for every round }r\geq r_\star\text{ that it executes}.
\]
A certificate in round $r\geq r_\star$ declared under an earlier
configuration $R_c$ would require
$q_c>f_c$ votes and hence a vote from a correct process whose round-$r$
enabled-configs value contains $R_c$, which is impossible above the chosen
cutoff. Since no later configuration exists, every certificate in a round
$r\geq r_\star$ is declared under $R_L$. The BRB
broadcaster-membership check then restricts it to one of the $n_L$
round-broadcaster slots of $R_L$.
\end{proof}

\begin{corollary}[Every correct $R_L$-member that remains active completes
every round $r\geq r_\star$]
\label{corollary:stable-suffix-round-progress}
Under the conclusion of \cref{corollary:eventual-stable-suffix}, for every
correct process $p_i\in\mathsf{members}(R_L)$ that remains active and every
$r\geq r_\star$, process $p_i$ eventually DAG-delivers at least $q_L$
admissible round-$r$ blocks in $\mathsf{ParentBlocks}_i(R_L,r)$,
completes round $r$, and enters round $r+1$.
\end{corollary}

\begin{proof}
If $L=1$, take $r_0=1$; every active correct process in
$\mathsf{members}(R_1)$ enters that round, and at least $q_L$ of them invoke
the prescribed append and reach its parent-free $\texttt{BRB-Broadcast}$
call. If $L>1$,
\cref{lemma:configuration-level-induction} supplies the induction hypothesis
for $R_{L-1}$, and the third conclusion of
\cref{lemma:configuration-step-closure}, applied to
$R_{L-1}\rightarrow R_L$, supplies such a round $r_0$. Because the conclusion
of \cref{corollary:eventual-stable-suffix} remains true when its cutoff is
increased, replace $r_\star$ by $\max\{r_\star,r_0\}$. BRB Agreement for
every round-broadcaster slot from
\cref{lemma:brb-safety} supplies the Agreement premise of
\cref{corollary:round-advancement}. We now induct on rounds $r\geq r_0$.
The invariant is that every correct $R_L$-member that remains active
eventually enters round $r$, and at least $q_L$ active correct
$R_L$-members invoke the prescribed round-$r$ $\texttt{DAG-Append}$
operation and reach its $\texttt{BRB-Broadcast}$ call. The choice of $r_0$
gives the base case. For the induction step, apply
\cref{lemma:dag-quorum-dissemination} to those $q_L$ broadcasters. Every
correct $R_L$-member that remains active eventually DAG-delivers at least
$q_L$ admissible round-$r$ blocks in distinct broadcaster slots and hence
satisfies
\[
  |\mathsf{ParentBlocks}_i(R_L,r)|\geq q_L.
\]
By \cref{corollary:round-advancement}, every such process then completes
round $r$ and enters round $r+1$. At least $q_L$ correct members remain
active by the assumption in \cref{sec:problem-definition} that at least
$q_L$ correct members of $R_L$ remain active, and each invokes its prescribed
append upon entering round $r+1$. Its ordinary entry condition supplies the
same round-$r$ parents on which that append waits, so the append reaches
$\texttt{BRB-Broadcast}$ and establishes the invariant for the next round.
The induction first brings every correct $R_L$-member that remains active to
$r_\star$ and then proves the claim for every round $r\geq r_\star$.
\end{proof}

\subsubsection{A Common Core in Every Stable Wave}

Following the terminology of DAG-Rider~\cite{dagrider}, a set
$\mathcal K_w$ is a \emph{common core} of a stable wave $w$ if it contains
exactly $q_L$ certified blocks declared under $R_L$ in distinct
round-broadcaster slots of round $\mathsf{first}(w)$ and, for every correct
process $p_i$ and every block $Y\in DAG_i[\mathsf{last}(w)]$,
\[
  Y\leadsto_i B
  \qquad\text{for every }B\in\mathcal K_w.
\]
The next lemma proves that the blocks' fixed parent sets determine such a
common core before valid shares have been released by $f_L+1$ distinct
broadcasters.
The following corollary proves that every correct $R_L$-member that remains
active eventually DAG-delivers it.

\begin{lemma}[Before $f_L+1$ distinct broadcasters release valid coin
shares, fixed parent sets determine a common core]
\label{lemma:stable-wave-reachability}
Assume that every correct broadcaster that produces a singleton-enabled
last-round block releases its prescribed valid coin share in that block.
Also assume that a broadcaster that is uncorrupted when it releases such a
share fixes the block's parents before generating the share, and that those
parents cannot later change. Fix a stable wave $w$ under $R_L$ with
$\mathsf{first}(w)\geq r_\star$ from
\cref{corollary:eventual-stable-suffix}, and write
$r_k=\mathsf{first}(w)+k-1$ for $k\in\{1,2,3,4\}$.

Before valid coin shares for $(R_L,w)$ have been released by $f_L+1$
distinct broadcaster identities, the fixed parent sets determine a set
$\mathcal K_w$ of exactly
$q_L$ certified round-$r_1$ blocks declared under $R_L$ in distinct
round-broadcaster slots. Moreover, for every correct process $p_i$ and every
$Y\in DAG_i[r_4]$,
\[
  \mathcal K_w\subseteq DAG_i[r_1]
  \qquad\text{and}\qquad
  Y\leadsto_i B\quad\text{for every }B\in\mathcal K_w.
\]
Thus $\mathcal K_w$ is a common core of $w$.
\end{lemma}

\begin{proof}
By \cref{corollary:eventual-stable-suffix}, every certified block in
rounds $r_1,r_2,r_3,r_4$ is declared under $R_L$. Correct voters in
rounds $r_2,r_3,r_4$ use the ordinary parent-count rule, since their
enabled-configs value is $\{R_L\}$.
Consequently, every block in rounds $r_2,r_3,r_4$ that is DAG-delivered
by a correct process has exactly $q_L$ parents in distinct
round-broadcaster slots of the preceding round, with broadcaster
identities in $\mathsf{members}(R_L)$.

By \cref{corollary:stable-suffix-round-progress} and the assumption that
every correct broadcaster producing a singleton-enabled last-round block
releases its prescribed valid coin share in that block, at least
$q_L\geq f_L+1$ active correct broadcasters eventually release such shares.
Consider all valid-share release events for $(R_L,w)$, including events at
Byzantine broadcasters, in the order chosen by the adversarial scheduler.
Among the first $f_L+1$ distinct broadcasters,
at least one is uncorrupted when its release occurs. Let $X$ be the
round-$r_4$ block for the first such release. Because $X$'s broadcaster is
uncorrupted when it releases the share, the second assumption in the lemma
statement ensures that $X$'s $q_L$ round-$r_3$ parents are already fixed
before its share is generated and released. Immediately before that release,
fewer than
$f_L+1$ distinct valid shares have been released. Denote $X$'s fixed parent
set by $\mathcal P_X$. Each block in $\mathcal P_X$ is certified, so its
metadata, including its parents, cannot later change; the same holds
recursively for its certified ancestors.

There are $q_L^2$ pairs $(Z,H')$ such that $Z\in\mathcal P_X$ and $H'$ is a
round-$r_2$ parent of $Z$. These pairs involve at most $n_L$ round-$r_2$
broadcaster slots. Therefore, some round-$r_2$ slot contains a parent of at
least $f_L+1$ blocks in $\mathcal P_X$, because
\[
  \frac{q_L^2}{n_L}
  =n_L-2f_L+\frac{f_L^2}{n_L}
  \geq f_L+1,
\]
where the final inequality follows from $n_L\geq3f_L+1$. Among the
qualifying slots, choose the one whose broadcaster appears first in the
ordered list $\Pi_L$, and let $H$ be its unique certified block. This choice
depends only on the already fixed set $\mathcal P_X$. Let
\[
  \mathcal W
  =\{Z\in\mathcal P_X: H\text{ is a parent of }Z\}.
\]
Uniqueness follows from BRB Agreement for every round-broadcaster slot in
\cref{lemma:brb-safety}, and $|\mathcal W|\geq f_L+1$. Define
$\mathcal K_w$ to be the $q_L$
round-$r_1$ parents of $H$. Thus $\mathcal K_w$ is fixed before the
$(f_L+1)$-st distinct valid share for $(R_L,w)$ is released.

Now consider any $Y\in DAG_i[r_4]$ at a correct process $p_i$. Its $q_L$
round-$r_3$ parent slots intersect the at least $f_L+1$ slots occupied by
$\mathcal W$, since
\[
  q_L+|\mathcal W|
  \geq q_L+(f_L+1)=n_L+1,
\]
and each set uses distinct broadcaster slots. Let $Z\in\mathcal W$ occupy
a slot that also occurs among the parent slots of $Y$. The parent of $Y$ in
that slot and $Z$ are in the same round-broadcaster slot, so BRB Agreement
implies that they are the same block. Since $H$ is a parent of $Z$, for
every $B\in\mathcal K_w$,
\[
  Y\leadsto_i Z\leadsto_i H\leadsto_i B.
\]
Causal Closure places $Z$, $H$, and every block in $\mathcal K_w$ in
$DAG_i$, completing the proof.
\end{proof}


\begin{corollary}[Every correct $R_L$-member that remains active eventually
DAG-delivers the same common core and at least $q_L$ distinct leader-voting
blocks that reach every block in it]
\label{corollary:stable-wave-availability}
Fix a stable wave $w$ satisfying the premises of
\cref{lemma:stable-wave-reachability}, let
$r_4=\mathsf{last}(w)$, and let $\mathcal K_w$ be the common core determined
by that lemma.

Every correct process
$p_i\in\mathsf{members}(R_L)$ that remains active eventually DAG-delivers
every block in $\mathcal K_w$ and at least $q_L$ distinct $R_L$-typed
leader-voting blocks $Y$ such that
\[
  Y\leadsto_i B
  \qquad\text{for every }B\in\mathcal K_w.
\]
Thus the same common core $\mathcal K_w$ is eventually present in the DAG of
every correct $R_L$-member that remains active.
\end{corollary}

\begin{proof}
Choose any $q_L$ correct $R_L$-members that remain active; their existence is
guaranteed by the assumption in \cref{sec:problem-definition} that at least
$q_L$ correct members of $R_L$ remain active.
\Cref{corollary:stable-suffix-round-progress} brings them to round $r_4$.
Upon entering that round, each selected process produces its prescribed
block. Here $r_4=\mathsf{last}(w)$, and stability gives
$\mathsf{local}(\mu)=R_L$ and $\mathsf{enabled}(\mu)=\{R_L\}$. The first
assumption of \cref{lemma:stable-wave-reachability} says that every correct
broadcaster producing such a block releases its prescribed valid
$R_L$-scoped coin share in that block. Thus the resulting $q_L$ blocks are
$R_L$-typed leader-voting blocks in distinct slots. By
\cref{lemma:known-coverage}, for every selected broadcaster $p_k$, every
correct $R_L$-member belongs to $\mathit{known}_k$. The required broadcasters,
processes that send BRB echoes and votes, and recipients remain active and
continue the prescribed tasks in the stable suffix. For each application,
take $\mathcal G$ to be these selected $q_L$ broadcasters. Therefore
BRB Agreement through round $r_4$ follows from \cref{lemma:brb-safety}.
Then \cref{lemma:correct-block-dissemination}, applied to each selected
broadcaster, makes every correct $R_L$-member that remains active eventually
DAG-deliver all $q_L$ blocks. The conclusion now
follows from \cref{lemma:stable-wave-reachability} and Causal Closure.
\end{proof}

\subsubsection{Direct-Anchor Recurrence and Progress}

\paragraph*{Conditional Probability Requirement.}
For sufficiently large security parameter $\lambda$, put
\[
  \varepsilon_L
  =\frac{q_L}{n_L}-\operatorname{negl}(\lambda)>0.
\]
Since $q_L=n_L-f_L$ and $n_L\geq3f_L+1$, this is at least
$2/3-\operatorname{negl}(\lambda)$.

To prove Probabilistic Progress with probability one, we make the following
additional requirement. For every stable wave $w$ satisfying the premises of
\cref{lemma:stable-wave-reachability}, let $\mathcal K_w$ be the common core
determined by that lemma. Consider any finite execution prefix at whose end
$\mathcal K_w$ has been fixed but fewer than $f_L+1$ distinct valid shares
for $(R_L,w)$ have been released. For every adversary view occurring with
positive probability at the end of such a prefix, conditioned on that view,
the probability that the coin selects the broadcaster of a block in
$\mathcal K_w$ is at least $\varepsilon_L$.

The construction in \cref{lemma:stable-wave-reachability} determines
$\mathcal K_w$ from parent sets already fixed at that point. The adversary's
view includes all earlier scheduling decisions and coin outcomes. Thus the
same positive lower bound holds after every possible prior history,
including any pattern of failures in earlier stable waves; coin outcomes in
different waves need not be independent.

For a concrete computational implementation, Unpredictability instead
bounds the overall failure probability across any polynomially bounded
number of waves, up to negligible error. It does not by itself imply this
conditional bound after every finite execution prefix.

\begin{proposition}[With probability at least $\varepsilon_L$, every correct
$R_L$-member that remains active identifies a direct anchor in wave $w$ or
later]
\label{proposition:stable-wave-direct-anchor}
Let $w$ be a stable wave satisfying the premises of
\cref{lemma:stable-wave-reachability}, and let $\mathcal K_w$ be the common
core determined by that lemma. Suppose that, conditioned on the adversary's
current view after any finite execution prefix ending after $\mathcal K_w$ is fixed
and before the $(f_L+1)$-st distinct valid share for $(R_L,w)$ is released,
the coin selects the broadcaster of a block in $\mathcal K_w$ with
probability at least $\varepsilon_L$. Then, after every such prefix, with
conditional probability at least $\varepsilon_L$, every correct process
$p_i\in\mathsf{members}(R_L)$ that remains active eventually identifies a
direct anchor in wave $w$ or in a later wave.
\end{proposition}

\begin{proof}
By \cref{lemma:stable-wave-reachability}, $\mathcal K_w$ occupies exactly
$q_L$ of the $n_L$ round-broadcaster slots and is fixed before the
$(f_L+1)$-st distinct valid share for $(R_L,w)$ is released. By the
probability bound in the proposition statement, the coin selects the
broadcaster of some $B\in\mathcal K_w$ with probability at least
$\varepsilon_L$.

On this event, \cref{corollary:stable-wave-availability} makes every correct
$R_L$-member $p_i$ that remains active DAG-deliver $B$ and at least
$q_L\geq f_L+1$ leader-voting blocks carrying valid shares and reaching
$B$. From any $f_L+1$ of these valid shares, $p_i$ reconstructs the coin for
$(R_L,w)$. Its output is $B$'s broadcaster, and $p_i$ has DAG-delivered $B$,
so $\texttt{WaveLeader}_i(R_L,w)$ returns $B$. When the last of these $q_L$
leader-voting blocks arrives, the protocol rechecks $w$ and identifies $B$
as its direct anchor, unless $p_i$ has already identified an anchor in a
later wave. Either case gives the conclusion.
\end{proof}

\begin{lemma}[Every correct $R_L$-member that remains active identifies
direct anchors in arbitrarily late stable waves with probability one]
\label{lemma:infinitely-many-direct-anchors}
Let $R_L$ and $r_\star$ be as in
\cref{corollary:eventual-stable-suffix}. Assume that the premises of
\cref{lemma:stable-wave-reachability} hold for every stable wave $w$ under
$R_L$ with $\mathsf{first}(w)\geq r_\star$. For each such wave, let
$\mathcal K_w$ be the common core determined by that lemma. Also assume that,
for every such wave, conditioned on the adversary's current view after any
finite execution prefix ending after $\mathcal K_w$ is fixed and before the
$(f_L+1)$-st distinct valid share for $(R_L,w)$ is released, the coin selects
the broadcaster of a block in $\mathcal K_w$ with probability at least
$\varepsilon_L$. Then,
over the configuration-scoped perfect coins, with probability one, for every
correct process
$p_i\in\mathsf{members}(R_L)$ that remains active and every wave index
$b\geq1$, there exists a stable wave $h\geq b$ in which $p_i$ identifies a
direct anchor.
\end{lemma}

\begin{proof}
By \cref{corollary:stable-suffix-round-progress}, every correct
$R_L$-member that remains active advances through all rounds of infinitely
many stable waves. In every such wave $w$, at least $q_L\geq f_L+1$ correct
broadcasters release valid shares, which are sufficient to reconstruct the
coin for $(R_L,w)$.
Fix any starting wave index $b$. When each common core is determined,
the adversary's view includes all preceding coin outcomes. Averaging the
conditional lower bound over views in which the coins in earlier stable
waves with indices at least $b$ did not select a broadcaster from their
respective common cores shows that the next coin selects such a broadcaster
with conditional probability at least $\varepsilon_L>0$.
Thus the conditional-probability chain rule gives, for every $N\geq1$,
\[
  \Pr\!\left[
    \begin{array}{c}
      \text{in none of the first $N$ stable waves $h\geq b$ does the coin}\\
      \text{for $(R_L,h)$ select the broadcaster of a block in $\mathcal K_h$}
    \end{array}
  \right]
  \leq(1-\varepsilon_L)^N.
\]
As $N$ tends to infinity, the right-hand side tends to zero. Hence the
probability that no stable wave $h\geq b$ has this property is zero.
Taking the countable union over starting wave indices $b$ shows that,
with probability one, the coin selects a broadcaster from the common core
in arbitrarily late stable waves. Whenever this occurs in a wave $h$,
\cref{proposition:stable-wave-direct-anchor} shows that every correct
$R_L$-member that remains active identifies a direct anchor in wave $h$
or a later wave.
\end{proof}




\begin{theorem}[Probabilistic Progress]
\label{theorem:leader-probabilistic-progress}
Let $R_L$ and $r_\star$ be as in
\cref{corollary:eventual-stable-suffix}, and suppose that the two premises
of \cref{lemma:stable-wave-reachability} hold in every stable wave $w$ with
$\mathsf{first}(w)\geq r_\star$. Also suppose that the conditional
probability premise of \cref{lemma:infinitely-many-direct-anchors} holds for
every such wave. With probability one, every correct $R_L$-member that
remains active finalizes infinitely many blocks.
\end{theorem}

\begin{proof}
Fix a correct $R_L$-member $p_i$ that remains active and any wave index $b$.
By \cref{lemma:infinitely-many-direct-anchors}, $p_i$ eventually identifies
a direct anchor in some wave $h\geq b$. Its serialized
$\texttt{Elect}_i$ invocation either changes its finalized leader sequence to
$\Phi_h$, which contains that direct anchor, or an invocation for some later
wave $s\geq h$ has already changed the sequence to $\Phi_s$. In the latter
case, $\Phi_h\preceq\Phi_s$ by
\cref{lemma:nested-leader-scans}, so the sequence still contains the
wave-$h$ anchor. Because $b$ is arbitrary and direct anchors in distinct
waves are distinct blocks, the lengths of $p_i$'s finalized leader sequences
are unbounded. Hence $p_i$ finalizes infinitely many blocks with probability
one.
\end{proof}

\subsection{Dynamic DAG Byzantine Atomic Broadcast}


The probability-one Probabilistic Progress conclusion in the following
restatement uses the conditional probability bound in
\cref{sec:probabilistic-progress}. The deterministic Integrity and Total
Order conclusions do not use the Unpredictability or Fairness properties of
the configuration-scoped perfect coin.
\dynamicdagbabcorrectness*

\begin{proof}
We verify each component of the abstraction:

\emph{Underlying Dynamic DAG-structured Byzantine Broadcast.}
By \cref{lemma:dynamic-dag-structured-broadcast}, the BRB, DAG,
state-transfer, and reconfiguration protocols satisfy
\cref{def:dynamic-dag-structured-broadcast}, including completion of
every valid transition.

\emph{Integrity.}
By \cref{corollary:leader-integrity}, a correct process finalizes at
most one leader from each wave, never inserts the same block twice in
its finalized leader sequence, and finalizes only blocks already inserted in
its local DAG. These are exactly the Integrity clauses of
\cref{def:dynamic-dag-bab}.

\emph{Total Order.}
The canonical leader prefixes are nested by
\cref{lemma:nested-leader-scans,lemma:canonical-local-prefixes}; hence
\cref{theorem:leader-total-order} gives Total Order throughout every
transition, independently of eventual stability.

\emph{Probabilistic Progress.}
By \cref{corollary:eventual-stable-suffix,%
corollary:stable-suffix-round-progress}, the execution eventually has
infinitely many stable waves under $R_L$. As required by
\cref{sec:our-solution}, every correct broadcaster of a singleton-enabled
last-round block releases its prescribed valid coin share, and every
broadcaster that is uncorrupted when releasing such a share fixes the
block's parents before generating and releasing the share. These are
precisely the two premises of \cref{lemma:stable-wave-reachability}.
The conditional probability bound assumed by the theorem supplies the
remaining premise of \cref{lemma:infinitely-many-direct-anchors}. Hence
\cref{theorem:leader-probabilistic-progress} gives Probabilistic Progress as
defined in \cref{def:dynamic-dag-bab}.
\end{proof}